\documentclass[11pt]{article}
\usepackage{pgfplots}
\usepackage{geometry}
\usepackage{times,xspace}
\usepackage{multirow}

\usepackage{graphicx}
\usepackage{amsmath, amsthm, amssymb}
\usepackage[shortlabels,inline]{enumitem}
\usepackage{comment}
\usepackage{array,booktabs,longtable}
\usepackage{csvsimple}

\usepackage{microtype}

\usepackage[ruled, linesnumbered, vlined]{algorithm2e}

\newif\ifdvi
\dvitrue
\ifdvi
\else
\usepackage[backref=page]{hyperref}
\hypersetup{
	bookmarks=true,
	bookmarksnumbered=true,
	bookmarksopen=true
}
\fi

\newcommand{\np}{{\em NP}\xspace} 
\newcommand{\nphard}{\np-hard\xspace}

\newtheorem{thm}{Theorem}[section]
\newtheorem{theorem}[thm]{Theorem}

\newtheorem{lemma}[thm]{Lemma}

{\theoremstyle{definition} 

}

{\theoremstyle{remark}  \newtheorem{remark}{Remark}}

\newenvironment{proofof}[1]{\begin{proof}[Proof of #1]}{\end{proof}}

\DeclareMathOperator{\argmin}{argmin}

\newcommand{\sm}{\ensuremath{\setminus}}
\newcommand{\es}{\ensuremath{\emptyset}}
\newcommand{\sse}{\subseteq}

\newcommand{\R}{\ensuremath{\mathbb R}}

\newcommand{\C}{\ensuremath{\mathcal{C}}}

\newcommand{\T}{\ensuremath{\mathcal T}}

\newcommand{\e}{\ensuremath{\epsilon}}
\newcommand{\ve}{\ensuremath{\varepsilon}}
\newcommand{\gm}{\ensuremath{\gamma}}

\newcommand{\ld}{\ensuremath{\lambda}}
\newcommand{\Ld}{\ensuremath{\Lambda}}

\newcommand{\al}{\ensuremath{\alpha}}
\newcommand{\tht}{\ensuremath{\theta}}

\newcommand{\dt}{\ensuremath{\delta}}

\newcommand{\into}{\ensuremath{\mathrm{in}}}
\newcommand{\out}{\ensuremath{\mathrm{out}}}

\newcommand{\ceil}[1]{\ensuremath{\left\lceil#1\right\rceil}}

\newcommand{\poly}{\operatorname{\mathsf{poly}}}

\newcommand{\OPT}{\ensuremath{\mathit{OPT}}}
\newcommand{\opt}{\ensuremath{\mathsf{Opt}}}
\newcommand{\val}{\ensuremath{\mathsf{Val}}}
\newcommand{\iopt}{\ensuremath{O^*}}
\newcommand{\pcval}{\ensuremath{\mathsf{PCC}}}
\newcommand{\pcw}{PCW\xspace}
\newcommand{\lb}{\ensuremath{\mathsf{LB}}}
\newcommand{\ub}{\ensuremath{\mathsf{UB}}}
\newcommand{\UBrooted}{\ub{}1}
\newcommand{\UBptprw}{\ub{}2}
\newcommand{\UBptpwt}{\ub{}3}
\newcommand{\UBcychalfb}{\ub{}4}
\newcommand{\UBcycbetter}{\ub{}5}
\newcommand{\UBptp}{\ensuremath{\ub\text{-}\mathsf{P2P}}\xspace}
\newcommand{\UBcyc}{\ensuremath{\ub\text{-}\mathsf{Cyc}}\xspace}
\newcommand{\bc}{\ensuremath{\overline{c}}}
\newcommand{\bpi}{\ensuremath{\overline{\pi}}}
\newcommand{\bG}{\ensuremath{\overline{G}}}
\newcommand{\bV}{\ensuremath{\overline{V}}}
\newcommand{\bS}{\ensuremath{\overline{S}}}
\newcommand{\bT}{\ensuremath{\overline{T}}}
\newcommand{\bA}{\ensuremath{\overline{A}}}
\newcommand{\bx}{\ensuremath{\overline{x}}}
\newcommand{\by}{\ensuremath{\overline{y}}}
\newcommand{\bz}{\ensuremath{\overline{z}}}
\newcommand{\bp}{\ensuremath{\overline{p}}}
\newcommand{\tc}{\ensuremath{\widetilde{c}}}
\newcommand{\tD}{\ensuremath{\widetilde{D}}}
\newcommand{\tE}{\ensuremath{\widetilde{E}}}
\newcommand{\tpi}{\ensuremath{\widetilde{\pi}}}

\newcommand{\ty}{\ensuremath{\widetilde{y}}}

\newcommand{\met}{\ensuremath{\mathsf{met}}}

\newcommand{\lpname}[1][]{\ensuremath{(\ref{eq:primal}_{{#1}})}\xspace}
\newcommand{\lpopt}{\OPT}
\newcommand{\dualub}{\ensuremath{Y}}
\newcommand{\creg}{\ensuremath{c^{\mathsf{reg}}}}
\newcommand{\binlb}{\ensuremath{\mathsf{low}}}
\newcommand{\binub}{\ensuremath{\mathsf{high}}}
\newcommand{\budg}{\ensuremath{L}}

\newcommand{\rorientlp}[1][]{(\text{RO-P}_{#1})\xspace}
\newcommand{\rorientlpopt}[1][]{\ensuremath{\OPT_{\rorientlp[#1]}}}

\newcommand{\vrp}{\ensuremath{\mathsf{VRP}}\xspace}
\newcommand{\mlp}{\ensuremath{\mathsf{MLP}}\xspace}
\newcommand{\kmlp}{\ensuremath{k\text{-}\mlp}\xspace}
\newcommand{\ptp}{\ensuremath{\mathsf{P2P}}\xspace}
\newcommand{\cyc}{C}
\newcommand{\rewd}{\ensuremath{\Pi}}
\newcommand{\ropt}{\ensuremath{\rewd^*}}
\newcommand{\roptptp}{\ensuremath{\rewd^{*\ptp}}}
\newcommand{\roptcyc}{\ensuremath{\rewd^{*\mathsf{Cyc}}}}
\newcommand{\regret}{\ensuremath{R}}
\newcommand{\opttsp}{\ensuremath{\mathsf{TSP}\opt}\xspace}
\newcommand{\bUBptp}{\ensuremath{\overline{\UBptp}}\xspace}

\SetKwFunction{IterPCA}{IterPCA}
\SetAlgoFuncName{Algorithm}{Algorithm}
\SetFuncSty{textsc}
\SetKwFunction{BinSearchPCA}{BinSearchPCA}

\title{Combinatorial Algorithms for Rooted Prize-Collecting Walks and Applications to
  Orienteering and Minimum-Latency Problems} 
\author{Sina Dezfuli\thanks{{\tt dezfuli@ualberta.ca}.  
    Dept. of Computer Science, Univ. Alberta, Edmonton, AB T6G 2E8.} 
\and 
    Zachary Friggstad\thanks{{\tt zacharyf@ualberta.ca}.  
    Dept. of Computer Science, Univ. Alberta, Edmonton, AB T6G 2E8.
    Supported by an NSERC Discovery grant and an NSERC Discovery Accelerator Supplement Award.}
\and 
    Ian Post\thanks{{\tt ian@ianpost.org}}
\and
    Chaitanya Swamy\thanks{{\tt cswamy@uwaterloo.ca}.  
    Dept. of Combinatorics and Optimization, Univ. Waterloo, Waterloo, ON N2L 3G1. 
    Supported in part by NSERC grant 327620-09 and an NSERC Discovery Accelerator
    Supplement Award.}
}
\date{}

\begin{document}

\maketitle
\def\thepage{}
\thispagestyle{empty}

\begin{abstract}
%
We consider the rooted {\em prize-collecting walks} (\pcw)
problem, wherein we seek a collection $\C$ of 
rooted walks having minimum {\em prize-collecting cost}, which is the (total cost of 
walks in $\C$) + (total node-reward of the nodes not visited by any walk in $\C$).
This problem arises naturally as the Lagrangian relaxation of
both {\em orienteering}, 
where we seek a length-bounded walk of maximum reward,
and 
the {\em $\ell$-stroll problem}, where we seek a minimum-length walk covering at least
$\ell$ nodes. 
Our main contribution is to devise a {\em simple, combinatorial algorithm}
for the \pcw problem that 
returns a rooted tree whose prize-collecting cost is at most the {\em optimum value} 
of the prize-collecting walks problem. 
This result applies to both directed and
undirected graphs, and holds for arbitrary nonnegative edge costs.

We present two applications of our result, where we utilize our algorithm to develop
combinatorial approximation algorithms for two fundamental vehicle-routing
problems (\vrp{}s): (1) orienteering; 
and (2) {\em $k$-minimum-latency problem} (\kmlp), wherein we seek to cover all nodes
using $k$ paths starting at a prescribed root node, so as to minimize the sum of the node
visiting times.
Our combinatorial algorithm allows us to sidestep the part where we 
solve a preflow-based LP in the LP-rounding algorithms of~\cite{FriggstadS17}
for orienteering, and in the state-of-the-art $7.183$-approximation algorithm for
\kmlp in~\cite{PostS15}. Consequently, we obtain combinatorial implementations of these
algorithms (with the same approximation factors).
Compared to algorithms that achieve the current-best approximation factors for
orienteering and \kmlp, our algorithms have substantially improved running time, and 
achieve approximation guarantees that match (\kmlp), or are slightly worse
(orienteering) than the current-best approximation factors for these problems.

We report various computational results for our resulting (combinatorial
implementations of) orienteering algorithms, 
which show that the algorithms perform quite well in practice, both in
terms of the quality of the solution they return, 
as also the upper bound they yield on the orienteering optimum (which is
obtained by leveraging the workings of our \pcw algorithm).
\end{abstract}

\newpage
\pagenumbering{arabic} \normalsize

\section{Introduction} \label{intro}
Vehicle-routing problems (\vrp{}s) are a rich class of optimization problems that find
various applications, and have been extensively studied in the 
Operations Research and Computer Science literature (see, 
e.g.,~\cite{TothV02}.) 
Broadly speaking, we can distinguish between two types of vehicle-routing problems: one
where resource constraints require us to select which set of nodes or clients to visit 
{\em and} plan a suitable route(s) for visiting these clients; and the other, where we
have a {\em fixed set} of clients, and seek the most effective route(s) for visiting
these clients.

We consider two prominent and well-motivated problems in these two categories: 
(1) {\em orienteering}~\cite{BlumCKLMM07,BansalBCM04,ChekuriKP12,FriggstadS17}, belonging 
to the first category, wherein nodes have associated rewards for visiting them, 
and we seek a length-bounded path that collects maximum reward; and 
(2) {\em minimum-latency problems} (\mlp{}s)~\cite{BC+94,ChaudhuriGRT03,PostS15},
belonging to the second category, wherein, we 
seek one or more rooted paths to visit a given set of clients so as to minimize the sum of
the client visiting times (i.e., the total latency).
Besides its appeal as a natural and clean way of capturing resource constraints in a \vrp, 
the fundamental nature of orienteering stems from the fact that 
it often naturally arises as a subroutine when solving other \vrp{}s, both in
approximation algorithms---e.g., for minimum-latency problems
(see~\cite{FakcharoenpholHR07,ChakrabartyS16,PostS15}),  
TSP with time windows~\cite{BansalBCM04}, 
\vrp{}s distance bounds~\cite{NagarajanR12} and regret bounds~\cite{FriggstadS14}, as
also in computational methods where orienteering 
corresponds to the ``pricing'' problem encountered in solving set covering/partitioning
LPs (a.k.a configuration LPs) for \vrp{}s via a column-generation or branch-cut-and-price 
method (see, e.g.,~\cite{Dezfuly19}). 
In particular, in various settings (including \mlp{}s, \vrp with distance- and regret-
bounds), we can formulate the \vrp as the problem of covering a set of clients using
suitable paths, and solving this covering problem, approximately via a set-cover approach,
or its corresponding configuration-LP relaxation, then entails solving an orienteering
problem. 

Some recent work on orienteering~\cite{FriggstadS17} and \mlp{}s~\cite{PostS15}, has
led to promising LP-based approaches for tackling these problems, yielding, for
multi-vehicle \mlp{}s, the current-best approximation factors. This approach is based on
moving to a bidirected version of the underlying metric and considering a preflow-based
LP-relaxation for rooted walk(s), 
and using a powerful arborescence-packing result of Bang-Jensen et
al.~\cite{BangjensenFJ95} to decompose an (optimal) LP solution into a convex combination of
arborescences that is ``at least as good'' as the LP solution. Viewing these arborescences
as rooted trees in the undirected graph, one can convert the tree into a rooted
path/cycle by doubling and shortcutting, and the above works show how to leverage the
resulting convex combination of paths/cycles to extract a good solution.

\vspace*{-1ex}
\paragraph{Our contributions and related work.}
We study the {\em prize-collecting walks} (\pcw) problem, which is the problem of
finding a collection $\C$ of $r$-rooted walks in a digraph $G=(V,E)$ with nonnegative edge 
costs 
and node rewards, 
having minimum  
{\em prize-collecting cost}, which is the total cost of the 
walks in $\C$ + the total node-reward of the nodes not visited by any walk in $\C$.
This problem arises as the Lagrangian relaxation of orienteering, and 
a subroutine encountered in \mlp{} algorithms, namely that of finding a rooted path 
of minimum cost covering a certain number of nodes.

Our main contribution is to devise a {\em simple, combinatorial algorithm} for the \pcw
problem that returns a {\em directed tree} (more precisely, an out-arborescence)
rooted at $r$ {\em whose prize-collecting cost is at most the optimal value of the \pcw
problem}. 
To state our result a bit more precisely, we introduce some notation.
Let $G = (V,A)$ be a directed graph with arc-set $A$, arc lengths $c_a\geq 0$ for all 
$a\in A$, and root $r$. Let each node $v \in V$ have a reward or penalty $\pi_v\geq 0$. 
For a multiset of arcs $T$, define 
$c(T) = \sum_{a \in A} c_a\cdot(\text{number of occurrences of $a$ in $T$})$.
Define $\pi(S) = \sum_{v \in S} \pi_v$ for any set of nodes $S$. 
An {\em out-arborescence rooted at $r$} is a subgraph $T$ whose undirected
version is a tree containing $r$, and where every node spanned by $T$ except $r$ has
exactly one incoming arc in $T$; 
we will often abbreviate this to an out-arborescence. 
For any subgraph $T$ of $G$ where all nodes in $V(T)$ are reachable from $r$ in $T$ 
(such as an out-arborescence rooted at $r$), define the {\em prize-collecting cost} of $T$
to be $\pcval(T):=c(T)+\pi(V\sm V(T))$.

We give a combinatorial polynomial-time algorithm \IterPCA (see Section~\ref{comb}), that
finds an out-arborescence $T$ whose prize-collecting cost is at most the prize-collecting
cost of any collection of $r$-rooted walks, i.e.,   
\[
c(T) + \pi(V\setminus V(T)) \le \iopt:=
\min_{\begin{subarray}{l}\text{collections $\C$ of} \\ \text{$r$-rooted walks} \end{subarray}}
\ \Bigl[\sum_{P\in\C} c(P) + \pi\Bigl(V\setminus\bigcup_{P\in\C}V(P)\Bigr)\Bigr] \;.
\]
We actually obtain the stronger guarantee that $\pcval(T)$ is at most the optimal value 
$\OPT$ of a preflow-based LP-relaxation of \pcw \eqref{eq:primal}.

We briefly discuss the ideas underlying our combinatorial algorithm \IterPCA. Our
algorithm and analysis is quite simple, and resembles Edmonds' algorithm for finding a
minimum-cost arborescence. It is based on three main ideas for iteratively simplifying the
instance. 

We observe that if we modify the instance by picking any non-root node $v$, and
subtracting a common value $\tht$ from the cost of all incoming arcs of $v$ and
from $\pi_v$, while ensuring that the new values of these quantities is nonnegative, then
it suffices to prove the desired guarantee for the modified instance. 
Next, by choosing a suitable $\tht_v$ for all all non-root nodes, and modifying costs and
rewards as above, we may assume that in the modified instance, either: 
(a) there is a node $v\in V'$ with zero reward; 
(b) there is a (directed) cycle $Z$ consisting of zero (modified) cost arcs; or
(c) there is an out-arborescence consisting of zero cost arcs.
If (c) applies, then we are done. If (a) or (b) apply, then we argue that may further
simplify the instance as follows: in case (a), we shortcut past $v$ by merging every pair
of incoming and outgoing arcs of $v$ and deleting $v$; in case (b), we
contract the cycle $Z$ and set the reward of the contracted node to be the sum of the
(modified) rewards of nodes in $Z$.
We then recurse on 
the simplified instance. 
We believe that the above result, and the techniques underlying it, are of independent
interest, and will find various applications. 
We present two applications of our result (see Section~\ref{apps}), where we use our  
combinatorial algorithm for \pcw to give combinatorial implementations of the LP-rounding
algorithms in~\cite{FriggstadS17} and~\cite{PostS15} for orienteering and \kmlp 
respectively. 
We now discuss these applications, and in doing so place our main result in the 
context of some extant work.
We say that $x\in\R_+^A$ is an $r$-preflow (or simply preflow),  
if we have $x\bigl(\dt^\into(v)\bigr)\geq x\bigl(\dt^\out(v)\bigr)$ for all non-root nodes
$v$. 
\begin{enumerate}[label=$\bullet$, topsep=0.5ex, noitemsep, leftmargin=*] 
\item Friggstad and Swamy~\cite{FriggstadS17} proposed a novel LP-based approach for
  orienteering, wherein the LP-relaxation searches for a ($r$-) preflow of large
  reward (see \eqref{rorlp} in Appendix~\ref{append-lpbounds}). 
  The first step (and key insight) in their rounding algorithm is 
  to utilize the arborescence-packing result of~\cite{BangjensenFJ95} to cast
  the LP-solution $x$ as a convex combination of arborescences whose expected reward is
  at least the LP-optimum and whose expected cost is at most the length bound, say $B$.
  They leverage this to show that one can then extract a rooted path having reward at
  least (LP-optimum)/$3$ via a simple combinatorial procedure. 

  We show (see Section~\ref{orient-apps}) that one can utilize our algorithm \IterPCA, in
  conjunction with binary 
  search, to obtain the desired convex combination combinatorially, that is, without
  having to solve their LP-relaxation, and thereby obtain a combinatorial
  $3$-approximation.  
  This follows because the \pcw problem is obtained
  by Lagrangifying the ``length at most $B$'' constraint. A standard fine tuning of
  the Lagrangian variable (which affects the node rewards) via binary search then yields
  the desired distribution (over at most two rooted trees). The same ideas also apply and
  yield combinatorial approximation algorithms for other variants of orienteering, such as
  {\em P2P-orienteering} (where the other end-point of  the path is also specified) and
  {\em cycle orienteering} (where we seek a cycle containing $r$.%
 \footnote{Cycle orienteering is not considered in~\cite{FriggstadS17}, but their ideas
   can be easily adapted.})

  While the approximation factor of $3$ does not as yet beat the
  $(2+\e)$-approximation factor for orienteering~\cite{ChekuriKP12}, our algorithm is
  significantly simpler and faster algorithm than prior dynamic-programming (DP) based 
  algorithms for orienteering~\cite{BlumCKLMM07,BansalBCM04,ChekuriKP12}.%
  \footnote{A straightforward implementation of our combinatorial algorithm for
    orienteering takes 
    $O(n^4\cdot K)$ time, where $K$ is the time for binary search. In contrast, the 
    the algorithm in~\cite{ChekuriKP12} has running time at least 
    $O\bigl(n^{1/\ve^2}\cdot K\bigr)$ for obtaining a $\frac{2}{1-\ve}$-approximation;
    thus, $O(n^9\cdot K)$ time for returning a $3$-approximation. The DP-algorithm of Blum
    et al.~\cite{BlumCKLMM07} has running time at least $O(n^5\cdot K)$, and its
    approximation guarantee is no better than $4$. 
  } 
  Moreover, an added subtle benefit of the algorithms in~\cite{FriggstadS17} is that they
  also yield an upper bound on the optimum, which is useful since it can be used to
  evaluate the approximation factor of the solution computed on a per-instance basis. Our
  combinatorial algorithms inherit this benefit, and also provide an upper bound on the
  orienteering optimum.
  
  Our combinatorial algorithm and the associated upper bound may also find use in the
  context of computational methods for solving other \vrp{}s, since (as mentioned earlier)
  orienteering corresponds to the pricing problem that needs to be solved in these
  contexts. Indeed~\cite{Dezfuly19} utilizes our combinatorial algorithm to obtain
  near-optimal solutions to distance-constrained vehicle routing.

  \smallskip
  In Section~\ref{compres}, we undertake an extensive computational study of our 
  combinatorial orienteering algorithms, in order to better understand the performance of
  our algorithms in practice. Our computational experiments show that our algorithms
  perform fairly well in practice---both in terms of the solution computed, and the upper
  bound computed---and much better than that indicated by the theoretical analysis. 


\item Post and Swamy~\cite{PostS15} consider multi-vehicle \mlp{}s. For \kmlp, 
  wherein we seek $k$ rooted paths of minimum total latency that together visit all nodes,
  they devise two $7.183$-approximation algorithms. 
  One of their algorithms (Algorithm 3 in \S6.2~\cite{PostS15}) utilizes a subroutine for
  computing a distribution of rooted trees covering at least $k$ nodes in expectation,
  whose expected cost is at most that of any collection of rooted walks that together
  cover at least $k$ nodes. 
  Lagrangifying the coverage constraint again yields a \pcw problem. Post and
  Swamy~\cite{PostS15} devised an LP-rounding algorithm for this problem, by considering
  its LP-relaxation \eqref{eq:primal}, using arborescence packing to obtain a rooted tree
  with $\pcval(T)$ at most the LP-optimum $\OPT$, and then fine-tuning the node
  rewards via binary search to obtain the desired distribution.
  In particular, for the \pcw problem, they obtain the same guarantee that we do,
  but via solving the LP \eqref{eq:primal}. While not a
  combinatorial algorithm, they dub their resulting \kmlp algorithm a ``more
  combinatorial'' algorithm (as opposed to their other $7.183$-approximation algorithm,
  which needs to explicitly solve a configuration LP).

  We can instead utilize {\em our combinatorial algorithm} to produce the rooted tree $T$ 
  (see Section~\ref{kmlp-apps}); incorporating this within the ``more combinatorial''
  algorithm of~\cite{PostS15} yields a fully and truly combinatorial $7.183$-approximation
  algorithm for \kmlp, which is the state-of-the-art for this problem.

We remark that our result bounding the prize-collecting
cost of the tree $T$ by the prize-collecting cost of {\em any} collection of rooted walks
is a substantial generalization of an analogous result in~\cite{ChaudhuriGRT03}, who
compare against the prize-collecting cost of a {\em single} walk (and specifically in
undirected graphs). As noted in~\cite{PostS15}, the stronger guarantee where we
compare against multiple walks is essential for obtaining guarantees for \kmlp.
\end{enumerate}

\section{LP-relaxation for the prize-collecting-walks problem} \label{lps}
Recall that we are given a directed graph $G = (V,A)$, arc costs $c_a\geq 0$ for all 
$a\in A$, root node $r\in V$, and a reward or penalty $\pi_v\geq 0$ for each node $v$.
(Note that $\pi_r$ is inconsequential, as it does not affect the prize-collecting cost of
any rooted object (out-arborescence, walk); so it will sometimes be convenient
notationally to assume that $\pi_r=0$.)  

Our LP-relaxation \eqref{eq:primal} for prize-collecting walks has
a variable $x_a$ for each arc $a$, which represents the multiplicity of arc $a$ in the
walk-collection, and a variable $p_v$ for each node $v\neq r$, which indicates whether
node $v$ is not covered.
\begin{alignat}{3}
\min & \quad & \sum_{a \in A} c_ax_a + \sum_{v\in V} & \pi_vp_v 
\tag{P} \label{eq:primal} \\
\text{s.t.} && 
x\bigl(\dt^\into(S)\bigr)+p_v & \ge 1 \qquad 
&& \forall S \subseteq V\setminus\{r\}, v\in S \label{eq:coverage} \\
&& 
x\bigl(\dt^\into(v)\bigr) & \ge 
x\bigl(\dt^\out(v)\bigr) \qquad && \forall v\in V\sm\{r\} \label{eq:degree} \\
&& x,p & \ge 0. \notag 
\end{alignat}
Constraint \eqref{eq:coverage} encodes that for every set $S$ of nodes $S$ not containing
the root, and $v\in S$, either $S$ has an incoming arc or we pay the penalty $\pi_v$ for
not visiting $v$. Constraint \eqref{eq:degree} encodes that the in-degree of every node
other than the root is always at least its out-degree,  
so that the solution corresponds to a collection of walks rather than a tree. 
(Note that while we have included the variable $p_r$ above, it does not appear in any
constraint, so we may assume that $p_r=0$ in any feasible solution to \eqref{eq:primal}.)

\section{A combinatorial algorithm} \label{comb}
We now present a combinatorial algorithm for prize-collecting walks based on iteratively
simplifying the instance. 
Recall that $\iopt$ is the minimum value of 
$\bigl[\sum_{P\in\C} c(P) + \pi(V\setminus\bigcup_{P\in\C}V(P))\bigr]$ over all
collections $\C$ of $r$-rooted walks.  
(Recall that a walk may have repeated nodes and arcs, and 
$c(T) = \sum_{a \in A} c_a\cdot(\text{number of occurrences of $a$ in $T$})$
for a multiset of arcs $T$.)
Throughout this section, the root will remain $r$, so will frequently drop $r$ from the
notation used to refer to an instance.
Since we will modify the instance $(G,c,\pi)$ during the course of our algorithm (but not
change the root), we use  
$\iopt(G,c,\pi)$ to denote the above quantity. Also, we use $\lpname[(G,c,\pi)]$ to refer
to the LP-relaxation \eqref{eq:primal} for the instance $(G,c,\pi)$, and $\lpopt(G,c,\pi)$
to denote its optimal value. 
We use 
\mbox{$\pcval(T;G,c,\pi):=c(T)+\pi(V\sm V(T))$} to denote the prize-collecting value of
$T$ under 
arc costs $c$ and penalties $\pi$, where $T$ is a subgraph of $G$ such that all nodes in
$V(T)$ are reachable from $r$ in $T$. Whenever we say optimal solution below, we mean the
optimal walk-collection (i.e., an optimal integral solution to \eqref{eq:primal}). 

Our algorithm \IterPCA resembles Edmond's algorithm for finding a minimum-cost
arborescence, and is 
based on three main ideas for simplifying the instance. However, unlike in the case of
min-cost spanning arborescences, our simplifications do {\em not} leave the problem
unchanged; we really exploit the asymmetry that we seek an out-arborescence but are
comparing its value against the best collection of $r$-rooted 
walks in $(G,c,\pi)$. 

Let $V'=V\sm\{r\}$. 
We observe that we may modify the instance by picking a node $v\in V'$, and
subtracting a common value $\tht$ from the cost of all incoming arcs of $v$ and
from $\pi_v$, while ensuring that the new values of these quantities is nonnegative (see
step~\eqref{modinst}). That is, it suffices to prove the desired guarantee for the
modified instance $(G,\tc,\tpi)$: if $T$ is an out-arborescence with 
$\pcval(T;G,\tc,\tpi)\leq\iopt(G,\tc,\tpi)$, then
$\pcval(T;G,c,\pi)\leq\iopt(G,c,\pi)$ 
(Lemma~\ref{subtract}).   
By choosing a suitable $\tht_v$ for all $v\in V'$ and modifying costs and penalties as
above, we may assume that either: 
(a) there is a node $v\in V'$ with $\tpi_v=0$;
(b) there is a (directed) cycle $Z$ consisting of zero $\tc$-cost arcs; or
(c) there is an out-arborescence consist sing of zero $\tc$-cost arcs.
If (c) applies, then we are done. If (a) or (b) apply, then we argue that may further
simplify the instance as follows. In case (a), we shortcut past $v$ by merging every pair
of incoming and outgoing arcs of $v$ to create a new arc, and delete $v$ (see
steps~\eqref{zeropen-start}--\eqref{zeropen-end}, Lemma~\ref{zeropen}). In case (b), we
contract the cycle $Z$ and set the penalty of the contracted node to be 
$\sum_{v\in V(Z)}\tpi_v$ (see steps~\eqref{zerocyc-start}--\eqref{zerocyc-end2},
Lemma~\ref{zerocyc}). We then recurse on 
the simplified instance. 

An additional feature of our algorithm is that, 
by aggregating the $\tht_v$ values computed by our algorithm
across all recursive calls and translating them suitably to the original graph $G$, we 
obtain a certificate $y=(y_S)_{S\sse V'}$ such that the 
quantity $\dualub=\sum_{S\sse V'}y_S$ is sandwiched between the prize-collecting value
$\pcval(T;G,c,\pi)$ of our solution, and $\iopt(G,c,\pi)$ (which is \nphard to compute).  
(We can in fact strengthen the upper bound on $\dualub$ to 
$\dualub\leq\lpopt(G,c,\pi)$, where recall that $\lpopt(G,c,\pi)$ is the optimal value of
$\lpname[(G,c,\pi)]$; see Theorem~\ref{ybnd}.) 

This property of our algorithm 
is especially useful when we utilize \IterPCA 
to implement approximation algorithms for orienteering (see Section~\ref{orient-apps}),
because there we can utilize $\dualub$ to obtain a suitable {\em upper bound} on the
optimum value of the orienteering problem (and in fact, the optimal value of the
LP-relaxation for orienteering proposed by~\cite{FriggstadS17}). This allows us to obtain an 
{\em instance-wise approximation guarantee} i.e., 
an instance-specific bound on the approximation factor of the solution computed for
each instance. This instance-wise approximation guarantee is often significantly better
than the worst-case approximation guarantee, as is demonstrated by our computational
results (Section~\ref{compres}). Our 
computational results also show that our upper bound is a fairly good (over-)estimate of
the orienteering optimum. We remark that having both (good) lower and upper bounds on the
optimum can be quite useful also for {\em exact} computational methods for orienteering
based on the branch-and-bound method. 

The precise description of our algorithm appears as Algorithm \IterPCA.
By the ``null'' vector below, we mean a vector with no-coordinates.

\begin{function}[ht!]
\caption{IterPCA($G,c,\pi,r$): iterative simplification algorithm for prize-collecting arborescence}
\label{alg:comb}
\KwIn{\pcw instance $\bigl(G=(V,A),c,\pi,r\bigr)$} 
\KwOut{$r$-rooted out-arborescence $T$ in $G$; $y=(y_S)_{S\sse V\sm\{r\}}$ (of polynomial support)}
\SetKwComment{simpc}{// }{}
\SetCommentSty{textnormal}
\DontPrintSemicolon

Let $V'=V\sm\{r\}$, \quad initialize $y\gets\vec{0}$, \ \ $\ty\gets\vec{0}$ \; \label{init}

\lIf{$|V|=1$}{\Return{$(T=\es,\ \text{null vector})$}} 

\uIf{$|V|=2$,\ say $V=\{r, v\}$}{
Set $y_{\{v\}}\gets\min\{c_{r,v},\pi_v\}$ \;
\leIf{$\pi_v>c_{r,v}$}{\Return{$\bigl(T=\{(r,v)\}, y\bigr)$}}{\Return{$(T=\es, y)$}}
}
Set $\tht_v\gets\min\bigl\{\min_{(u,v)\in A}c_{u,v},\pi_v\bigr\}$ for all 
$v\in V'$ \;

For all $v\in V'$, set $\tc_{u,v}\gets c_{u,v}-\tht_v$ for all $(u,v)\in A$,
and $\tpi_v\gets\pi_v-\tht_v$; 
set 
$\tpi_r\gets 0$ \; \label{modinst}

\BlankLine
\uIf{there exists $v\in V'$ with $\tpi_v=0$}{
Set $\bG\gets\bigl(V\sm\{v\}, 
A\sm(\dt^\into(v)\cup\dt^\out(v))\cup\{(u,w): u\in V\sm\{v\}, w\in V\sm\{r,v\}\}\bigr)$ \; 
\label{zeropen-start}

For all $u\in V\sm\{v\}$, $w\in V\sm\{r,v\}$, set 
$\bc_{u,w}\gets\min\{\tc_{u,w},\tc_{u,v}+\tc_{v,w}\}$ \;

Set 
$\bpi\gets\{\tpi_u\}_{u\in V(\bG)}$ \;

$(\bT,\by)\gets\IterPCA(\bG,\bc,\bpi,r)$ 
\; \label{zeropen-rcall}

$\bA\gets\{(u,w)\in\bT:\bc_{u,w}<\tc_{u,w}\}$
\tcp*[r]{note that $\bc_{u,w}=\tc_{uv}+\tc_{v,w}\ \ \forall (u,w)\in\bA$}

$T'\gets\bT\sm\bA\cup\bigcup_{(u,w)\in\bA}\{(u,v),(v,w)\}$ \;

$T\gets$ minimum $\tc$-cost spanning arborescence in $(V(T'),A(T'))$ \; \label{zeropen-end}

Set $\ty_S\gets\by_S$ for all $S\sse V\sm\{r,v\}$ \; \label{zeropen-y}
\BlankLine
}
\uElseIf
{there exists a cycle $Z$ with $r\notin V(Z)$ and $\tc_{u,v}=0$ for all $(u,v)\in A(Z)$}{
Set $\bG\gets$ digraph obtained from $G$ by contracting $Z$ into a single supernode
$u_Z$, removing self-loops, and replacing parallel (incoming or outgoing) arcs incident to
$u_z$ by a single arc \; \label{zerocyc-start} 

Set $\bc_{u,v}\gets\tc_{u,v}$ for all $u\in V\sm V(Z), v\in V\sm V(Z)$ \;

For all $u\in V\sm V(Z)$ such that $\dt^\out(u)\cap\dt^\into(Z)\neq\es$, 
set $\bc_{u,u_z}\gets\min_{(u,v)\in\dt^\into(Z)}\tc_{u,v}$ \;

For all $u\in V'\sm V(Z)$ such that $\dt^\into(u)\cap\dt^\out(Z)\neq\es$,
set $\bc_{u_z,u}\gets\min_{(v,u)\in\dt^\out(Z)}\tc_{v,u}$ \;

Set $\bpi_{u_z}\gets\sum_{v\in V(Z)}\tpi_v$, $\bpi_u\gets\tpi_u$ for all $u\in V\sm V(Z)$ \;

$(\bT,\by)\gets\IterPCA(\bG,\bc,\bpi,r)$ \; \label{zerocyc-rcall}

\If{$u_z\in V(\bT)$}{
Obtain $T'$ from $\bT$ as follows:
replace every arc $a\in\bT$ entering or leaving $u_z$ by the arc in
$G$ entering or leaving $V(Z)$ respectively whose $\tc$-cost defines $\bc_a$; also
add (the nodes and edges of) $Z$ \;

$T\gets$ minimum $\tc$-cost spanning arborescence in $(V(T'),A(T'))$ \; \label{zerocyc-end1}
}
\Else{$T\gets\bT$ \label{zerocyc-end2}}

For each set $\bS\sse V(\bG)\sm\{r\}$, consider the corresponding set $S\sse V'$, which is
$\bS$ if $u_Z\notin\bS$, and $\bS\sm\{u_z\}\cup V(Z)$ otherwise; set $\ty_S\gets\by_S$ \;
\label{zerocyc-y}


\BlankLine
}
\Else{Let $T\gets$ arborescence spanning $V$ with $\tc_{u,v}=0$ for all 
$(u,v)\in A(T)$ \label{zeroarb}} 

\BlankLine

Set $y_{\{v\}}\gets \ty_{\{v\}}+\tht_v$ for all $v\in V'$, and $y_S\gets \ty_S$ for all
other subsets $S\sse V'$.

\Return{$(T, y)$} \label{finalarb}
\end{function}

\paragraph{Analysis.} We prove the following guarantee.

\begin{theorem} \label{combthm} \label{pcarbthm}
On any input $(G,c,\pi)$, algorithm \IterPCA runs in polynomial time and returns an
out-arborescence $T$ and vector $y$ such that 
$\pcval(T;G,c,\pi)\leq\sum_{S\sse V\sm\{r\}}y_S\leq\iopt(G,c,\pi)$. 
\end{theorem}

As noted earlier, one of the above inequalities can be strengthened to 
$\sum_{S\sse V\sm\{r\}}y_S\leq\lpopt(G,c,\pi)$. We defer the proof of this,
which is a bit technical and involves suitably extrapolating the arguments made
for the integral case, to Section~\ref{dualbnd}.

Given the recursive nature of \IterPCA, it is natural that the proof of
Theorem~\ref{pcarbthm} uses induction (on $|V(G)|$). First, Lemma~\ref{subtract} argues
that it suffices to show the inequalities stated in Theorem~\ref{pcarbthm} hold for
the instance $(G,\tc,\tpi)$ specified in step~\eqref{modinst} (with ``simpler'' edge
costs and penalties), the out-arborescence $T$, and the vector $\ty$ returned in
step~\eqref{zerocyc-y} or~\eqref{zeropen-y}. 
Next, Lemmas~\ref{zeropen} and~\ref{zerocyc} supply essentially the induction step. 
They show that if the output $(\bT,\by)$ of \IterPCA when it is
called recursively on the smaller instance $(\bG,\bc,\bpi)$ in step~\eqref{zeropen-rcall}
or~\eqref{zerocyc-rcall} satisfies the inequalities stated in Theorem~\ref{pcarbthm}, then 
$(T,\ty)$ satisfies $\pcval(T;G,\tc,\pi)\leq\sum_{S\sse V'}\ty_S\leq\iopt(G,\tc,\tpi)$. 
Combining this with Lemma~\ref{subtract} finishes the proof.

\begin{lemma} \label{subtract}
Consider the \pcw instance $(G,\tc,\tpi)$ obtained after step~\eqref{modinst}.   
If the out-arborescence $T$ 
computed in step~\eqref{zeropen-end},
\eqref{zerocyc-end1}, \eqref{zerocyc-end2}, or \eqref{zeroarb}, and the vector $\ty$
satisfy
$\pcval(T;G,\tc,\tpi)\leq\sum_{S\sse V'}\ty_S\leq\iopt(G,\tc,\tpi)$, then 
$T$ and the final vector $y$ returned
satisfy $\pcval(T;G,c,\pi)\leq\sum_{S\sse V'}y_S\leq\iopt(G,c,\pi)$.
\end{lemma}

\begin{proof}
We show that 
$\pcval(T;G,c,\pi)=\pcval(T;G,\tc,\tpi)+\sum_{v\in V'}\tht_v$, and 
$\iopt(G,\tc,\tpi)\leq\iopt(G,c,\pi)-\sum_{v\in V'}\tht_v$. Combining these
inequalities, along with the fact that
$\sum_{S\sse V'}y_S=\sum_{S\sse V'}\ty_S+\sum_{v\in V'}\tht_v$, 
yields the lemma.

The first equality follows quite easily, since every node $v\in V'$ covered by $T$ has
exactly one incoming edge whose cost increases by $\tht_v$ when going from $\tc$ to $c$,
and the penalty of every node $v\in V'$ not covered by $T''$ increases by $\tht_v$ when
going from $\tpi$ to $\pi$.
(Note that here we are crucially exploiting that $T$ is an {\em out-arborescence}; if
$T$ were instead the (multi)set of edges of an $r$-rooted walk, or collection of walks,
then $\pcval(T;G,c,\pi)$ could be larger than $\pcval(T;G,\tc,\tpi)+\sum_{v\in V'}\tht_v$
since $T$ could contain multiple edges entering a node.)

To see the second inequality, 
let $\C$ be an optimal solution to the $(G,c,\pi)$ instance. 
So for every node $v\in V'$, if $v'$ is covered by $\C$, it has at least one
incoming edge in this collection of paths, whose cost decreases by $\tht_v$ when moving
from $c$ to $\tc$; if $v'$ is not covered, its penalty decreases by $\tht_v$ when moving
from $\pi$ to $\tpi$. 
Hence, $\iopt(G,\tc,\tpi)\leq\iopt(G,c,\pi)-\sum_{v\in V'}\tht_v$.
\end{proof}

\begin{lemma} \label{zeropen}
Consider a recursive call \IterPCA$\!(G,c,\pi,r)$, where
steps~\eqref{zeropen-start}--\eqref{zeropen-y} are executed. If 
$(\bT,\by)$ obtained in step~\eqref{zeropen-rcall} satisfies
$\pcval(\bT;\bG,\bc,\bpi)\leq\sum_{S\sse V(\bG)\sm\{r\}}\by_S\leq\iopt(\bG,\bc,\bpi)$, 
then the out-arborescence $T$ and the vector $\ty$
computed in steps~\eqref{zeropen-end}, \eqref{zeropen-y} satisfy
$\pcval(T;G,\tc,\tpi)\leq\sum_{S\sse V'}\ty_S\leq\iopt(G,\tc,\tpi)$.
\end{lemma}

\begin{proof}
The key observation is that $\iopt(\bG,\bc,\bpi)\leq\iopt(G,\tc,\tpi)$. 
Consider an optimal solution $\C$ to the \pcw instance $(G,\tc,\tpi)$. 
If $v$ is not covered by $\C$, it is easy to see that $\C$ is
a feasible solution to $(\bG,\bc,\bpi)$. 
Otherwise, we modify each walk $P\in\C$ containing $v$ to obtain a corresponding walk in
$\bG$ as follows. 
Consider an occurrence of $v$ on $P$, and let $(u,v)$ be the arc entering $v$ in this
occurrence. 
If $(u,v)$ is the last arc of $P$, then we simply delete this arc;
note that $\tc_{u,v}\geq 0$. 
Otherwise, if $(v,w)$ is the arc in $P$ leaving $v$ in this occurrence, then we replace
arcs $(u,v),(v,w)$ in $P$ with the arc $(u,w)$; note that
$\bc_{u,w}\leq\tc_{u,v}+\tc_{v,w}$.  
Doing this for all occurrences of $v$ on $P$ yields an $r$-rooted walk in $\bG$, and doing
this for all walks $P\in\C$ containing $v$ yields a feasible solution $\C'$ to
$(\bG,\bc,\bpi)$. 
of no greater prize-collecting cost, i.e., 
$\pcval(\C';\bG,\bc,\bpi)\leq\pcval(\C;G,\tc,\tpi)$. Therefore,
$\iopt(\bG,\bc,\bpi)\leq\iopt(G,\tc,\tpi)$. 

We now have the following sequence of inequalities.
\begin{alignat}{1}
\pcval(T;G,\tc,\tpi) & \leq\pcval(T';G,\tc,\tpi) 
\tag{$T$ is a min $\tc$-cost spanning arborescence in $(V(T'),A(T'))$} \\
& =\pcval(\bT;\bG,\bc,\bpi) 
\tag{if we add $(u,v),(v,w)$ to $T'$, we remove $(u,w)$;\ $\bc_{u,w}=\tc_{u,v}+\tc_{v,w}$} \\ 
& \leq\sum_{S\sse V(\bG)\sm\{r\}}\by_S\leq\iopt(\bG,\bc,\bpi) \tag{given by lemma statement} \\
& \leq\iopt(G,\tc,\tpi). \tag{shown above} 
\end{alignat}
Finally, note that, by definition, $\sum_{S\sse V'}\ty_S=\sum_{S\sse V(\bG)\sm\{r\}}\by_S$.
\end{proof}

\begin{lemma} \label{zerocyc}
Consider a recursive call \IterPCA$\!(G,c,\pi,r)$, where
steps~\eqref{zerocyc-start}--\eqref{zerocyc-y} are executed. If 
$(\bT,\by)$ obtained in step~\eqref{zerocyc-rcall} satisfies
$\pcval(\bT;\bG,\bc,\bpi)\leq\sum_{S\sse V(\bG)\sm\{r\}}\by_S\leq\iopt(\bG,\bc,\bpi)$,
then the out-arborescence $T$ computed in step~\eqref{zerocyc-end1}
or~\eqref{zerocyc-end2}, and the vector $\ty$ computed in step \eqref{zerocyc-y} satisfy 
$\pcval(T;G,\tc,\tpi)\leq\sum_{S\sse V'}\ty_S\leq\iopt(G,\tc,\tpi)$.
\end{lemma}

\begin{proof}
Again, the key property to show is that $\iopt(\bG,\bc,\bpi)\leq\iopt(G,\tc,\tpi)$.
Consider an optimal solution $\C$ to the \pcw instance $(G,\tc,\tpi)$. 
If no nodes of $Z$ are covered by $\C$, then it is easy to see
that $\C$ is a feasible solution to $(\bG,\bc,\bpi)$ of no-greater prize-collecting
cost, so $\iopt(\bG,\bc,\bpi)\leq\iopt(G,\tc,\tpi)$. 
Otherwise, pick some $v\in V(Z)$ that lies on some walk in our collection $\C$, and
think of $Z$ being contracted into the node $v$; i.e., formally, we are replacing every 
occurrence of every node of $Z$ in our collection $\C$ by the contracted node $u_Z$ of
$\bG$ that stands for the cycle $Z$, and deleting self-loops.
This yields a walk-collection in $\bG$ visiting $\bigcup_{P\in\C}V(P)\cup\{u_Z\}$ 
where the $\bc$-cost of the arcs used is at most $\sum_{P\in\C}\tc(P)$, since for every arc
$(u,v)\in\dt^\into(Z)$ (respectively $(u,v)\in\dt^\out(Z)$), we have the arc $(u,u_Z)\in\bG$
(respectively, $(u_z,u)\in\bG$) with $\bc_{(u,u_Z)}\leq\tc_{u,v}$ (respectively,
$\bc_{u_z,u}\leq\tc_{v,u}$). 
So we again have
$\iopt(\bG,\bc,\bpi)$ is at most 
$\sum_{P\in\C}\tc(P)+\tpi\bigl(V\sm\bigcup_{P\in\C}V(P)\bigr)=\iopt(G,\tc,\tpi)$.

If we obtain $T$ in step~\eqref{zerocyc-end1}, then
$$
\pcval(T;G,\tc,\tpi)\leq\pcval(T';G,\tc,\tpi)=\pcval(\bT;\bG,\bc,\bpi)
\leq\sum_{S\sse V(\bG)\sm\{r\}}\by_S\leq\iopt(\bG,\bc,\bpi)\leq\iopt(G,\tc,\tpi).
$$
The equality above follows since all arcs $a\in Z$ have $\tc_a=0$, and for every arc $a$
of $\bG$ in $\bT$ that is replaced by an arc $a'$ of $G$, we have $\bc_a=\tc_a$. 
If we obtain $T$ in step~\eqref{zerocyc-end2}, then clearly
$\pcval(T;G,\tc,\tpi)=\pcval(\bT;\bG,\bc,\bpi)$, which is at most
$\sum_{S\sse V(\bG)\sm\{r\}}\by_S\leq\iopt(\bG,\bc,\bpi)\leq\iopt(G^\met,\tc,\tpi)$ as
before. 

Finally, the lemma follows by noting that 
$\sum_{S\sse V'}\ty_S=\sum_{S\sse V(\bG)\sm\{r\}}\by_S$. 
\end{proof}

\begin{proofof}{Theorem~\ref{pcarbthm}}
The proof follows by induction on $|V(G)|$. The bases cases are when $|V(G)|\leq 2$, for
which the statement follows trivially. Suppose that the statement is true whenever
$|V(G)|\leq k$, and consider an instance $(G,c,\pi)$ with $|V(G)|=k+1$. Recall that
$V'=V\sm\{r\}$. 

By Lemma~\ref{subtract}, it suffices to show that 
$\pcval(T;G,\tc,\tpi)\leq\sum_{S\sse V'}\ty_S\leq\iopt(G,\tc,\tpi)$.
If $T$ is obtained in step~\eqref{zeroarb}, then clearly
$\pcval(T;G,\tc,\tpi)=0=\sum_{S\sse V'}\ty_S\leq\iopt(G,\tc,\tpi)$.
Otherwise, by the induction
hypothesis, we have that the tuple $(\bT,\by)$ returned for $(\bG,\bc,\bpi)$ in step
\eqref{zeropen-rcall} or \eqref{zerocyc-rcall} satisfies
$\pcval(\bT;\bG,\bc,\bpi)\leq\sum_{S\sse V(\bG)\sm\{r\}}\by_S\leq\iopt(\bG,\bc,\bpi)$,
since $|V(\bG)|\leq k$. 
Lemmas~\ref{zeropen} and~\ref{zerocyc}, then show that
$\pcval(T;G,\tc,\tpi)\leq\sum_{S\sse V'}\ty_S\leq\iopt(G,\tc,\tpi)$. 
This completes the induction step and the induction proof showing that 
$\pcval(T;G,c,\pi)\leq\sum_{S\sse V'}y_S\leq\iopt(G,c,\pi)$.
\end{proofof}

\subsection{Showing that \boldmath $\sum_{S\sse V\sm\{r\}}y_S\leq\lpopt(G,c,\pi)$} 
\label{dualbnd} 
Recall that $V'=V\sm\{r\}$, and that $\lpopt(G,c,\pi)$ is the optimal value of the
LP-relaxation \eqref{eq:primal} for the instance $(G,c,\pi)$. 
We prove the above inequality by suitably generalizing the arguments
involving $\iopt$ in Lemmas~\ref{subtract}--\ref{zerocyc} to work with
fractional solutions to \eqref{eq:primal}.
%
A key technical tool that we utilize, is the following powerful {\em splitting-off}
result due to Frank~\cite{Frank89} and Jackson~\cite{Jackson88}.
For a digraph $D$, and any ordered pair of nodes $u,v$, let $\ld_D(u,v)$ denote the
$(u,v)$ edge connectivity in $D$, which is the number of $u\leadsto v$ edge-disjoint paths
in $D$. 

\begin{theorem}[\cite{Frank89,Jackson88}] \label{dirsplit}
Let $D=(N+s,E)$ be a an Eulerian digraph, possibly with parallel edges. Then, for every
arc $(u,s)\in\dt^\into(s)$, there is an arc $(s,w)\in\dt^\out(s)$ such that letting
$D_{uw}$ be the digraph obtained by replacing the pair of arcs $(u,s),(s,w)$ with (a new
parallel copy of) the arc $(u,w)$---an operation called splitting off $(u,s), (s,w)$---we
have that $\ld_{D_{uw}}(v,t)=\ld_D(v,t)$ for all $v,t\in N$. 
\end{theorem}

Given a digraph $D=(N,E)$ with root node $r\in N$, we say that a vector $x\in\R_+^E$ is an
{\em $r$-preflow} 
if $x\bigl(\dt^\into(v)\bigr)\geq x\bigl(\dt^\out(v)\bigr)$ holds for every $v\in N\sm\{r\}$.
We say that that $D$ is an $r$-preflow digraph 
if $\chi^E$ is an $r$-preflow.
Given a solution $(x,p)$ to \eqref{eq:primal}, scaling $x$ suitably yields an $r$-preflow
digraph, whereas Theorem~\ref{dirsplit} pertains to Eulerian digraphs.
However, since we are only interested in $(r,u)$ edge-connectivities, we can always make
this $r$-preflow digraph Eulerian by adding enough parallel $(v,r)$ edges for each node
$v\neq r$. Applying Theorem~\ref{dirsplit} repeatedly to the resulting Eulerian digraph
then yields the following.

\begin{lemma} \label{preflowsplit}
Let $D=(N+s,E)$ be an $r$-preflow digraph, where $r\in N$. Then, we can perform a
sequence of the following two types of operations: (i) delete an arc entering $s$; (ii)
split off arcs $(u,s)\in\dt^\into(s)$ and $(s,w)\in\dt^\out(s)$, to obtain an $r$-preflow 
digraph $D'=(N,E')$ such that $\ld_{D'}(r,v)=\ld_D(r,v)$ for all $v\in N$.
\end{lemma}

\begin{proof}
We first make $D$ Eulerian by adding, for each node $v$, $|\dt^\into(v)|-|\dt^\out(v)|$
parallel $(v,r)$ edges; we call these edges artificial edges. Let $D''=(N+s,E'')$ be the
resulting Eulerian digraph. Note that this operation leaves the $(r,v)$
edge-connectivities unchanged, for all $v\in N\cup\{s\}$. Now, we 
apply Theorem~\ref{dirsplit} repeatedly to split off pairs of incoming and outgoing edges
incident to $s$. If the outgoing edge of $s$ that is split off is an artificial edge, then
we simply delete the corresponding incoming edge. Note that each such operation preserves
the property that every node $v\in N\sm\{r\}$ has in-degree at least its out-degree.
This yields a digraph $\tD=(N,\tE)$ such
that $\ld_{\tD}(r,v)=\ld_{D''}(r,v)=\ld_D(r,v)$ for all $x\in N$. Removing the artificial
edges from $\tD$ (which again does not affect $(r,v)$ edge-connectivities) yields the
desired $r$-preflow digraph $D'$.
\end{proof}

\begin{theorem} \label{ybnd}
The vector $y$ returned by algorithm \IterPCA satisfies 
$\sum_{S\sse V'}y_S\leq\lpopt(G,c,\pi)$.
\end{theorem}

\begin{proof}
As with Theorem~\ref{pcarbthm}, we proceed by induction on $|V(G)|$. The statement again
holds trivially for the base cases, where $|V(G)|\leq 2$. So consider an instance
$(G,c,\pi)$ with $|V(G)|\geq 3$ and $\dt_G^\into(r)=\es$.

We first claim that $\lpopt(G,\tc,\tpi)\leq\lpopt(G,c,\pi)-\sum_{v\in V'}\tht_v$.
This follows because if $(x,p)$ is a feasible solution to 
$\lpname[(G,c,\pi)]$ (with $p_r=0$) then we have 
$$
\sum_{a\in A}c_ax_a+\sum_{v\in V}\pi_vp_v
=\sum_{v\in V'}\Bigl(\sum_{a\in\dt^\into(v)}\tc_ax_a+\tpi_vp_v+\tht_v\bigl(x(\dt^\into(v))+p_v\bigr)\Bigr)
\geq\lpopt(G,\tc,\tpi)+\sum_{v\in V'}\tht_v.
$$
Given the above, 
it suffices to argue that the vector 
$\ty$ that we have in the algorithm at the end of step~\eqref{zeropen-y},
\eqref{zerocyc-y}, or \eqref{zeroarb} satisfies $\sum_{S\sse V'}\ty_S\leq\lpopt(G,\tc,\tpi)$.   
If $\ty$ is the vector at step~\eqref{zeroarb}, then $\ty=\vec{0}$, so this holds
trivially. 

\medskip
Suppose that $\ty$ is obtained from step~\eqref{zeropen-y}. By the induction hypothesis, 
the vector $\by$ returned in step~\eqref{zeropen-rcall} satisfies 
$\sum_{S\sse V(\bG)\sm\{r\}}\by_S\leq\lpopt(\bG,\bc,\bpi)$. We show that the RHS is at
most $\lpopt(G,\tc,\tpi)$ by showing that any feasible solution 
$(x,p)$ to $\lpname[(G,\tc,\tpi)]$
induces a feasible solution to $\lpname[(\bG,\bc,\bpi)]$ of no greater cost.

Let $K$ be such that $Kx$ is 
integral. Consider the digraph $D=(V,E)$ obtained by including $Kx_{u,w}$ parallel
$(u,w)$ arcs for every $(u,w)\in A$. Observe that $D$ is an $r$-preflow digraph. 
We apply Lemma~\ref{preflowsplit} to $D$,
taking $s=v$, to obtain an $r$-preflow digraph $D'=(V\sm\{v\},E')$ with
$\ld_{D'}(r,u)=\ld_D(r,u)\geq K(1-p_u)$ for all $u\in V'\sm\{v\}$. 
We give every parallel edge $(u,w)$ in $E'$ cost equal to $\bc_{u,w}$.
Observe that $\bc(E')\leq\tc(E)$ since 
every edge $(u,w)\in E'\sm E$ is obtained by splitting off a pair $(u,v), (v,w)$, and 
$\bc_{u,w}\leq\tc_{u,v}+\tc_{v,w}$. Let $\bx\in\R^{A(\bG)}_+$ be the
vector where $\bx_{u,w}=(\text{no. of parallel copies of $(u,w)$ in $D'$})/K$. Then, note
that $\bigl(\bx,\{p_u\}_{u\in V\sm\{v\}}\bigr)$ is a feasible solution to 
$\lpname[(\bG,\bc,\bpi)]$ having objective value at most 
$\bc(E')/K+\sum_{u\in V'\sm\{v\}}\bpi_up_u\leq\sum_{a\in A}\tc_ax_a+\sum_{u\in V}\tpi_up_u$.

\medskip
Next, suppose that $\ty$ is obtained from step~\eqref{zerocyc-y}.
Again, by the induction hypothesis, the vector $\by$ returned in
step~\eqref{zerocyc-rcall} satisfies 
$\sum_{S\sse V(\bG)\sm\{r\}}\by_S\leq\lpopt(\bG,\bc,\bpi)$, and  we show that the RHS is
at most $\lpopt(G,\tc,\tpi)$.
Let $(x,p)$ be a feasible solution to $\lpname[(G,\tc,\tpi)]$.
Define $\bx\in\R_+^{A(\bG)}$ and $\bp\in\R_+^{V(\bG)}$ as follows.
For every arc $(u,v)\in A(\bG)$, where $u,v\in V\sm V(Z)$, set $\bx_{u,v}=x_{u,v}$; 
for every arc $(u,u_Z)\in A(\bG)$, set $\bx_{u,u_Z}=x\bigl(\dt^\out(u)\cap\dt^\into(Z)\bigr)$;
for every arc $(u_z,u)\in A(\bG)$, set $\bx_{u_z,u}=x\bigl(\dt^\into(u)\cap\dt^\out(Z)\bigr)$.
Set $\bp_u=p_u$ for all $u\in V\sm V(Z)$, and $\bp_{u_Z}:=\min_{u\in V(Z)}p_u$.
We claim that $(\bx,\bp)$ is a feasible solution to $\lpname[(\bG,\bc,\bpi)]$ of cost at
most the cost of $(x,p)$ for $\lpname[(G,\tc,\tpi)]$.

Any set $\bS\sse V(\bG)\sm\{r\}$ maps to a corresponding set $S\sse V'$, which is
$\bS$ if $u_Z\notin\bS$, and $\bS\sm\{u_z\}\cup V(Z)$ otherwise, and we have defined $\bx$
to ensure that $\bx\bigl(\dt^\into(\bS)\bigr)=x\bigl(\dt^\into(S)\bigr)$. So 
for any $u\in\bS$, taking $w=u$ if $u\neq u_Z$, and $w=\argmin_{v\in V(Z)}p_v$ if $u=u_Z$, 
we obtain that $\bx\bigl(\dt^\into(\bS)\bigr)+\bp_u=x\bigl(\dt^\into(S)\bigr)+p_w\geq 1$. 
We have $\sum_{a\in A(\bG)}c_a\bx_a\leq\sum_{a\in E}\tc_ax_a$, since each
arc $(u,u_Z)\in A(\bG)$ has $\bc_{u,u_Z}=\min_{v\in V(Z)}\tc_{u,v}$, and 
each $(u_Z,u)\in A(\bG)$ has $\bc_{u_Z,u}=\min_{v\in V(Z)}\tc_{v,u}$.
Finally, we also have 
$$
\sum_{v\in V(\bG)}\bpi_v\bp_v=
\sum_{v\in V\sm V(Z)}\tpi_vp_v+\Bigl(\sum_{v\in V(Z)}\tpi_v\Bigr)\cdot\min_{v\in V(Z)}p_v\leq
\sum_{v\in V\sm V(Z)}\tpi_vp_v+\sum_{v\in V(Z)}\tpi_vp_v.
$$
This completes the induction step, and hence the proof.
\end{proof}

\section{Applications} \label{apps}

\subsection{Orienteering} \label{orient-apps}
We now show that our algorithm for prize-collecting arborescence (PCA) can be used to
obtain a fast, combinatorial implementation of the LP-rounding based approximation
algorithms devised by Friggstad and Swamy~\cite{FriggstadS17} for orienteering.
The input to the orienteering problem consists of a (rational) metric space
$(V,c)$, root $r\in V$, a distance bound $B \geq 0$, and nonnegative node
rewards $\{\pi_v\}_{v \in V}$. Let $G=(V,E)$ denote the complete graph on $G$.
Three versions of orienteering are often considered in the literature. 
\begin{enumerate}[label=$\bullet$, topsep=0.5ex, noitemsep, leftmargin=*]
\item {\em Rooted orienteering}: find an $r$-rooted path of cost at
  most $B$ that collects the maximum reward.  
\item {\em Point-to-point (P2P) orienteering}: we are also given an end node $t$, and
  we seek an $r$-$t$ path of cost at most $B$ that collects maximum reward.
\item {\em Cycle orienteering}: 
  find a cycle containing $r$ of cost at most $B$ that collects maximum reward.
\end{enumerate}
By merging nodes at zero distance from each other, we may assume that all distances are
positive, and by scaling, we may further assume that they are positive integers. 
We may therefore also assume that $B$ is an integer.

Friggstad and Swamy~\cite{FriggstadS17} propose an LP-relaxation for rooted
orienteering, and show that an optimal LP-solution can be rounded to an integer
solution losing a factor of $3$. This is obtained by decomposing an LP-optimal
solution into a convex combination of out-arborescences, and then extracting a rooted path  
from these arborescences. They adapt their approach to also obtain a
$6$-approximation for P2P orienteering. 

We show that one can utilize \IterPCA to obtain 
combinatorial algorithms for rooted- and P2P- orienteering with the above approximation
factors. 
The high level idea is that Lagrangifying the ``cost at most $B$'' constraint for rooted
orienteering yields a prize-collecting walks problem, and by fine-tuning the value of the
Lagrangian variable, we can leverage \IterPCA to obtain a distribution of $r$-rooted trees
having expected cost at most $B$, and expected reward (essentially) at least the optimum
of the rooted orienteering problem. We can then combine this with the LP-rounding algorithm
in~\cite{FriggstadS17} to obtain the stated approximation factors.
Our algorithms can thus be seen as a combinatorial implementation of the LP-rounding
algorithms in~\cite{FriggstadS17}.  
For cycle orienteering, we adapt the above idea and the analysis in~\cite{FriggstadS17},
to obtain a combinatorial $4$-approximation algorithm.
%

We also leverage the certificate $y$ returned by \IterPCA (whose value
$\sum_{S\sse V'}y_S$ is a lower bound on the optimal value of \pcw problem)
and show that this can be used to provide {\em upper bounds on the optimal value of the 
\{rooted, P2P, cycle\}- orienteering problem}.
As mentioned  earlier, having such upper bounds is quite useful 
as 
it allows to assess the approximation guarantee on an instance-by-instance basis, which
can often be much better than the worst-case approximation guarantee (of $3$). Indeed, our
computational experiments in Section~\ref{compres} emphatically confirm this.

\vspace*{-1ex}
\paragraph{Finding a cost-bounded tree with good reward.} 
The following primitive will be the basis of all our algorithms. 
Let $N\sse V$ be a node-set, $r,w\in N$, and $\budg$ be a cost-budget such that 
$c_{rw}\leq L$.  
Let $Q^*$ be an $r$-rooted path such that $\{w\}\sse V(Q^*)\sse N$
(but $w$ need not be an end-node of $Q^*$) and $c(Q^*)\leq\budg$, 
and collecting the maximum reward among such paths; let $\ropt=\pi(V(Q^*))$. 

We would ideally like to find an $r$-rooted tree $T$ such that:
(a) $\{w\}\sse V(T)\sse N$; (b) $c(T)\leq\budg$; and
(c) $\pi(V(T))\geq\ropt$. We will not quite be able to
achieve this, we describe how to use our PCA algorithm to obtain a distribution of (at
most) two trees satisfying (a) with probability 1, and (b) and (c) in expectation.
Roughly speaking, we run \IterPCA 
for the prize-collecting $r$-rooted walks problem with
node-set $N$, metric $c$, and penalties $\{\ld\pi_v\}_{v\in N}$, and tune the parameter
$\ld$ using binary search to obtain the weighted trees; 
algorithm \BinSearchPCA describes this precisely.
%

\begin{function}[ht!]
\caption{BinSearchPCA($N,c,\pi,\budg,r,w;\e$): binary search using IterPCA}
\label{treeorient}
\KwOut{$(\gm_1,\T_1)$, $(\gm_2,\T_2)$ such that $\gm_1,\gm_2\geq 0$, $\gm_1+\gm_2=1$; if
  $\gm_1=1$, we do not specify $(\gm_2,\T_2)$}
\SetKwComment{simpc}{// }{}
\SetCommentSty{textnormal}
\DontPrintSemicolon

Let $D=(N,A)$ be the digraph obtained by restricting $G$ to the nodes in $N$
and bidirecting its edges, where both $(u,v)$ and $(v,u)$ get cost $c_{uv}$. 
Let $n=|N|$ and $N'=N\sm\{r\}$. \;

Let $c_{\max}=\max_{u,v\in N}c_{uv}$,
$\pi_{\min}=\min_{u\in N:\pi_u>0}\pi_u$, 
$\lb=\max\;\{\pi_u: u\in N, \min\{c_{ru},c_{rw}\}+c_{uw}\leq\budg\}$ \;

Set $\tpi_v(\ld)=\ld\pi_v$ for all $v\in N\sm\{w\}$, 
and $\tpi_w(\ld)=\tpi_w=nc_{\max}$ 

\tcp{For $\ld\geq 0$, let $\bigl(T_\ld,y^{(\ld)}\bigr)$ denote the tuple returned by
\IterPCA$\!(D,c,\tpi(\ld),r)$}

Let $\binub\gets nc_{\max}/\pi_{\min}$ and $\binlb\gets 1/\pi(N')$ 
\tcp*[r]{we show that $c(T_{\binlb})\leq\budg$}

\lIf{$c(T_\binub)\leq\budg$}{\Return{$(1,T_\binub)$}} \label{binupper}

Perform binary search in the interval $[\binlb,\binub]$ to find either:
(i) a value $\ld\in[\binlb,\binub)$ such that $c(T_\ld)=\budg$; or 
(ii) values $\ld_1,\ld_2\in[\binlb,\binub]$ with 
$0<\ld_2-\ld_1\leq\e\cdot\binlb^2\cdot\lb$ 
such that $c(T_{\ld_1})<\budg<c(T_{\ld_2})$ \;

\lIf{case (i) occurs}{\Return{$(1,T_\ld)$}} \label{casei}

\uIf{case (ii) occurs}{
Let $a,b\geq 0,\ a+b=1$ be such that 
$a\cdot c(T_{\ld_1})+b\cdot c(T_{\ld_2})=\budg$.
\Return{$(a,T_{\ld_1})$, $(b,T_{\ld_2})$} \label{caseii}
}
\end{function}

\begin{theorem} \label{pcabsearch}
Let $\e\in(0,1)$.
The output of \BinSearchPCA$\!(N,c,\pi,\budg,r,w;\e)$ satisfies the following:
(a) $\{w\}\sse V(\T_i)\sse N$ for $i=1,2$; 
(b) $\sum_{i=1}^2\gm_i c(\T_i)\leq\budg$; and
(c) $\sum_{i=1}^2\gm_i\pi(V(\T_i))\geq(1-\e)\ropt$.
\end{theorem}

\begin{proof}
We abbreviate $\iopt\bigl(D,c,\tpi(\ld),r\bigr)$
to $\iopt(\ld)$, and $\pcval\bigl(T_\ld;D,c,\tpi(\ld),r\bigr)$ to $\pcval(\ld)$.   
Let $\dualub(\ld)$ denote $\sum_{S\sse N'}y^{(\ld)}_S$.
Part (b) holds by construction.
Note that $\iopt(\ld)\leq (n-1)c_{\max}$ for all $\ld\geq 0$. 
So $T_\ld$ must include $w$ for all $\ld\geq 0$, and so (a) holds. 
We also have that $c(T_\binlb)\leq\budg$ as
otherwise, $c(T_\binlb)\geq\budg+1$, whereas $\iopt(\binlb)<c_{rw}+1\leq\budg+1$.

For all $\ld\geq 0$, we have 
$\iopt(\ld)\leq c(Q^*)+\ld\cdot\pi(N\sm V(Q^*))\leq\budg+\ld\cdot\pi(N\sm V(Q^*))$.
By Theorem~\ref{pcarbthm}, for all $\ld\geq 0$, we then have
\begin{equation}
c(T_\ld)+\ld\cdot\pi(N\sm V(T_\ld))=\pcval(\ld)\leq\dualub(\ld)\leq\iopt(\ld) 
\leq \budg+\ld\cdot\pi(N\sm V(Q^*)). 
\label{pcbnd}
\end{equation}

Note that $T_\binub$ must be an arborescence spanning $N$, otherwise we have
$\pcval(\binub)\geq nc_{\max}>\iopt(\binub)$. So if we return in step~\eqref{binupper}
then we are done.
If we return in step~\eqref{casei}, then again we are done, since \eqref{pcbnd}
implies that $\pi(V(T_\ld))\geq\ropt$.

Suppose we return in step~\eqref{caseii}.
Note that $\ropt\geq\lb$.
Multiplying \eqref{pcbnd} for $\ld=\ld_1$ by $a$, and \eqref{pcbnd} for $\ld=\ld_2$ by
$b$, and adding and simplifying,  
we obtain that
\begin{alignat*}{2}
&& a\ld_1\cdot\pi(N\sm V(T_{\ld_1}))+b\ld_2\cdot\pi(N\sm V(T_{\ld_2})) &
\leq (a\ld_1+b\ld_2)\cdot\pi(N\sm V(Q^*)) \\
\implies & \quad 
& \ld_2\Bigl[a\cdot\pi(N\sm V(T_{\ld_1}))+b\cdot\pi(N\sm V(T_{\ld_2}))\Bigr] & 
\leq\ld_2\cdot\pi(N\sm V(Q^*))+a(\ld_2-\ld_1)\cdot\pi(N\sm(V(T_{\ld_1}))).
\end{alignat*}
This implies that 
$a\cdot\pi(V(T_{\ld_1}))+b\cdot\pi(V(T_{\ld_2}))\geq\ropt
-\bigl(a\e\cdot\binlb^2\cdot\lb\cdot\pi(N')\bigr)/{\ld_2}\geq(1-\e)\ropt$.   
\end{proof}

We remark that we can avoid the $(1-\e)$-factor loss in Theorem~\ref{pcabsearch} (while
still retaining polynomial running time) by terminating the binary search 
at a smaller value of $\ld_2-\ld_1$; 
see Appendix~\ref{append-avoidloss}.

\subsubsection{Rooted orienteering} \label{rooted}
The rounding theorem of~\cite{FriggstadS17} is stated below, paraphrased to suit our
purposes.  
The {\em regret} (also called {\em excess}~\cite{BlumCKLMM07,BansalBCM04}) of a $u$-$v$
path $P$ with respect to its end-points is $\creg(P)=c(P)-c_{uv}$.

\begin{theorem}[\cite{FriggstadS17}] \label{orient-rnd}
Fix $w\in V$. Let $T_1,\ldots,T_k$ be rooted trees in $G$ with associated weights  
$\gm_1,\ldots,\gm_k\geq 0$ such that: 
(i) $\sum_{i=1}^k\gm_i=1$; (ii) $\sum_{i=1}^k \gm_ic(T_i)\leq B$; and 
(iii) 
$w\in V(T_i)$ for all $i=1,\ldots,k$. 
Then, for each $i=1,\ldots,k$, we can extract a rooted path $P_i$ from $T_i$ 
(visiting some subset of $V(T_i)$) with $\creg(P_i)\leq B-c_{rw}$, such that  
$\max_{i=1,\ldots,k}\pi\bigl(V(P_i)\bigr)\geq\frac{1}{3}\cdot\sum_{i=1}^k\gm_i\pi\bigl(V(T_i)\bigr)$.
\end{theorem}

Friggstad and Swamy~\cite{FriggstadS17} obtain the $T_i$s from an
optimal solution to the LP-relaxation for rooted orienteering that they propose, and they
satisfy $V(T_i)\sse\{u\in V: c_{ru}\leq c_{rw}\}$, so that each path $P_i$ obtained above
has $c(P_i)\leq B$.
But it is not hard to see that the trees required in Theorem~\ref{orient-rnd} 
can instead be supplied using algorithm \BinSearchPCA,
thereby avoiding the need for solving the LP in~\cite{FriggstadS17} and decomposing its
optimal solution into out-arborescences.

For any $w\in V$ with $c_{rw}\leq B$, let 
$\ropt_w$ be the 
maximum reward of rooted path that visits $w$,
and can only visit nodes in $\bV_w=\{u\in V: c_{rv}\leq c_{rw}\}$. 
Let $\e\in(0,1)$ be a given parameter. 
We execute \BinSearchPCA$\!(\bV_w,c,\pi,B,r,w;\e)$. Recall that this procedure 
calls \IterPCA on \pcw-instances of the form $(D,c,\tpi(\ld),r)$, for different
$\ld\geq 0$ values.
The output of \BinSearchPCA is a distribution over at most two $T_\ld$ trees, each
containing only nodes from $\bV_w$, which, by Theorem~\ref{pcabsearch}, 
satisfies the conditions of Theorem~\ref{orient-rnd} (when viewed as a weighted collection
of trees) 
and has expected reward at least $(1-\e)\ropt_w$. So combining this with
Theorem~\ref{orient-rnd} yields the following.

\begin{theorem} \label{orient-thm}
Considering all $w\in V$ with $c_{rw}\leq B$, and
applying Theorem~\ref{orient-rnd} to the output of \BinSearchPCA$(\bV_w,c,\pi,B,r,w;\e)$,
and taking the best solution returned, yields a
combinatorial $3/(1-\e)$-approximation algorithm for rooted orienteering.  
\end{theorem}

\vspace*{-2ex}
\paragraph{Upper bound.} 
%
For a given guess $w$, and any $\ld\ge 0$, we have from \eqref{pcbnd} that 
$\dualub(\ld)\leq\iopt(\ld)\leq B+\ld\bigl(\pi(\bV_w)-\ropt_w\bigr)$.
Rearranging gives 
$\ropt_w\leq\UBrooted(w,B;\ld):=\pi(\bV_w)+\frac{B-\dualub(\ld)}{\ld}$.
This holds for all $\ld\geq 0$, so we have
$\ropt_w\leq\min_{\ld\geq 0}\UBrooted(w,B;\ld)$.%
\footnote{As $\ld\to\infty$, $\UBrooted(w,B;\ld)$ approaches the trivial upper bound 
$\pi(\bV_w)$; we expect the term $\frac{B-\dualub(\ld)}{\ld}$ to be negative for large
$\ld$ (unless the min-cost arborescence spanning $\bV_w$ has cost at most $B$), 
and hence to yield 
{savings over this trivial upper bound.}} 
Since we do not know the right choice for $w$, 
we can say that the optimal value for rooted orienteering is at most 
$\UBrooted(B):=\max_{w\in V:c_{rw}\leq B}\min_{\ld\geq 0}\UBrooted(w,B;\ld)$.%
\footnote{Note that $\UBrooted(r,B;\ld)=\pi_r$ for all $\ld\geq 0$, since $\bV_r=\{r\}$ and so
$T_\ld=\es$ and $\dualub(\ld)=0$ for all $\ld$.}

\begin{remark} \label{lpbounds}
We can strengthen our bounds to show that our approximation guarantee and our upper
bound $\UBrooted(B)$, both hold with respect to the optimal value of the LP-relaxation for
rooted orienteering proposed by~\cite{FriggstadS14}. More precisely, let $\rorientlpopt$
be the optimal value of the LP-relaxation for rooted orienteering in~\cite{FriggstadS17}. 
In Appendix~\ref{append-lpbounds}, we show that if $\rewd$ is the reward of the solution
returned by Theorem~\ref{orient-thm}, then we have 
$\rorientlpopt\leq\UBrooted(B)\leq\frac{3}{1-\e}\cdot\rewd$. In particular, similar
to~\cite{FriggstadS17}, we obtain an approximation guarantee of $\frac{3}{1-\e}$ with
respect to the LP-optimum. 
\end{remark}

\subsubsection{P2P orienteering} \label{ptporient}
Recall that here we seek an $r$-$t$ path of cost at most $B$ that achieves
maximum reward.
Friggstad and Swamy~\cite{FriggstadS17} show that one can utilize Theorem~\ref{orient-rnd}
on two suitable weighted collections of trees to obtain a $6$-approximation for P2P
orienteering. Suppose we ``guess'' the node $w$ on the optimal $r$-$t$ path with largest 
$c_{ru}+c_{ut}$ value.  
Suppose that we can obtain two weighted collections of trees
$\bigl(\gm_i^{(r)},T_i^{(r)}\bigr)_{i=1,\ldots,k}$ and 
$\bigl(\gm_i^{(t)},T_i^{(t)}\bigr)_{i=1,\ldots,\ell}$ such that:
\begin{enumerate}[(a), topsep=0.5ex, noitemsep, leftmargin=*]
\item $\bigl(\gm_i^{(r)},T_i^{(r)}\bigr)_{i=1,\ldots,k}$ satisfies the
conditions of Theorem~\ref{orient-rnd} with root $r$ and cost budget $B^{(r)}=B-c_{wt}$;
\item 
$\bigl(\gm_i^{(t)},T_i^{(t)}\bigr)_{i=1,\ldots,\ell}$ satisfies the
conditions of Theorem~\ref{orient-rnd} with root $t$ and cost budget $B^{(t)}=B-c_{rw}$;
\item each tree in both collections contains only nodes $u$ with 
$c_{ru}+c_{ut}\leq c_{rw}+c_{wt}$.
\end{enumerate}
%
We apply Theorem~\ref{orient-rnd} with root $r$ and budget $B^{(r)}$ 
on the $(\gm_i^{(r)},T_i^{(r)})_{i=1,\ldots,k}$ collection, 
and apply Theorem~\ref{orient-rnd} with root $t$ and budget $B^{(t)}$
on the $(\gm_i^{(t)},T_i^{(t)})_{i=1,\ldots,\ell}$ collection.  
Due to (c), the path extracted from any of the trees in the two
collections extends to an $r$-$t$ path of cost at most $B$. So one of the paths
extracted from the two collections attains reward  
at least $\frac{1}{6}$ times the total weighted reward of the two collections.

In~\cite{FriggstadS17}, the two collections are obtained from an optimal solution to their
P2P-orienteering LP. But, as with rooted orienteering, we can utilize \BinSearchPCA
to obtain the two collections. 
More precisely, let $w$ be a guess with $c_{rw}+c_{wt}\leq B$.  
Let $P^*_w$ be the optimal P2P-orienteering solution that visits $w$, and only visits
nodes in $\bV^{\ptp}_w=\{u\in V: c_{ru}+c_{ut}\leq c_{rw}+c_{wt}\}$. 
Let $\roptptp_w=\pi(V(P^*_w))$.
For any two nodes $u,v\in V(P^*_w)$, we use $P^*_{w,uv}$ to denote the $u$-$v$ portion of 
$P^*_w$.
%
\begin{enumerate}[label=$\bullet$, topsep=0.5ex, noitemsep, leftmargin=*]
\item 
\BinSearchPCA$\!(\bV^{\ptp}_w,c,\pi,B-c_{wt},r,w;\e)$ yields a distribution
over at most two $r$-rooted trees, each containing $w$ and only nodes from $\bV^{\ptp}_w$, with
expected cost at most $B-c_{wt}$ and expected reward at least $(1-\e)\pi(V(P^*_{w,rw}))$. 
\item 
\BinSearchPCA$\!(\bV^{\ptp}_w,c,\pi,B-c_{rw},t,w;\e)$ yields a distribution over at most two
$t$-rooted trees, each containing $w$ and only nodes from $\bV^{\ptp}_w$,  
with expected cost at most $B-c_{rw}$ and expected reward at least
$(1-\e)\pi(V(P^*_{w,wt}))$. 
\end{enumerate}
Thus, as discussed above, using Theorem~\ref{orient-rnd} on these two distributions, we
obtain an $r$-$t$ path of reward at least $\frac{1-\e}{6}\cdot\roptptp_w$. Trying
all $w\in V$ with $c_{rw}+c_{wt}\leq B$, and returning the best solution yields a
$6/(1-\e)$-approximation.

As a side-note, we can refine the analysis in~\cite{FriggstadS17} to show that their
P2P-orienteering algorithm is in fact a $4$-approximation algorithm. 
We discuss this in Appendix~\ref{append-ptpimprov}

\vspace*{-1ex}
\paragraph{Upper bound.} 
We can easily extend the upper-bound approach used for rooted orienteering as follows. 
For a given guess $w$, let $\dualub_r(\ld)$ denote 
$\sum_{S\sse N'}y^{(\ld)}_S$ in the run 
\BinSearchPCA$\!(\bV^{\ptp}_w,c,\pi,B-c_{wt},r,w;\e)$, and 
$\dualub_t(\ld)$ denote 
$\sum_{S\sse N'}y^{(\ld)}_S$ in the run 
\BinSearchPCA$\!(\bV^{\ptp}_w,c,\pi,B-c_{rw},t,w;\e)$.

For any $r$-rooted walk $Q$ such that $\{w\}\sse V(Q)\sse\bV^\ptp_w$, and any $\ld\geq 0$,
we have that $\dualub_r(\ld)\leq c(Q)+\ld\cdot\pi(\bV^\ptp_w)-\ld\cdot\pi(V(Q))$, which
implies that $\pi(V(Q))\leq\pi(\bV^\ptp_w)+\frac{c(Q)-\dualub_r(\ld)}{\ld}$. 
\begin{enumerate}[label=$\bullet$, topsep=1ex, itemsep=0.75ex, leftmargin=*]
\item Taking $Q=P^*_{w,rw}$ gives 
$\pi\bigl(V(P^*_{w,rw})\bigr)\leq\UBptprw(w,B;\ld):=\pi(\bV^{\ptp}_w)+\frac{B-c_{wt}-\dualub_r(\ld)}{\ld}$.
\item Taking $Q=P^*_w$ gives
$\roptptp_w=\pi\bigl(V(P^*_w)\bigr)\leq\pi(\bV^\ptp_w)+\frac{B-\dualub_r(\ld)}{\ld}$.
\end{enumerate}
Similarly, for any $t$-rooted walk $Q$ with $\{w\}\sse V(Q)\sse\bV^\ptp_w$, and any
$\ld\geq 0$, we have that \linebreak
$\pi(V(Q))\leq\pi(\bV^\ptp_w)+\frac{c(Q)-\dualub_t(\ld)}{\ld}$.
Therefore
$$
\pi\bigl(V(P^*_{w,wt})\bigr)\leq\UBptpwt(w,B;\ld):=\pi(\bV^{\ptp}_w)+\frac{B-c_{rw}-\dualub_t(\ld)}{\ld},
\qquad
\roptptp_w\leq\pi(\bV^\ptp_w)+\frac{B-\dualub_t(\ld)}{\ld}.
$$
%
%
Combining all these bounds, we obtain that
\begin{multline*}
\UBptp(B):=\max_{w\in V:c_{rw}+c_{wt}\leq B}\ \min\Bigl\{
\pi(\bV^{\ptp}_w)+
\min_{\ld\geq 0}\,\min\bigl\{\tfrac{B-\dualub_r(\ld)}{\ld},\tfrac{B-\dualub_t(\ld)}{\ld}\bigr\}, \\
\Bigl(\min_{\ld\geq 0}\,\UBptprw(w,B;\ld)
+\min_{\ld\geq 0}\,\UBptpwt(w,B;\ld)\Bigr)\Bigr\}
\end{multline*}
is an upper bound on the optimal value 
for P2P orienteering.

\subsubsection{Cycle orienteering} \label{cycorient}
Recall that here we seek a cycle containing $r$ of cost at most $B$ that achieves maximum 
reward. Taking $t=r$ in our approach for P2P-orienteering yields a combinatorial
$6$-approximation algorithm.%
\footnote{If $w$ is a node on the optimal solution with maximum $c_{ru}$, then
running Theorem~\ref{orient-rnd} on the output of 
\BinSearchPCA$\!(\bV_w,c,\pi,B-c_{rw},r,w;\e)$, where $\bV_w=\{u\in V:c_{ru}\leq c_{rw}\}$, 
yields an $r$-rooted path $P$ ending at some node $u\in\bV_w$ with 
$\creg(P)\leq B-2c_{rw}$ and $\pi(V(P))$ at least $\frac{1-\e}{6}$ times the optimum; $P$
then extends to a cycle of cost at most $B$.}
But we can refine this approach and utilize \BinSearchPCA to obtain a $4$-approximation, 
as also refine our upper-bounding strategy, by leveraging the fact that the tree returned
by \IterPCA has prize-collecting cost at most the optimal value of \eqref{eq:primal}.  

Following a familiar theme, for any $w\in V$ with $c_{rw}\leq B/2$, let $\cyc^*_w$ be the
optimal cycle-orienteering solution that visits $w$, and only visits nodes in
$\bV_w=\{u\in V: c_{ru}\leq c_{rw}\}$. Let $\roptcyc_w=\pi(V(\cyc^*_w))$.
Consider the distribution over (at most two)
$r$-rooted trees output by \BinSearchPCA$\!(\bV_w,c,\pi,B/2,r,w;\e)$. We claim that this
has expected reward at least $\roptcyc_w/2-\e\cdot\roptcyc_w$. This is because
sending a 
$\frac{1}{2}$-unit of flow from $r$ to $w$ along the two $r$-$w$ paths in $\cyc^*_w$
yields a solution to the LP-relaxation $\lpname[(D,c,\tpi(\ld),r)]$ of objective value at
most $B/2+\ld\bigl(\pi(\bV_w)-\roptcyc_w/2-\pi_r/2-\pi_w/2\bigr)
\leq B/2+\ld\bigl(\pi(\bV_w)-\roptcyc_w/2\bigr)$. 
The claim then follows by reasoning exactly as in the proof of Theorem~\ref{pcabsearch}.

If \BinSearchPCA returns a single tree (of cost at most $B/2$), then doubling and
shortcutting yields a rooted cycle of cost at most $B$ and reward at least
$\roptcyc_w/2$. Otherwise, suppose \BinSearchPCA returns $(a,T_{\ld_1})$, $(b,T_{\ld_2})$.
We convert each $T_{\ld_i}$ to an 
$r$-$w$ path $Q_i$, and  
we can extract a rooted path $P_i$ from $Q_i$ ending at some node in $\bV_w$ such that 
$\creg(P_i)\leq B-2c_{rw}$, and 
$\pi(V(P_i))\geq\pi(V(Q_i))/\bigl(\frac{\creg(Q_i)}{B-2c_{rw}}+1\bigr)$  
(see Lemma 5.1 in~\cite{BlumCKLMM07}, or Lemma 2.2 in~\cite{FriggstadS14}). 
Each $P_i$ can thus be completed to a rooted cycle of cost at most $B$. 
Also, noting that $a\cdot\creg(Q_1)+b\cdot\creg(Q_2)\leq B-2c_{rw}$, 
it follows that one of the $P_i$s has reward at least 
$$
\frac{a\cdot\pi(V(Q_1))+b\cdot\pi(V(Q_2))}{a\cdot\frac{\creg(Q_1)}{B-2c_{rw}}+b\cdot\frac{\creg(Q_2)}{B-2c_{rw}}+1}
\geq\frac{(\frac{1}{2}-\e)\cdot\roptcyc_w}{2}=\frac{1-2\e}{4}\cdot\roptcyc_w.
$$

\begin{theorem} \label{cycorientthm} 
Executing the above procedure for all $w\in V$ with $c_{rw}\leq B/2$, and returning the 
best solution yields a $4/(1-2\e)$-approximation algorithm for cycle orienteering. 
\end{theorem}

\vspace*{-2ex}
\paragraph{Upper bound.}
Fix a given $w$ with $c_{rw}\leq B/2$. 
Let $\dualub(\ld)$ denote $\sum_{S\sse N'}y^{(\ld)}_S$ in the run
\BinSearchPCA$\!(\bV_w,c,\pi,\budg,r,w;\e)$. (For a fixed $\ld\geq 0$, the
instance $(D,c,\tpi(\ld),r)$ constructed in \BinSearchPCA$\!(\bV_w,c,\pi,\budg,r,w;\e)$,
and hence $\dualub(\ld)$, does not depend on $\budg$.) 
Comparing $\dualub(\ld)$ with the prize-collecting cost of two solutions, then yields
the following upper bounds.
\begin{enumerate}[label=$\bullet$, topsep=0.5ex, noitemsep, leftmargin=*]
\item Considering $\cyc^*_w$ (interpreted as two $r$-$w$ walks), we obtain that 
$\dualub(\ld)\leq B+\ld\bigl(\pi(\bV_w)-\roptcyc_w)$,
or $\roptcyc_w\leq\UBrooted(w,B;\ld)$, where recall that 
$\UBrooted(w,B;\ld):=\pi(\bV_w)+\frac{B-\dualub(\ld)}{\ld}$.


\item 
Considering the fractional solution 
where we send $\frac{1}{2}$-unit of flow along the two $r$-$w$ paths in $\cyc^*_w$,
we have that 
$\dualub(\ld)\leq B/2+\ld\bigl(\pi(\bV_w)-\roptcyc_w/2-\pi_r/2-\pi_w/2\bigr)$, which 
leads to
$\roptcyc_w\leq\UBcychalfb(w,B;\ld):=2\cdot\pi(\bV_w)-\pi_r-\pi_w+\frac{B-2\cdot\dualub(\ld)}{\ld}$. 

\end{enumerate}
Combining these bounds, we obtain that the optimal value for cycle orienteering is at most
$$
\UBcyc(B):=\max_{w\in V:c_{rw}\leq B/2}\min_{\ld\geq 0}\,
\min\,\Bigl\{\UBrooted(w,B;\ld),\UBcychalfb(w,B;\ld)\Bigr\}.
$$

\subsection{The \boldmath $k$ minimum-latency problem (\kmlp)} \label{kmlp-apps}
Recall that in \kmlp, 
we have a metric space $(V,c)$ and root
$r\in V$. The goal is to find (at most) $k$-rooted paths that together cover every node,
so as to minimize the sum of visiting times of the nodes. The current-best approximation
ratio for \kmlp is $7.183$~\cite{PostS15}. Post and Swamy~\cite{PostS15} devise two
algorithms for \kmlp, both having approximation ratio (roughly) $7.183$. 
One of their algorithms is ``more combinatorial'' (see Algorithm 3, in Section 6.2
in~\cite{PostS15}) and relies on 
having access to the following procedure: 
\begin{enumerate}[label=, topsep=0.5ex, noitemsep, leftmargin=2ex]
\item Given a node-set $N\sse V$, root $r\in N$, and an integer $1\leq k\leq|N|$, let
  $\budg^*$   
  be the minimum cost of a collection of $r$-rooted walks that together cover at least $k$
  nodes. Find a distribution over $r$-rooted trees, that in expectation cover at least $k$
  nodes, and whose expected cost is at most $\budg^*$. 
\end{enumerate}
In~\cite{PostS15}, the distribution is obtained by applying the arborescence-packing
result of Bang-Jensen et al.~\cite{BangjensenFJ95} to the optimal solution to
\eqref{eq:primal} with node rewards $\ld$ to obtain a rooted tree of no-greater
prize-collecting cost,
and then varying $\ld$ in a binary-search procedure (as we do) to obtain the desired
distribution (see the proof of Corollary 3.3 in~\cite{PostS15}).%
\footnote{
  More directly,
  one can modify \eqref{eq:primal} by changing the prize-collecting objective to
  the cost objective, and adding the coverage constraint to the LP, and
  apply~\cite{BangjensenFJ95} to decompose the optimal solution to this LP into a convex
  combination of $r$-rooted trees.}
We can instead 
utilize \IterPCA within a binary-search procedure to obtain the
desired distribution (over at most two trees).  
Incorporating this in the more-combinatorial algorithm of~\cite{PostS15}
yields a fully (and truly) combinatorial $7.183$-approximation algorithm for \kmlp.  

\begin{function}[ht!]
\caption{BSearch($N,k$): binary search using IterPCA}
\label{ntreeorient}
\KwOut{$(\gm_1,\T_1)$, $(\gm_2,\T_2)$ such that $\gm_1,\gm_2\geq 0$, $\gm_1+\gm_2=1$; if
  $\gm_1=1$, we do not specify $(\gm_2,\T_2)$}
\SetCommentSty{textnormal}
\DontPrintSemicolon

Let $D=(N,A)$ be the digraph obtained by restricting $G$ to the nodes in $N$
and bidirecting its edges, where both $(u,v)$ and $(v,u)$ get cost $c_{uv}$. 
Let $n=|N|$. \linebreak 
Let $c_{\max}=\max_{u,v\in N}c_{uv}$. Let $M$ be an integer (computable in polytime) with
$\log M=\poly(\text{input size})$ such that the $c_e$s are integer multiples of
$\frac{1}{M}$. 

\tcp{For $\ld\geq 0$, let $\bigl(T_\ld,y^{(\ld)}\bigr)$ denote the tuple returned by
\IterPCA$\!(D,c,\{\pi_v=\ld\}_{v\in N},r)$}

Let $\binub\gets nc_{\max}$ and $\binlb\gets 0$ 
\tcp*[r]{we show that $|V(T_{\binlb})|\leq k\leq|V(T_{\binub})|$}

\lIf{$|V(T_\binlb)|=k$}{\Return{$(1,T_\binlb)$}} \label{blower}
\lIf{$|V(T_\binub)|=k$}{\Return{$(1,T_\binub)$}} \label{bupper}

Perform binary search in the interval $[\binlb,\binub]$ to find either: 
(i) a value $\ld\in[\binlb,\binub)$ such that $|V(T_\ld)|=k$; or 
(ii) values $\ld_1,\ld_2\in[\binlb,\binub]$ with $0<\ld_2-\ld_1\leq 1/n^2M$
such that $|V(T_{\ld_1})|<k<|V(T_{\ld_2})|$ \;

\lIf{case (i) occurs}{\Return{$(1,T_\ld)$}} \label{ncasei}

\uIf{case (ii) occurs}{
Let $a,b\geq 0,\ a+b=1$ be such that 
$a\cdot |V(T_{\ld_1})|+b\cdot |V(T_{\ld_2})|=k$.
\Return{$(a,T_{\ld_1})$, $(b,T_{\ld_2})$} \label{ncaseii}
}
\end{function}

\begin{theorem} \label{bsearch}
The output of \BSearch$\!(N,k)$ satisfies:
(a) $\sum_{i=1}^2\gm_i c(\T_i)\leq\budg^*$; and
\nolinebreak\mbox{(b) $\sum_{i=1}^2\gm_i|V(\T_i)|=k$.}
\end{theorem}

\begin{proof}
We abbreviate $\iopt\bigl(D,c,\{\pi_v=\ld\}_{v\in N},r\bigr)$
to $\iopt(\ld)$, and $\pcval\bigl(T_\ld;D,c,\{\pi_v=\ld\}_{v\in N},r\bigr)$ to $\pcval(\ld)$.   
Part (b) holds by construction.

We mimic (and simplify) the proof of part (i) of Corollary 3.3 in~\cite{PostS15}.
We have $\iopt(\ld)\leq (n-1)c_{\max}$ for all $\ld\geq 0$,  
so $T_\binub$ must be an arborescence spanning $N$. Also, note that $T_{\binlb}$
is the trivial tree $\{r\}$.
Let $\C^*$ be a min-cost collection of $r$-rooted walks covering at least $k$ nodes. 
So $\budg^*=\sum_{P\in\C^*}c(P)$. Let $n^*=\bigl|\bigcup_{P\in\C^*}V(P)\bigr|\geq k$.
For all $\ld\geq 0$, we have 
$\iopt(\ld)\leq \budg^*+\ld(n-n^*)$. By Theorem~\ref{pcarbthm}, for all $\ld\geq 0$, we
then have 
\begin{equation}
\pcval(\ld)=c(T_\ld)+\ld\bigl(n-|V(T_\ld)|\bigr) 
\leq \budg^*+\ld(n-n^*). 
\label{npcbnd}
\end{equation}
If we return in steps~\eqref{blower}, \eqref{bupper}, or \eqref{ncasei}, then we are
done, since \eqref{npcbnd} implies that $c(T_\ld)\leq\budg^*$.

Suppose we return in step~\eqref{ncaseii}. Multiplying \eqref{npcbnd} for $\ld=\ld_1$ by
$a$, and \eqref{npcbnd} for $\ld=\ld_2$ by $b$, and adding and simplifying, we obtain that  
\begin{equation*}
\begin{split}
a\cdot c(T_{\ld_1})+b\cdot c(T_{\ld_2})
& \leq \budg^*+\ld_2(n-n^*)
-\Bigl[a\ld_1\bigl(n-|V(T_{\ld_1})|\bigr)+b\ld_2\bigl(n-|V(T_{\ld_2})|\bigr)\Bigr] \\
& \leq \budg^*+a(\ld_2-\ld_1)\bigl(n-|V(T_{\ld_1})|\bigr)
<\budg^*+\tfrac{1}{nM}.
\end{split}
\end{equation*}
Noting that $a$ and $b$ are rational numbers with denominator at most $n$, and all $c_e$s
are multiples of $1/M$, the quantity
$a\cdot c(T_{\ld_1})+b\cdot c(T_{\ld_2})-\budg^*$ is an integer multiple of
$\frac{1}{n'M}$ for some $n'\leq n$. 
So if this quantity is less then $\frac{1}{nM}$, then it must be nonpositive; that is, we
have $a\cdot c(T_{\ld_1})+b\cdot c(T_{\ld_2})\leq\budg^*$.
\end{proof}

\section{Computational results for orienteering} \label{compres}

In this section, and its continuation in Appendix~\ref{append-compres}, we present
various computational results on the performance of our orienteering algorithms (from
Section~\ref{orient-apps}) in order to assess the performance of our algorithms in
practice. Our experiments clearly demonstrate the effectiveness of both our algorithms
and our upper bounds. They show that the instance-wise approximation ratios, for both the
solution returned and the computed upper bound, are substantially better than the
theoretical worst-case bounds, and in fact fairly close to $1$.

We implemented our algorithms in Section~\ref{orient-apps} (and Section~\ref{comb}) in
\texttt{C++11}, and our experiments were run on a 2019 MacBook Pro laptop with a 2.3 GHz
Intel Core i9 processor (8 cores) and 16 GB of RAM. 
Our implementation matches almost exactly the description in Section~\ref{orient-apps}
(and Section~\ref{comb}), with the following salient differences. 
\begin{enumerate}[label=$\bullet$, topsep=0.5ex, noitemsep, leftmargin=*]
\item We terminate the binary search (in \BinSearchPCA) when the interval
$[\ld_1,\ld_2]$ has width $\ld_2 - \ld_1 \leq 10^{-6}$, the precision of the
{\tt double} data type in {\tt C++}.

\item For a given guess $w$, our analysis shows that if the binary search for $\ld$
terminates with the interval $[\ld_1,\ld_2]$, then 
the solution extracted from one of $T_{\ld_1}$ and $T_{\ld_2}$ achieves the stated
approximation. 
In our experiments, we extract a solution from $T_\ld$ for every $\ld$ encountered in the
binary search, and return the best of these solutions (which clearly retains the stated
approximation guarantee).

\item In the computation of the upper bounds on the orienteering optimum---i.e.,
$\UBrooted$, $\UBptp$, $\UBcyc$---instead of considering all $\ld\geq 0$, we consider
only the $\ld$s encountered in the relevant binary search procedures.

\item In our cycle orienteering experiments, we actually consider two binary searches. One
where the search is guided by a cost budget of $B/2$ and one where the search is
guided by a cost budget of $B - c_{rw}$. 
This produces a different set of $\lambda$ values to consider, which, in many cases, led
to better solutions (and upper bounds) than those found 
by using just the binary search for distance $B/2$.%
\footnote{Theoretically speaking, if we consider {\em all} $\ld$ values (instead of only
  those encountered in the binary search),  
  then the cost budget becomes inconsequential and does not influence either the final
  solution obtained or the upper bound computed.}   

\item When extracting a rooted path of a given regret bound $\regret$ from a path $P$
(that is obtained from a tree), 
instead of using a greedy procedure (Lemma 5.1 in~\cite{BlumCKLMM07}, or Lemma 2.2
in~\cite{FriggstadS14}), we find the maximum-reward subpath $Q$ of $P$ such that the 
rooted path with node sequence $r, V(Q)$ has regret at most $\regret$.

\item 
  If the current binary search for a guessed node
  $w$ encounters a value $\lambda$ such that the upper bound for orienteering computed
  from the  
  output of \IterPCA is at most the value of the best orienteering solution found thus
  far, then we prematurely terminate the binary search for this guess $w$ and move on to
  the next guess, since we know that we cannot find a better solution for this guess. This
  pruning improves the running time by roughly 40\%.
\end{enumerate}

We proceed to describe our computational results for \{rooted, P2P, cycle\}- orienteering
in detail.  
We discuss cycle orienteering first, as this 
is computationally the most
well-studied version of orienteering, 
and discuss rooted orienteering (Section~\ref{rooted-compres}) last, as this does not seem
to have been computationally studied in the literature.  
We use the following legend throughout: 
{$\val$} denotes the reward of the solution returned by our algorithm,
{$\opt$} is the optimal value, and 
{$\ub$} is the upper bound computed by our algorithm.
For each orienteering version, we first give an overview of the approximation
ratios obtained by our algorithm, and the quality of our upper bounds, and then present
detailed results.  
Our histograms specify the distribution of various quantities (e.g., the ratios
$\opt/\val$, $\ub/\opt$ etc.)
across the instances used in the computational experiments.
Each histogram bar corresponds to a range of values (for a particular quantity)
as indicated on the $x$-axis, 
and its height specifies the number
of instances for which the quantity lies in the range.

\subsection{Cycle orienteering experiments} \label{cycle-compres}
As noted earlier,
for each guess $w$ of furthest node, we run two binary-search procedures, with target
budgets $B/2$ and $B-c_{rw}$, 
and output 
the best cycle-orienteering solution
extracted from the trees $T_{w,\ld}$, over all $w$, and all $\ld$ values encountered in
the binary-search procedures run for $w$.
Also, the upper bound we compute is $\UBcyc(B)$ except that for a given $w$, we only
consider the $\ld$ values encountered in the binary searches for $w$.

\vspace*{-1ex}
\paragraph{Test Data.}  
We use the same TSP instances considered in \cite{Kobeaga21}. All but three of these datasets were also considered in \cite{Fischetti98}.
We note that 5 of these are not metric, however upper bounds still apply to non-metric instances so we kept them in our experiments.
Interestingly, we still obtain good approximations in these instances.
See Appendix~\ref{sec:tables} for details.

For each dataset, \cite{Fischetti98} and
\cite{Kobeaga21} generate node rewards in three ways: 
\begin{itemize}
\item {\bf Gen 1} - {\bf Uniform Rewards}: All nodes apart from the root $r$ have reward 1.
\item {\bf Gen 2} - {\bf Pseudo-Random Rewards}: The reward of the $j$'th node in the dataset is $1 + (7141 \cdot j + 73) \pmod {100}$ apart from the root, which has reward 0.
\item {\bf Gen 3} - {\bf Far Away Rewards}: The reward of a node $v \neq r$ is $1 + \lfloor 99 \cdot c_{rv} / \max_w c_{rw}\rfloor$. This is meant to create more challenging instances where the high-reward nodes are further from the root.
\end{itemize}
In total, this yields 135 different datasets. The distance bound used in each case is 
$\lceil\opttsp/2\rceil$, where $\opttsp$ is the cost of the optimal TSP-tour for that
dataset (which is provided in TSPLIB). 
%
Optimal values are known for all these datasets; most of these were computed
in~\cite{Fischetti98}, and the rest are from~\cite{Kobeaga21}. This allows us to evaluate
not only the  
gap between our solution's value and our computed upper bound, but also the true
approximation guarantee of our algorithm on these datasets. 

\pgfplotstableread{csv/histogram.txt}\histogramtable

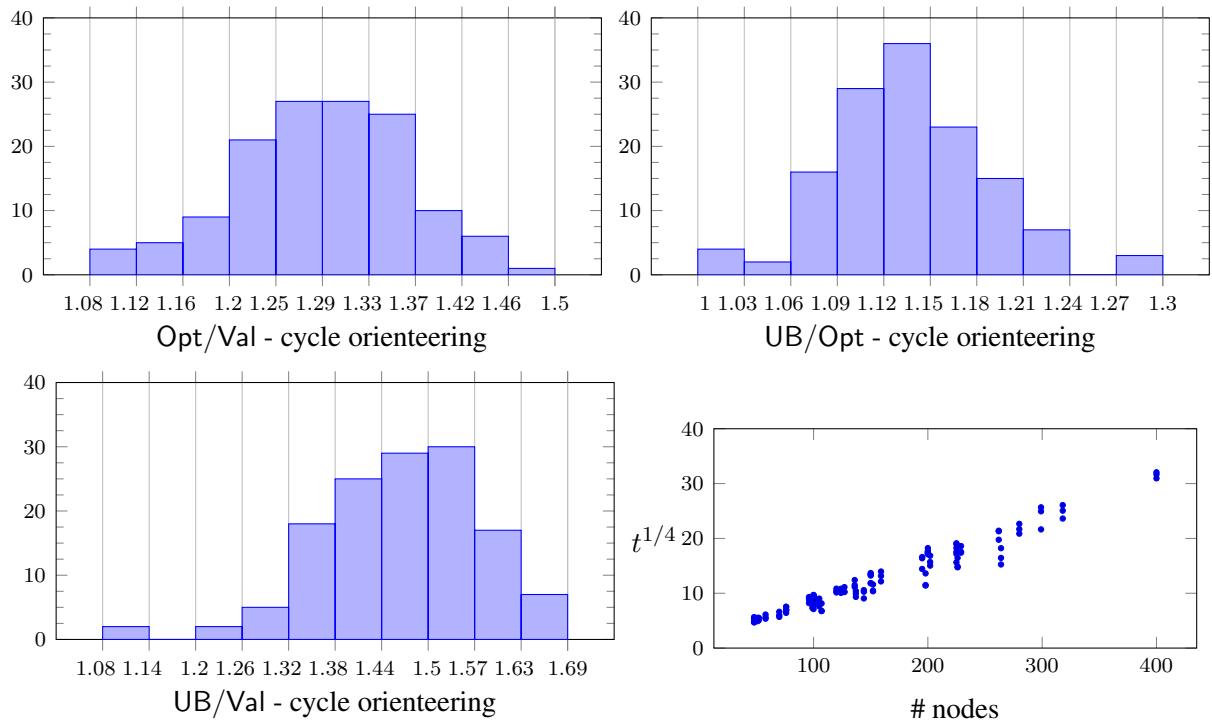
\begin{figure}[ht!]
\begin{center}
\begin{tabular}{|c|c|c|c|c|c| }\hline
 No. of nodes & \# Experiments & Mean $\opt/\val$ & Max $\opt/\val$ & Mean $\ub/\opt$ & Max $\ub/\opt$ \\ \hline
 48--400 & 135 & 1.29 & 1.5 & 1.14 & 1.304 \\ \hline
\end{tabular}

\bigskip
\hspace*{-0.2in}
\begin{tikzpicture}
\begin{axis}[ybar interval, width=9cm, height = 5cm, ymax=40, ymin=0, 
    minor y tick num = 3, tick label style={font=\scriptsize}, 
    extra x ticks={1.5}, 
    extra x tick style={x tick label as interval=false, xticklabel style={anchor=north}},
    xlabel={$\opt/\val$ - cycle orienteering}, xticklabel style={anchor=north east}]
\addplot+ [hist={bins=10}] table [y=cyc-apx] {\histogramtable};
\end{axis}
\end{tikzpicture}
\begin{tikzpicture}
\begin{axis}[ybar interval, width=9cm, height = 5cm, ymax=40, ymin=0, 
    minor y tick num = 3, tick label style={font=\scriptsize}, 
    extra x ticks={1.304}, 
    extra x tick style={x tick label as interval=false, xticklabel style={anchor=north}},
    xlabel={$\ub/\opt$ - cycle orienteering}, xticklabel style={anchor=north east}]
\addplot+ [hist={bins=10}] table [y=cyc-ubopt] {\histogramtable};
\end{axis}
\end{tikzpicture}

\hspace*{-0.2in}
\begin{tikzpicture}
\begin{axis}[ybar interval, width=9cm, height = 5cm, ymax=40, ymin=0, 
    minor y tick num = 3, tick label style={font=\scriptsize}, 
    extra x ticks={1.688}, 
    extra x tick style={x tick label as interval=false, xticklabel style={anchor=north}},
    xlabel={$\ub/\val$ - cycle orienteering}, xticklabel style={anchor=north east}]
\addplot+ [hist={bins=10}] table [y=cyc-ubval] {\histogramtable};
\end{axis}
\end{tikzpicture}
%
%
\begin{tikzpicture}
\begin{axis}[width=8cm, height = 4.5cm, ymax=40, ymin=0, 
    tick label style={font=\scriptsize},
    xlabel={\# nodes}, ylabel={$t^{1/4}$}, ylabel style={rotate=-90, anchor=west}]
\addplot+ [
    only marks,
    mark=*, 
    mark size=1.0pt] table [x=nodes, y=runtime] {\histogramtable};
\end{axis}
\end{tikzpicture}

\end{center}

\vspace*{-4ex}
\caption{Cycle orienteering statistics. In the scatter plot, 
  $t$ denotes the running time in milliseconds on the corresponding instance. 
The slowest running time was 17.6 minutes.
  \label{cyc-stats}} 
\end{figure}

Tables~\ref{cyc-gen1} and~\ref{cyc-gen3} report the results of our
experiments on the Gen 1 and Gen 3 data sets. Table~\ref{cyc-gen2} in
Appendix~\ref{append-compres} gives our results for the Gen 2 data sets. Besides reporting 
the value ($\val$) of our solution, and the computed upper bound ($\ub$), we also report
the instance-wise approximation ratio $\opt/\val$, and the ratio $\ub/\opt$, which measures
the quality of our upper bound.

The five TSPLIB datasets marked by * in these tables, \texttt{att48}, \texttt{gr48},
\texttt{hk48}, \texttt{brazil58} and \texttt{gr120}, are in fact non-metric (but still
symmetric) instances. 
The maximum of $c_{uv}/ (c_{uw} + c_{wv})$ over all distinct triples of nodes in these
datasets is $1.002$, $1.580$, $1.326$, $9.783$, and $5.218$, respectively. 
Rather than moving to the shortest-path metric, we used these distance matrices as is,
to facilitate a comparison with the optimal solution computed in prior work.
Our upper bounds 
remain valid even with non-metric instances. 
While we cannot claim any (worst-case) guarantees on non-metric instances,
as our results show, our algorithm performs fairly well even on these non-metric
instances. 


\medskip
It is pertinent to compare our results with Paul et al.~\cite{Paul19}, which is the only 
other work that performs a computational evaluation of a (polytime) approximation
algorithm for orienteering. 
They develop a $2$-approximation algorithm for cycle orienteering (which they call
budgeted prize-collecting TSP), and run computational experiments on two types of
datasets: 
1) 37 metric TSPLIB instances with at most 400 nodes, each having unit reward; and 
2) 37 instances constructed from different weeks of usage of the Citi Bike network of bike
sharing stations in New York City, where node rewards correspond to an estimate of the
number of broken docks at that station during the week.
In all datasets, they consider cost budgets equal to $0.5$, $1$, $1.5$ times the cost of
an MST (note that twice the MST cost is an upper bound on the optimal TSP-tour cost).

However, their experiments are only run for the unrooted version of cycle orienteering,
wherein the solution does not have to involve any particular root node. 
Running our algorithm by trying all possible root nodes would lead to a factor-$n$ blow up
in our running time. Instead of this, we considered running our algorithms in 
{\em combination} to see if this yields improved solutions for the underlying instances. 
In particular, we run our algorithm as a fast {\em postprocessing} step: we pick an
arbitrary node on the solution output by the algorithm in~\cite{Paul19} as the root node
$r$, and (in the same spirit) pick the node furthest from $r$ on their solution as our guess $w$
of the furthest node.
Our results (see Fig.~\ref{fig:paul}) show that this postprocessing almost always
yields improvements, sometimes by a significant factor,
on both the TSPLIB and the Citi Bike data sets. 
(This is despite the fact that our worst-case analysis only shows that the algorithm that
tries all possible guesses for $w$ is a $4$-approximation for the rooted case,
whereas the algorithm in~\cite{Paul19} is a $2$-approximation.) 
The results are summarized in Figure \ref{fig:paul}.%
\footnote{We have omitted one instance each from the TSPLIB and Citi Bike datasets
in~\cite{Paul19}, where our algorithm yields an improvement by a factor exceeding
$2$. This would seemingly contradict the $2$-approximation guarantee of~\cite{Paul19}, but 
the discrepancy arises because the implementation in~\cite{Paul19} is slightly
different from their $2$-approximation algorithm~\cite{Paul21}.} 

\begin{figure}[ht!]
\begin{center}

\hspace*{-0.2in}
\begin{tikzpicture}
\begin{axis}[ybar interval, width=9cm, height = 5cm, ymax=50, ymin=0, 
    minor y tick num = 3, tick label style={font=\scriptsize}, 
    extra x ticks={1.742}, 
    extra x tick style={x tick label as interval=false, xticklabel style={anchor=north}},
    xlabel={Citi Bike Dataset}, xticklabel style={anchor=north east}]
\addplot+ [hist={bins=10}] table [y=improvement] {csv/bike_unrooted.txt};
\end{axis}
\end{tikzpicture}
%
\begin{tikzpicture}
\begin{axis}[ybar interval, width=9cm, height = 5cm, ymax=50, ymin=0, 
    minor y tick num = 3, tick label style={font=\scriptsize}, 
    extra x ticks={1.4151}, 
    extra x tick style={x tick label as interval=false, xticklabel style={anchor=north}},
    xlabel={TSPLIB Dataset}, xticklabel style={anchor=north east}]
\addplot+ [hist={bins=10}] table [y=improvement] {csv/tsp_unrooted.txt};
\end{axis}
\end{tikzpicture}

\end{center}

\vspace*{-4ex}
\caption{A histogram of the factors by which we improve the solutions from
  \cite{Paul19}. Ratios $\leq 1.0$ indicate that we found no improvement. \label{fig:paul}} 
\end{figure}
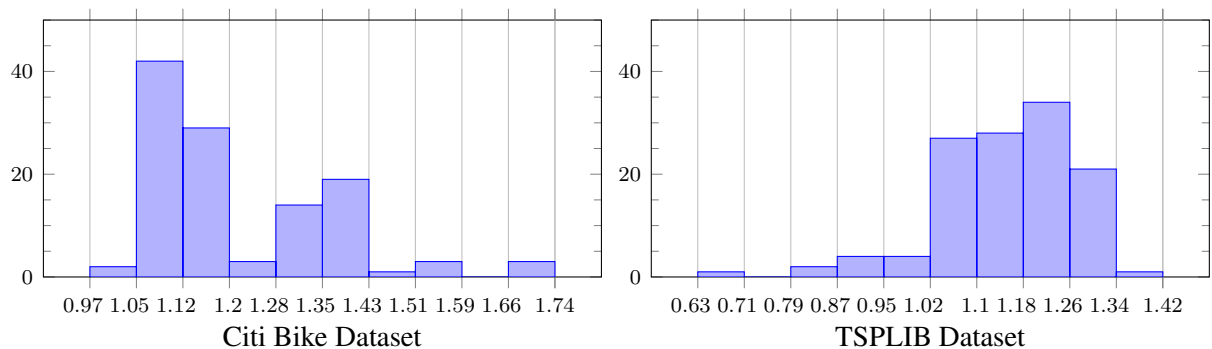

\begin{table}[t!]
\begin{center}
{\small
    \caption{Cycle orienteering results - Gen 1 \label{cyc-gen1}}

    \medskip\medskip
    \begin{tabular}{l|l|l|l|l|l|l|l}%
    \bfseries Dataset & {\boldmath $B$} & {\boldmath $\val$} & {\boldmath $\opt$} &
    {\boldmath $\ub$} & {\boldmath $\opt/\val$} & {\boldmath $\ub/\opt$} & {\boldmath $\ub/\val$} \\\hline
    \hline \vspace{-4.325mm}
    \begin{filecontents*}{csv/cycles-gen1.csv}
Dataset,nodes,Distance Limit,Orienteering,Distance,Upper bound,opt-val,Approximation factor,Upper bound / val,Upper bound / OPT,Comment,
*att48,48,5314,26,4870,35.92,31,1.192,1.381,1.159,Rounded,docker run -i --init ghcr.io/ssinad/dcvr-thesis/cycle-orienteering-app -D 5314 < tsp/att48.tsp-gen-1.txt 2>/dev/null;
*gr48,48,2523,23,2472,35.48,31,1.348,1.543,1.145,Not metric,docker run -i --init ghcr.io/ssinad/dcvr-thesis/cycle-orienteering-app -D 2523 < tsp/gr48.tsp-gen-1.txt 2>/dev/null;
*hk48,48,5731,21,5707,34.17,30,1.429,1.627,1.139,Not metric,docker run -i --init ghcr.io/ssinad/dcvr-thesis/cycle-orienteering-app -D 5731 < tsp/hk48.tsp-gen-1.txt 2>/dev/null;
eil51,51,213,22,203.702,32.31,29,1.318,1.469,1.114,,docker run -i --init ghcr.io/ssinad/dcvr-thesis/cycle-orienteering-app -D 213 < tsp/eil51.tsp-gen-1.txt 2>/dev/null;
berlin52,52,3771,32,3623.51,42.34,37,1.156,1.323,1.144,,docker run -i --init ghcr.io/ssinad/dcvr-thesis/cycle-orienteering-app -D 3771 < tsp/berlin52.tsp-gen-1.txt 2>/dev/null;
*brazil58,58,12698,39,12670,52.87,46,1.179,1.356,1.149,Not metric,docker run -i --init ghcr.io/ssinad/dcvr-thesis/cycle-orienteering-app -D 12698 < tsp/brazil58.tsp-gen-1.txt 2>/dev/null;
st70,70,338,34,332.035,48.09,43,1.265,1.414,1.118,,docker run -i --init ghcr.io/ssinad/dcvr-thesis/cycle-orienteering-app -D 338 < tsp/st70.tsp-gen-1.txt 2>/dev/null;
eil76,76,269,36,261.349,49.94,47,1.306,1.387,1.063,,docker run -i --init ghcr.io/ssinad/dcvr-thesis/cycle-orienteering-app -D 269 < tsp/eil76.tsp-gen-1.txt 2>/dev/null;
pr76,76,54080,39,49731.1,56.06,49,1.256,1.437,1.144,,docker run -i --init ghcr.io/ssinad/dcvr-thesis/cycle-orienteering-app -D 54080 < tsp/pr76.tsp-gen-1.txt 2>/dev/null;
gr96,96,27605,47,23049.9,72.45,64,1.362,1.542,1.132,,docker run -i --init ghcr.io/ssinad/dcvr-thesis/cycle-orienteering-app -D 27605 < tsp/gr96.tsp-gen-1.txt 2>/dev/null;
rat99,99,606,38,593.788,58.27,52,1.368,1.533,1.12,,docker run -i --init ghcr.io/ssinad/dcvr-thesis/cycle-orienteering-app -D 606 < tsp/rat99.tsp-gen-1.txt 2>/dev/null;
kroA100,100,10641,46,10445.7,63.99,56,1.217,1.391,1.143,,docker run -i --init ghcr.io/ssinad/dcvr-thesis/cycle-orienteering-app -D 10641 < tsp/kroA100.tsp-gen-1.txt 2>/dev/null;
kroB100,100,11071,47,10500,65.66,58,1.234,1.397,1.132,,docker run -i --init ghcr.io/ssinad/dcvr-thesis/cycle-orienteering-app -D 11071 < tsp/kroB100.tsp-gen-1.txt 2>/dev/null;
kroC100,100,10375,41,10297.8,64.1,56,1.366,1.563,1.145,,docker run -i --init ghcr.io/ssinad/dcvr-thesis/cycle-orienteering-app -D 10375 < tsp/kroC100.tsp-gen-1.txt 2>/dev/null;
kroD100,100,10647,45,10236.5,67.59,59,1.311,1.502,1.146,,docker run -i --init ghcr.io/ssinad/dcvr-thesis/cycle-orienteering-app -D 10647 < tsp/kroD100.tsp-gen-1.txt 2>/dev/null;
kroE100,100,11034,45,10584.8,67.14,57,1.267,1.492,1.178,,docker run -i --init ghcr.io/ssinad/dcvr-thesis/cycle-orienteering-app -D 11034 < tsp/kroE100.tsp-gen-1.txt 2>/dev/null;
rd100,100,3955,47,3897.41,68.6,61,1.298,1.46,1.125,,docker run -i --init ghcr.io/ssinad/dcvr-thesis/cycle-orienteering-app -D 3955 < tsp/rd100.tsp-gen-1.txt 2>/dev/null;
eil101,101,315,50,295.023,69.47,64,1.28,1.389,1.086,,docker run -i --init ghcr.io/ssinad/dcvr-thesis/cycle-orienteering-app -D 315 < tsp/eil101.tsp-gen-1.txt 2>/dev/null;
lin105,105,7190,50,7017.25,77.58,66,1.32,1.552,1.175,,docker run -i --init ghcr.io/ssinad/dcvr-thesis/cycle-orienteering-app -D 7190 < tsp/lin105.tsp-gen-1.txt 2>/dev/null;
pr107,107,22152,49,22045.4,54,54,1.102,1.102,1,,docker run -i --init ghcr.io/ssinad/dcvr-thesis/cycle-orienteering-app -D 22152 < tsp/pr107.tsp-gen-1.txt 2>/dev/null;
*gr120,120,3471,54,3366,83.51,75,1.389,1.547,1.114,Not metric,docker run -i --init ghcr.io/ssinad/dcvr-thesis/cycle-orienteering-app -D 3471 < tsp/gr120.tsp-gen-1.txt 2>/dev/null;
pr124,124,29515,57,28860,89.32,75,1.316,1.567,1.191,,docker run -i --init ghcr.io/ssinad/dcvr-thesis/cycle-orienteering-app -D 29515 < tsp/pr124.tsp-gen-1.txt 2>/dev/null;
bier127,127,59141,88,55399.9,112.72,103,1.17,1.281,1.094,,docker run -i --init ghcr.io/ssinad/dcvr-thesis/cycle-orienteering-app -D 59141 < tsp/bier127.tsp-gen-1.txt 2>/dev/null;
pr136,136,48386,54,46567.2,78.03,71,1.315,1.445,1.099,,docker run -i --init ghcr.io/ssinad/dcvr-thesis/cycle-orienteering-app -D 48386 < tsp/pr136.tsp-gen-1.txt 2>/dev/null;
gr137,137,34927,61,34589.6,92.45,81,1.328,1.516,1.141,,docker run -i --init ghcr.io/ssinad/dcvr-thesis/cycle-orienteering-app -D 34927 < tsp/gr137.tsp-gen-1.txt 2>/dev/null;
pr144,144,29269,56,28261.2,93.55,77,1.375,1.67,1.215,,docker run -i --init ghcr.io/ssinad/dcvr-thesis/cycle-orienteering-app -D 29269 < tsp/pr144.tsp-gen-1.txt 2>/dev/null;
kroA150,150,13262,71,13197.5,97.37,86,1.211,1.371,1.132,,docker run -i --init ghcr.io/ssinad/dcvr-thesis/cycle-orienteering-app -D 13262 < tsp/kroA150.tsp-gen-1.txt 2>/dev/null;
kroB150,150,13065,63,11960,99.82,87,1.381,1.585,1.147,,docker run -i --init ghcr.io/ssinad/dcvr-thesis/cycle-orienteering-app -D 13065 < tsp/kroB150.tsp-gen-1.txt 2>/dev/null;
pr152,152,36841,61,36675.9,93.45,77,1.262,1.532,1.214,,docker run -i --init ghcr.io/ssinad/dcvr-thesis/cycle-orienteering-app -D 36841 < tsp/pr152.tsp-gen-1.txt 2>/dev/null;
u159,159,21040,73,20998.2,110.08,93,1.274,1.508,1.184,,docker run -i --init ghcr.io/ssinad/dcvr-thesis/cycle-orienteering-app -D 21040 < tsp/u159.tsp-gen-1.txt 2>/dev/null;
rat195,195,1162,73,1135.76,109.15,102,1.397,1.495,1.07,,docker run -i --init ghcr.io/ssinad/dcvr-thesis/cycle-orienteering-app -D 1162 < tsp/rat195.tsp-gen-1.txt 2>/dev/null;
d198,198,7890,91,7781.27,143.82,123,1.352,1.58,1.169,,docker run -i --init ghcr.io/ssinad/dcvr-thesis/cycle-orienteering-app -D 7890 < tsp/d198.tsp-gen-1.txt 2>/dev/null;
kroA200,200,14684,95,14469,131.82,117,1.232,1.388,1.127,,docker run -i --init ghcr.io/ssinad/dcvr-thesis/cycle-orienteering-app -D 14684 < tsp/kroA200.tsp-gen-1.txt 2>/dev/null;
kroB200,200,14719,89,14711.4,131.41,119,1.337,1.477,1.104,,docker run -i --init ghcr.io/ssinad/dcvr-thesis/cycle-orienteering-app -D 14719 < tsp/kroB200.tsp-gen-1.txt 2>/dev/null;
gr202,202,20080,116,19435.9,166.88,145,1.25,1.439,1.151,,docker run -i --init ghcr.io/ssinad/dcvr-thesis/cycle-orienteering-app -D 20080 < tsp/gr202.tsp-gen-1.txt 2>/dev/null;
ts225,225,63322,95,60251.4,127.64,124,1.305,1.344,1.029,,docker run -i --init ghcr.io/ssinad/dcvr-thesis/cycle-orienteering-app -D 63322 < tsp/ts225.tsp-gen-1.txt 2>/dev/null;
tsp225,225,1958,100,1937.96,144.03,129,1.29,1.44,1.116,,docker run -i --init ghcr.io/ssinad/dcvr-thesis/cycle-orienteering-app -D 1958 < tsp/tsp225.tsp-gen-1.txt 2>/dev/null;
pr226,226,40185,99,39674.7,150.3,126,1.273,1.518,1.193,,docker run -i --init ghcr.io/ssinad/dcvr-thesis/cycle-orienteering-app -D 40185 < tsp/pr226.tsp-gen-1.txt 2>/dev/null;
gr229,229,67301,143,62333.7,189.37,176,1.231,1.324,1.076,,docker run -i --init ghcr.io/ssinad/dcvr-thesis/cycle-orienteering-app -D 67301 < tsp/gr229.tsp-gen-1.txt 2>/dev/null;
gil262,262,1189,118,1143.57,172.19,158,1.339,1.459,1.09,,docker run -i --init ghcr.io/ssinad/dcvr-thesis/cycle-orienteering-app -D 1189 < tsp/gil262.tsp-gen-1.txt 2>/dev/null;
pr264,264,24568,106,24395.1,132,132,1.245,1.245,1,,docker run -i --init ghcr.io/ssinad/dcvr-thesis/cycle-orienteering-app -D 24568 < tsp/pr264.tsp-gen-1.txt 2>/dev/null;
a280,280,1290,115,1172.17,152.93,147,1.278,1.33,1.04,,docker run -i --init ghcr.io/ssinad/dcvr-thesis/cycle-orienteering-app -D 1290 < tsp/a280.tsp-gen-1.txt 2>/dev/null;
pr299,299,24096,119,22279.1,188.78,162,1.361,1.586,1.165,,docker run -i --init ghcr.io/ssinad/dcvr-thesis/cycle-orienteering-app -D 24096 < tsp/pr299.tsp-gen-1.txt 2>/dev/null;
lin318,318,21045,150,18899.5,228.33,205,1.367,1.522,1.114,,docker run -i --init ghcr.io/ssinad/dcvr-thesis/cycle-orienteering-app -D 21045 < tsp/lin318.tsp-gen-1.txt 2>/dev/null;
rd400,400,7641,181,7571.68,257.93,239,1.32,1.425,1.079,,docker run -i --init ghcr.io/ssinad/dcvr-thesis/cycle-orienteering-app -D 7641 < tsp/rd400.tsp-gen-1.txt 2>/dev/null;
\end{filecontents*}
\csvreader[head to column names]{csv/cycles-gen1.csv}{}
    {\\\hline\csvcoli & \csvcoliii & \csvcoliv & \csvcolvii & \csvcolvi & \csvcolviii  & \csvcolx &\csvcolix}
    \end{tabular}
}
\end{center}
\end{table}

\begin{table}[t!]
\begin{center}
{\small
    \caption{Cycle orienteering results - Gen 3 \label{cyc-gen3}}

    \medskip\medskip
    \begin{tabular}{l|l|l|l|l|l|l|l}%
    \bfseries Dataset & {\boldmath $B$} & {\boldmath $\val$} & {\boldmath $\opt$} &
    {\boldmath $\ub$} & {\boldmath $\opt/\val$} & {\boldmath $\ub/\opt$} & {\boldmath $\ub/\val$} \\\hline
    \hline \vspace{-4.325mm}
    \begin{filecontents*}{csv/cycles-gen3.csv}
Dataset,of nodes,Distance Limit,Orienteering,Distance,Upper bound,opt-val,Approximation factor,Upper bound / val,Upper bound / OPT,Comment,,,,,,,,,
*att48,48,5314,930,5217,1234.92,1049,1.128,1.328,1.177,Rounded,docker run -i --init ghcr.io/ssinad/dcvr-thesis/cycle-orienteering-app -D 5314 < tsp/att48.tsp-gen-3.txt 2>/dev/null;,,,,,,,,1236.66
*gr48,48,2523,1134,2523,1782.62,1480,1.305,1.572,1.204,Not metric,docker run -i --init ghcr.io/ssinad/dcvr-thesis/cycle-orienteering-app -D 2523 < tsp/gr48.tsp-gen-3.txt 2>/dev/null;,,,,,,,,1782.62
*hk48,48,5731,1265,5631,2039.93,1764,1.394,1.613,1.156,Not metric,docker run -i --init ghcr.io/ssinad/dcvr-thesis/cycle-orienteering-app -D 5731 < tsp/hk48.tsp-gen-3.txt 2>/dev/null;,,,,,,,,2039.93
eil51,51,213,1127,205.959,1638.23,1399,1.241,1.454,1.171,,docker run -i --init ghcr.io/ssinad/dcvr-thesis/cycle-orienteering-app -D 213 < tsp/eil51.tsp-gen-3.txt 2>/dev/null;,,,,,,,,1638.23
berlin52,52,3771,826,3541.54,1334.5,1036,1.254,1.616,1.288,,docker run -i --init ghcr.io/ssinad/dcvr-thesis/cycle-orienteering-app -D 3771 < tsp/berlin52.tsp-gen-3.txt 2>/dev/null;,,,,,,,,
*brazil58,58,12698,1473,12084,1902.32,1702,1.155,1.291,1.118,Not metric,docker run -i --init ghcr.io/ssinad/dcvr-thesis/cycle-orienteering-app -D 12698 < tsp/brazil58.tsp-gen-3.txt 2>/dev/null;,,,,,,,,1902.32
st70,70,338,1601,324.336,2440.8,2108,1.317,1.525,1.158,,docker run -i --init ghcr.io/ssinad/dcvr-thesis/cycle-orienteering-app -D 338 < tsp/st70.tsp-gen-3.txt 2>/dev/null;,,,,,,,,2440.8
eil76,76,269,1704,260.043,2721.6,2467,1.448,1.597,1.103,,docker run -i --init ghcr.io/ssinad/dcvr-thesis/cycle-orienteering-app -D 269 < tsp/eil76.tsp-gen-3.txt 2>/dev/null;,,,,,,,,2721.6
pr76,76,54080,1890,48971.8,2868.22,2430,1.286,1.518,1.18,,docker run -i --init ghcr.io/ssinad/dcvr-thesis/cycle-orienteering-app -D 54080 < tsp/pr76.tsp-gen-3.txt 2>/dev/null;,,,,,,,,2868.22
gr96,96,27605,2283,27595.5,3624.93,3170,1.389,1.588,1.144,,docker run -i --init ghcr.io/ssinad/dcvr-thesis/cycle-orienteering-app -D 27605 < tsp/gr96.tsp-gen-3.txt 2>/dev/null;,,,,,,,,3624.93
rat99,99,606,1939,604.773,3257.75,2908,1.5,1.68,1.12,,docker run -i --init ghcr.io/ssinad/dcvr-thesis/cycle-orienteering-app -D 606 < tsp/rat99.tsp-gen-3.txt 2>/dev/null;,,,,,,,,3257.75
kroA100,100,10641,2578,10094.7,3642.86,3211,1.246,1.413,1.134,,docker run -i --init ghcr.io/ssinad/dcvr-thesis/cycle-orienteering-app -D 10641 < tsp/kroA100.tsp-gen-3.txt 2>/dev/null;,,,,,,,,3642.86
kroB100,100,11071,2136,10851.7,3284.42,2804,1.313,1.538,1.171,,docker run -i --init ghcr.io/ssinad/dcvr-thesis/cycle-orienteering-app -D 11071 < tsp/kroB100.tsp-gen-3.txt 2>/dev/null;,,,,,,,,3284.42
kroC100,100,10375,2610,9711.66,3759.26,3155,1.209,1.44,1.192,,docker run -i --init ghcr.io/ssinad/dcvr-thesis/cycle-orienteering-app -D 10375 < tsp/kroC100.tsp-gen-3.txt 2>/dev/null;,,,,,,,,3759.26
kroD100,100,10647,2272,10410.2,3834.63,3167,1.394,1.688,1.211,,docker run -i --init ghcr.io/ssinad/dcvr-thesis/cycle-orienteering-app -D 10647 < tsp/kroD100.tsp-gen-3.txt 2>/dev/null;,,,,,,,,3834.63
kroE100,100,11034,2154,10354,3440.24,3049,1.416,1.597,1.128,,docker run -i --init ghcr.io/ssinad/dcvr-thesis/cycle-orienteering-app -D 11034 < tsp/kroE100.tsp-gen-3.txt 2>/dev/null;,,,,,,,,3440.24
rd100,100,3955,2291,3723.27,3546.66,2926,1.277,1.548,1.212,,docker run -i --init ghcr.io/ssinad/dcvr-thesis/cycle-orienteering-app -D 3955 < tsp/rd100.tsp-gen-3.txt 2>/dev/null;,,,,,,,,3546.66
eil101,101,315,2563,313.495,3674.77,3345,1.305,1.434,1.099,,docker run -i --init ghcr.io/ssinad/dcvr-thesis/cycle-orienteering-app -D 315 < tsp/eil101.tsp-gen-3.txt 2>/dev/null;,,,,,,,,3674.77
lin105,105,7190,2401,7177.02,3888.52,2986,1.244,1.62,1.302,,docker run -i --init ghcr.io/ssinad/dcvr-thesis/cycle-orienteering-app -D 7190 < tsp/lin105.tsp-gen-3.txt 2>/dev/null;,,,,,,,,3888.52
pr107,107,22152,1702,22018.3,2447.62,1877,1.103,1.438,1.304,,docker run -i --init ghcr.io/ssinad/dcvr-thesis/cycle-orienteering-app -D 22152 < tsp/pr107.tsp-gen-3.txt 2>/dev/null;,,,,,,,,2447.62
*gr120,120,3471,2995,3449,4479.38,3779,1.262,1.496,1.185,Not metric,docker run -i --init ghcr.io/ssinad/dcvr-thesis/cycle-orienteering-app -D 3471 < tsp/gr120.tsp-gen-3.txt 2>/dev/null;,,,,,,,,4479.38
pr124,124,29515,3193,29510,4303.36,3557,1.114,1.348,1.21,,docker run -i --init ghcr.io/ssinad/dcvr-thesis/cycle-orienteering-app -D 29515 < tsp/pr124.tsp-gen-3.txt 2>/dev/null;,,,,,,,,4303.36
bier127,127,59141,1852,59108.1,2769.98,2365,1.277,1.496,1.171,,docker run -i --init ghcr.io/ssinad/dcvr-thesis/cycle-orienteering-app -D 59141 < tsp/bier127.tsp-gen-3.txt 2>/dev/null;,,,,,,,,2769.98
pr136,136,48386,3068,48139.6,4994,4390,1.431,1.628,1.138,,docker run -i --init ghcr.io/ssinad/dcvr-thesis/cycle-orienteering-app -D 48386 < tsp/pr136.tsp-gen-3.txt 2>/dev/null;,,,,,,,,4994
gr137,137,34927,2759,33533.3,4382.23,3954,1.433,1.588,1.108,,docker run -i --init ghcr.io/ssinad/dcvr-thesis/cycle-orienteering-app -D 34927 < tsp/gr137.tsp-gen-3.txt 2>/dev/null;,,,,,,,,4382.23
pr144,144,29269,2859,28434,4603.35,3745,1.31,1.61,1.229,,docker run -i --init ghcr.io/ssinad/dcvr-thesis/cycle-orienteering-app -D 29269 < tsp/pr144.tsp-gen-3.txt 2>/dev/null;,,,,,,,,4603.35
kroA150,150,13262,4174,11958.6,5645.92,5039,1.207,1.353,1.12,,docker run -i --init ghcr.io/ssinad/dcvr-thesis/cycle-orienteering-app -D 13262 < tsp/kroA150.tsp-gen-3.txt 2>/dev/null;,,,,,,,,5645.92
kroB150,150,13065,4147,11895.4,6090.14,5314,1.281,1.469,1.146,,docker run -i --init ghcr.io/ssinad/dcvr-thesis/cycle-orienteering-app -D 13065 < tsp/kroB150.tsp-gen-3.txt 2>/dev/null;,,,,,,,,6090.14
pr152,152,36841,2960,36630.7,4799.2,3905,1.319,1.621,1.229,,docker run -i --init ghcr.io/ssinad/dcvr-thesis/cycle-orienteering-app -D 36841 < tsp/pr152.tsp-gen-3.txt 2>/dev/null;,,,,,,,,4799.2
u159,159,21040,4163,20802,6260.51,5272,1.266,1.504,1.188,,docker run -i --init ghcr.io/ssinad/dcvr-thesis/cycle-orienteering-app -D 21040 < tsp/u159.tsp-gen-3.txt 2>/dev/null;,,,,,,,,6260.51
rat195,195,1162,4384,1144.93,6851.98,6195,1.413,1.563,1.106,,docker run -i --init ghcr.io/ssinad/dcvr-thesis/cycle-orienteering-app -D 1162 < tsp/rat195.tsp-gen-3.txt 2>/dev/null;,,,,,,,,6851.98
d198,198,7890,4657,7879.62,7715.83,6320,1.357,1.657,1.221,,docker run -i --init ghcr.io/ssinad/dcvr-thesis/cycle-orienteering-app -D 7890 < tsp/d198.tsp-gen-3.txt 2>/dev/null;,,,,,,,,7715.83
kroA200,200,14684,4781,14639.6,6981.26,6123,1.281,1.46,1.14,,docker run -i --init ghcr.io/ssinad/dcvr-thesis/cycle-orienteering-app -D 14684 < tsp/kroA200.tsp-gen-3.txt 2>/dev/null;,,,,,,,,6981.26
kroB200,200,14719,4902,14223.6,6910.18,6266,1.278,1.41,1.103,,docker run -i --init ghcr.io/ssinad/dcvr-thesis/cycle-orienteering-app -D 14719 < tsp/kroB200.tsp-gen-3.txt 2>/dev/null;,,,,,,,,6910.18
gr202,202,20080,6651,20062.9,10048.9,8616,1.295,1.511,1.166,,docker run -i --init ghcr.io/ssinad/dcvr-thesis/cycle-orienteering-app -D 20080 < tsp/gr202.tsp-gen-3.txt 2>/dev/null;,,,,,,,,10048.9
ts225,225,63322,5642,58793,8362.16,7575,1.343,1.482,1.104,,docker run -i --init ghcr.io/ssinad/dcvr-thesis/cycle-orienteering-app -D 63322 < tsp/ts225.tsp-gen-3.txt 2>/dev/null;,,,,,,,,8362.16
tsp225,225,1958,6075,1791.02,9014.54,7740,1.274,1.484,1.165,,docker run -i --init ghcr.io/ssinad/dcvr-thesis/cycle-orienteering-app -D 1958 < tsp/tsp225.tsp-gen-3.txt 2>/dev/null;,,,,,,,,
pr226,226,40185,5638,39760.5,8687.94,6993,1.24,1.541,1.242,,docker run -i --init ghcr.io/ssinad/dcvr-thesis/cycle-orienteering-app -D 40185 < tsp/pr226.tsp-gen-3.txt 2>/dev/null;,,,,,,,,8687.94
gr229,229,67301,5179,63430.2,7120.19,6328,1.222,1.375,1.125,,docker run -i --init ghcr.io/ssinad/dcvr-thesis/cycle-orienteering-app -D 67301 < tsp/gr229.tsp-gen-3.txt 2>/dev/null;,,,,,,,,4
gil262,262,1189,7335,1129.78,10698.1,9246,1.261,1.459,1.157,,docker run -i --init ghcr.io/ssinad/dcvr-thesis/cycle-orienteering-app -D 1189 < tsp/gil262.tsp-gen-3.txt 2>/dev/null;,,,,,,,,10698.1
pr264,264,24568,5726,24320.5,8846.52,8137,1.421,1.545,1.087,,docker run -i --init ghcr.io/ssinad/dcvr-thesis/cycle-orienteering-app -D 24568 < tsp/pr264.tsp-gen-3.txt 2>/dev/null;,,,,,,,,8846.52
a280,280,1290,7036,1288.18,10600.8,9774,1.389,1.507,1.085,,docker run -i --init ghcr.io/ssinad/dcvr-thesis/cycle-orienteering-app -D 1290 < tsp/a280.tsp-gen-3.txt 2>/dev/null;,,,,,,,,
pr299,299,24096,7870,23945.6,11959.9,10343,1.314,1.52,1.156,,docker run -i --init ghcr.io/ssinad/dcvr-thesis/cycle-orienteering-app -D 24096 < tsp/pr299.tsp-gen-3.txt 2>/dev/null;,,,,,,,,11959.9
lin318,318,21045,7638,19875.2,11771.6,10368,1.357,1.541,1.135,,docker run -i --init ghcr.io/ssinad/dcvr-thesis/cycle-orienteering-app -D 21045 < tsp/lin318.tsp-gen-3.txt 2>/dev/null;,,,,,,,,11771.6
rd400,400,7641,11216,7444.17,14445.5,13223,1.179,1.288,1.092,,docker run -i --init ghcr.io/ssinad/dcvr-thesis/cycle-orienteering-app -D 7641 < tsp/rd400.tsp-gen-3.txt 2>/dev/null;,,,,,,,,14445.5
\end{filecontents*}
\csvreader[head to column names]{csv/cycles-gen3.csv}{}
    {\\\hline\csvcoli & \csvcoliii & \csvcoliv & \csvcolvii & \csvcolvi & \csvcolviii  & \csvcolx &\csvcolix}
    \end{tabular}
}
\end{center}
\end{table}

\clearpage
\subsection{P2P orienteering experiments} \label{ptp-compres}
For each $w$, we run two binary searches: one to search for an $r$-rooted PCA solution
with budget $B-c_{wt}$, and the other to search for a $t$-rooted PCA solution with budget
$B-c_{rw}$. Let $\Lambda^w_r, \Lambda^w_t$ denote the two sets of $\lambda$ values
encountered during these searches. The solution we output is the best P2P-orienteering
solution extracted from all the trees we find during the binary search procedures. 
We compute an upper bound in the same way as $\UBptp(B)$, except that we only consider the
$\ld$ values encountered in the binary search procedures.
Specifically, we return the upper bound 
\begin{multline*}
\bUBptp(B):=\max_{w\in V:c_{rw}+c_{wt}\leq B}\ \min\,\Bigl\{
\pi(\bV^{\ptp}_w)+
\min\,\bigl\{\min_{\ld_r \in \Lambda_r}\tfrac{B-\dualub_r(\ld_r)}{\ld_r},
\min_{\ld_t \in \Lambda_t}\tfrac{B-\dualub_t(\ld_t)}{\ld_t}\bigr\}, \\
\Bigl(\min_{\ld\in\Ld^w_r}\,\UBptprw(w,B;\ld)
+\min_{\ld\in\Ld^w_t}\,\UBptpwt(w,B;\ld)\Bigr)\Bigr\}
\end{multline*}

\vspace*{-2ex}
\paragraph{Test Data.} 
We use two datasets that have previously been used for evaluating P2P-orienteering
algorithms. Both involve Euclidean instances with various distance bounds.
The first dataset is from~\cite{Tsili84} (see Table~\ref{ptp-tsil}), and includes Euclidean
instances with 32, 21, and 33 nodes, with node rewards that are integer multiples of
$5$ in the range $[5, 40]$. 
The optimal values for these instances are reported in \cite{Schilde09}. 

The second dataset is from~\cite{Chao96} (see Table~\ref{ptp-chao}) and includes two
Euclidean instances.  
The first instance has 66 nodes that
are distributed among four concentric squares. The nodes take rewards from
$\{5,15,25,35\}$ depending on the square they lie on (smaller rewards for smaller
squares). The start and end nodes are located near each other in the middle of the
innermost square. 
The second instance has 64 nodes that are distributed in a diamond shape. The start and end
nodes are at the top and bottom of the diamond, respectively. Nodes further away from the
start and end nodes get higher rewards. The rewards are integer multiples of 6 between 6 and
42. 

In total, 89 different experiments were run on the above instances using varying distance
bounds. 
Figure~\ref{ptp-stats} presents a summary of our results; the detailed table of results is
included in Appendix \ref{sec:tables} (see Tables~\ref{ptp-tsil} and~\ref{ptp-chao}). 
The entries 
that refer to $\opt$ (i.e., two histograms and two table columns) 
only consider the 49 experiments using data from \cite{Tsili84}, where we do know the 
optimal values. 
The entries referring to $\ub/\val$ involve all the P2P instances considered.

\begin{figure}[ht!]
\begin{center}
\begin{tabular}{|c|c|c|c|c|c| }\hline
 No. of nodes & \# Experiments & Mean $\opt/\val$ & Max $\opt/\val$ & Mean $\ub/\val$ &
 Max $\ub/\val$ \\ \hline 
 21 -- 66  & 89 & 1.169 & 1.364 & 1.371 & 2.164 \\ \hline
\end{tabular}

\bigskip
\begin{tikzpicture}
\begin{axis}[ybar interval, width=9cm, height = 5cm, ymax=20, ymin=0, 
    minor y tick num = 3, tick label style={font=\scriptsize}, 
    extra x ticks={1.364}, 
    extra x tick style={x tick label as interval=false, xticklabel style={anchor=north}},
    xlabel={$\opt/\val$ - P2P orienteering}, xticklabel style={anchor=north east}]
\addplot+ [hist={bins=10}] table [y=p2p-apx] {\histogramtable};
\end{axis}
\end{tikzpicture}
\begin{tikzpicture}
\begin{axis}[ybar interval, width=9cm, height = 5cm, ymax=20, ymin=0, 
    minor y tick num = 3, tick label style={font=\scriptsize}, 
    extra x ticks={2.1635}, 
    extra x tick style={x tick label as interval=false, xticklabel style={anchor=north}},
    xlabel={$\ub/\opt$ - P2P orienteering}, xticklabel style={anchor=north east}]
\addplot+ [hist={bins=10}] table [y=p2p-ubopt] {\histogramtable};
\end{axis}
\end{tikzpicture}


\bigskip
\begin{tikzpicture}
\begin{axis}[ybar interval, width=9cm, height = 5cm, ymax=30, ymin=0, 
    minor y tick num = 3, tick label style={font=\scriptsize}, 
    extra x ticks={2.164}, 
    extra x tick style={x tick label as interval=false, xticklabel style={anchor=north}},
    xlabel={$\ub/\val$ - P2P orienteering}, xticklabel style={anchor=north east}]
\addplot+ [hist={bins=10}] table [y=p2p-ubval] {\histogramtable};
\end{axis}
\end{tikzpicture}

\end{center}

\vspace*{-2ex}
\caption{P2P orienteering statistics. 
\label{ptp-stats}}
\end{figure}
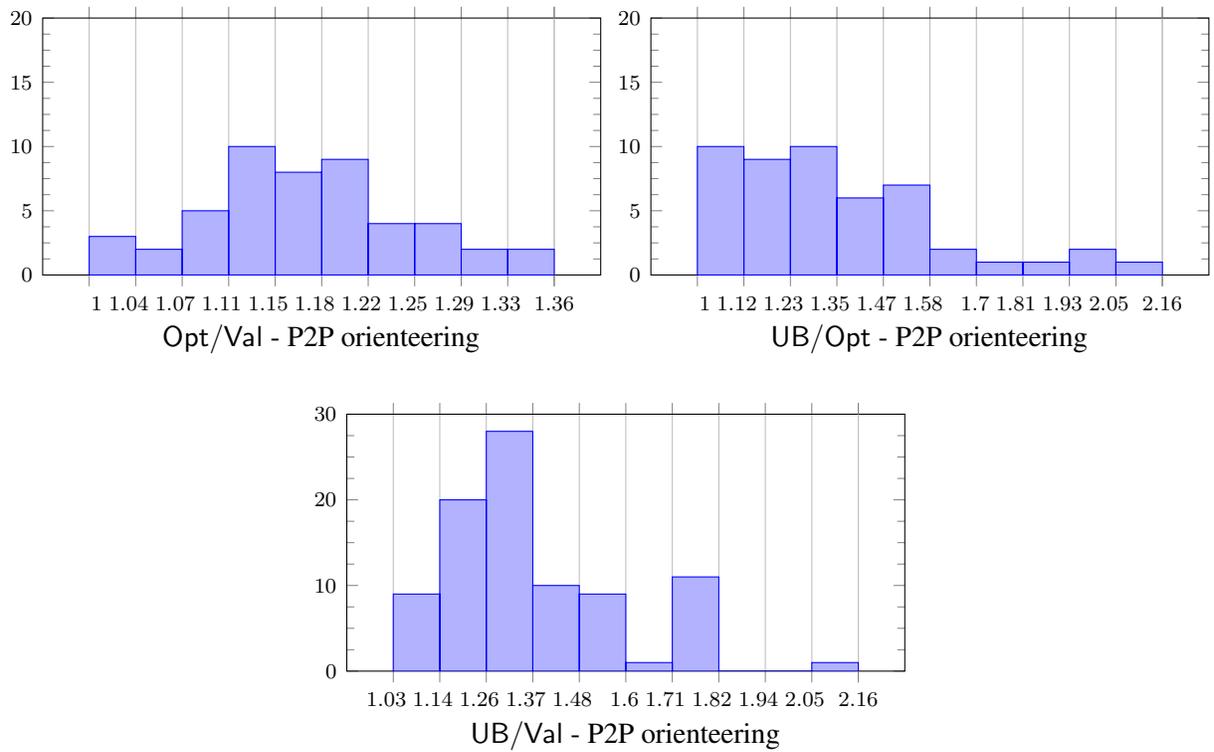


\clearpage

\subsection{Rooted orienteering experiments} \label{rooted-compres}
For each guess $w$ of furthest node, we run binary search with target budget $B$. 
We extract a rooted-orienteering solution from each tree that we
find in the binary search, and return the best of these solutions.
Our computed upper bound is $\UBrooted(B)$, but considering only the $\ld$ values
encountered in the binary search procedures.

\vspace*{-1ex}
\paragraph{Test Data.} 
We are not aware of any benchmark data for rooted orienteering, so we use the benchmark
data as for cycle orienteering. That is, we consider every cycle-orienteering
instance here as well, with the same distance bound. 
While we do not have the optimal value for any instance, 
we can still compute $\ub/\val$, which yields an upper bound on the approximation
ratio for the instance.
We continue to use the non-metric instances from TSPLIB as input to our rooted
orienteering experiments, just like we did with cycle orienteering. 

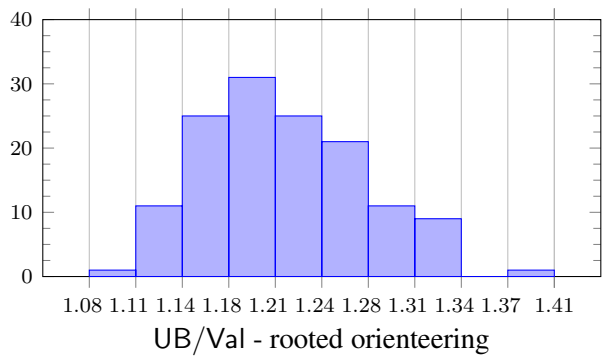
\begin{figure}[ht!]
\begin{center}
\begin{tabular}{|c|c|c|c|c| }\hline
No. of nodes &  \# Experiments & Mean UB / Val & Max UB / Val \\ \hline
48--400  & 135 & 1.216 & 1.407 \\ \hline
\end{tabular}


\bigskip
\begin{tikzpicture}
\begin{axis}[ybar interval, width=9cm, height = 5cm, ymax=40, ymin=0, 
    minor y tick num = 3, tick label style={font=\scriptsize}, 
    extra x ticks={1.407}, 
    extra x tick style={x tick label as interval=false, xticklabel style={anchor=north}},
    xlabel={$\ub/\val$ - rooted orienteering}, xticklabel style={anchor=north east}]
\addplot+ [hist={bins=10}] table [y=rooted-ubval] {\histogramtable};
\end{axis}
\end{tikzpicture}
\end{center}

\vspace*{-1ex}
\caption{Rooted orienteering statistics.} 
\end{figure}

\begin{table}[p!]
\begin{center}
{\small
    \caption{Rooted orienteering results \label{rooted-gen1gen2gen3}}

    \medskip\medskip
    \begin{tabular}{l|l||l|l|l||l|l|l||l|l|l}%
    \multirow{2}{*}{\begin{tabular}{l} {\bf Dataset} \end{tabular}} &
    \multirow{2}{*}{\begin{tabular}{l} \boldmath $B$ \end{tabular}} & 
    \multicolumn{3}{c||}{{\bf Gen 1}} & \multicolumn{3}{c||}{{\bf Gen 2}} & 
    \multicolumn{3}{c}{{\bf Gen 3}} \\ \cline{3-11}
    & & {\boldmath $\val$} & {\boldmath $\ub$} & {\boldmath $\ub/\val$} & 
    {\boldmath $\val$} & {\boldmath $\ub$} & {\boldmath $\ub/\val$} & 
    {\boldmath $\val$} & {\boldmath $\ub$} & {\boldmath $\ub/\val$} \\ 
    \hline \vspace{-4.225mm}
    \begin{filecontents*}{csv/rooted_gen1gen2gen3.csv}
Dataset, of nodes,Distance Limit,Rooted Orienteering,Distance,Upper bound,Upper bound / val,Comment,GEN 2:,Rooted Orienteering,Distance,Upper bound,Upper bound / val,Comment,GEN 3:,Rooted Orienteering,Distance,Upper bound,Upper bound / val,Comment
*att48,48,5314,30,5273,35.87,1.196,Rounded,,1546,5306,1900.14,1.229,Rounded,,1175,5240,1482.59,1.262,Rounded
*gr48,48,2523,27,2514,35.42,1.312,Not metric,,1616,2431,2035.82,1.26,Not metric,,1554,2506,1782.62,1.147,Not metric
*hk48,48,5731,29,5689,33.83,1.167,Not metric,,1483,5436,1875.09,1.264,Not metric,,1449,5700,2039.17,1.407,Not metric
eil51,51,213,26,209.009,32.18,1.238,,,1455,197.374,1875.55,1.289,,,1367,210.586,1613.83,1.181,
berlin52,52,3771,34,3589.4,42.18,1.241,,,1815,3490.44,2200.52,1.212,,,1036,3767.07,1333.1,1.287,
*brazil58,58,12698,47,12652,52.53,1.118,Not metric,,2225,12437,2620.01,1.178,Not metric,,1778,12551,2002.44,1.126,Not metric
st70,70,338,41,322.187,47.36,1.155,,,2317,326.292,2658.29,1.147,,,2267,324.413,2637.92,1.164,
eil76,76,269,44,268.665,49.94,1.135,,,2250,263.373,2841.78,1.263,,,2179,265.6,2715.55,1.246,
pr76,76,54080,48,53311.1,56.15,1.17,,,2489,53790.4,3017.23,1.212,,,2406,53989.5,2883.85,1.199,
gr96,96,27605,56,27532.4,72.23,1.29,,,3126,27470.9,3807.18,1.218,,,3132,27520.6,3621.27,1.156,
rat99,99,606,51,597.609,58.56,1.148,,,2862,603.197,3323.17,1.161,,,2838,604.257,3670.18,1.293,
kroA100,100,10641,54,10619.5,63.96,1.184,,,3057,10529.9,3635.36,1.189,,,2881,10618.6,3640.52,1.264,
kroB100,100,11071,53,10962.5,65.66,1.239,,,3106,11043.8,3627.93,1.168,,,2642,9736.9,3208.23,1.214,
kroC100,100,10375,50,10368.9,64.03,1.281,,,2965,10361.2,3338.71,1.126,,,3355,10088.6,3758.62,1.12,
kroD100,100,10647,56,10576.5,67.22,1.2,,,3110,10506.4,3827.83,1.231,,,3101,10400.4,3923.96,1.265,
kroE100,100,11034,55,10874,66.14,1.203,,,3017,10829.6,3635.8,1.205,,,2779,10906.3,3675.43,1.323,
rd100,100,3955,55,3947.14,68.6,1.247,,,3150,3789.34,3746.51,1.189,,,2990,3929.96,3544.12,1.185,
eil101,101,315,56,313.12,69.23,1.236,,,3171,312.945,3943.33,1.244,,,2851,310.071,3665.5,1.286,
lin105,105,7190,69,7111.99,77.54,1.124,,,3575,7184.88,3970.13,1.111,,,3254,7114.44,3889.56,1.195,
pr107,107,22152,54,20130.4,62.13,1.15,,,2839,22049.9,3597.88,1.267,,,3909,21960.7,4533.29,1.16,
*gr120,120,3471,64,3457,83.51,1.305,Not metric,,3717,3467,4819.29,1.297,Not metric,,3618,3449,4464.9,1.234,Not metric
pr124,124,29515,74,29335.1,89.24,1.206,,,3647,29354.5,4596.56,1.26,,,4440,28604.5,4785.48,1.078,
bier127,127,59141,95,57866.3,112.7,1.186,,,4943,58793.7,5873.66,1.188,,,2060,58357.1,2760.56,1.34,
pr136,136,48386,65,47730.3,78.03,1.2,,,3922,47997.5,4849.45,1.236,,,4222,48352.7,5073.17,1.202,
gr137,137,34927,80,34185.4,94.38,1.18,,,4425,34778,5092.28,1.151,,,4021,34869.1,5100.05,1.268,
pr144,144,29269,80,28800.7,93.55,1.169,,,4294,29197.7,4818.47,1.122,,,4296,29074.9,5736.39,1.335,
kroA150,150,13262,85,13240.4,97.22,1.144,,,4856,13250.5,5506.87,1.134,,,4778,13220,5644.87,1.181,
kroB150,150,13065,75,13014.3,99.8,1.331,,,4307,12660.6,5382.69,1.25,,,4982,13063,6088.58,1.222,
pr152,152,36841,81,35868.1,101.24,1.25,,,4524,36809.2,5375.38,1.188,,,5442,36835.8,6226.95,1.144,
u159,159,21040,90,20758.6,109.44,1.216,,,5050,20935.4,5810.4,1.151,,,5236,20719.2,6195.74,1.183,
rat195,195,1162,88,1161.64,109.34,1.243,,,4969,1153.2,6098.92,1.227,,,5431,1161.19,7088.95,1.305,
d198,198,7890,141,7712.94,165.35,1.173,,,7445,7805.39,8469.94,1.138,,,7389,7793.01,8583.75,1.162,
kroA200,200,14684,106,14554.2,131.66,1.242,,,5406,14596.5,7228.18,1.337,,,5798,14656,6980.88,1.204,
kroB200,200,14719,106,14702,131.37,1.239,,,5791,14679.7,7072.28,1.221,,,5884,14636.1,6906.16,1.174,
gr202,202,20080,135,20060.6,166.7,1.235,,,7537,20030.1,8787.99,1.166,,,8026,20063.9,10014.7,1.248,
ts225,225,63322,111,63234.8,127.64,1.15,,,5922,62925.5,7311.91,1.235,,,7001,63248.7,8549,1.221,
tsp225,225,1958,123,1955.08,144.25,1.173,,,6588,1949.39,7920.81,1.202,,,7336,1949.05,8996.21,1.226,
pr226,226,40185,143,40090.8,173.72,1.215,,,7555,40183.2,8925.1,1.181,,,9005,38998.7,10211.7,1.134,
gr229,229,67301,160,67233,189.24,1.183,,,8258,67157.6,9768.7,1.183,,,5905,66977.6,7119.57,1.206,
gil262,262,1189,141,1179.18,169.93,1.205,,,7109,1184.59,9262.61,1.303,,,8125,1184.71,10698.1,1.317,
pr264,264,24568,130,24466.8,165.24,1.271,,,7456,24546.3,9150.09,1.227,,,8698,24405.8,11255,1.294,
a280,280,1290,129,1285.4,154.79,1.2,,,7511,1286.68,9014.64,1.2,,,9384,1287.47,11030.8,1.175,
pr299,299,24096,152,23978.4,189.98,1.25,,,8565,24093.2,10835.5,1.265,,,10164,24047.9,12144.4,1.195,
lin318,318,21045,171,20879.2,228.08,1.334,,,9783,21035,12180.8,1.245,,,9397,20769.7,11760.9,1.252,
rd400,400,7641,194,7588.45,257.93,1.33,,,11973,7634.22,14635.7,1.222,,,12304,7636.74,14445.1,1.174,
\end{filecontents*}
\csvreader[head to column names]{csv/rooted_gen1gen2gen3.csv}{}
    {\\\hline \csvcoli & \csvcoliii & \csvcoliv & \csvcolvi & \csvcolvii &
    \csvcolx & \csvcolxii & \csvcolxiii & \csvcolxvi & \csvcolxviii & \csvcolxix}
    \end{tabular}
}
\end{center}
\end{table}

\subsection{Observations}
On all datasets, we see that the instance-wise approximation ratio---$\frac{\opt}{\val}$,
if $\opt$ is available, or $\frac{\ub}{\opt}$ otherwise---%
much smaller than the worst-case approximation guarantee of the corresponding
algorithm. 
While this could be due to the fact we are dealing with real-world or Euclidean
data, it is also possible that our worst-case analysis is not tight. We do have
examples showing that the ratio $\frac{\ub}{\opt}$ (and hence $\frac{\ub}{\val}$) for
rooted orienteering can be as large as $3$ (so Theorem~\ref{lpboundsthm} (b) is tight),
but we do not know if a better worst-case bound can be proved for $\frac{\opt}{\val}$ (or
$\rorientlpopt/\val$, where $\rorientlpopt$ is the optimal value of the LP-relaxation for
orienteering proposed in~\cite{FriggstadS17}).
%
Another noteworthy aspect to emerge from our results is that, despite the worst-case
bound on $\frac{\ub}{\opt}$ mentioned above, our
computed upper bounds are often quite close to the optimal value. 

\appendix

\section{Avoiding the \boldmath $(1-\e)$-factor loss in Theorem~\ref{pcabsearch}} 
\label{append-avoidloss} 
Recall that the $c_e$s are integers, and $c_{\max}=\max_{u,v\in N}$. 
Let $M$ be a polytime computable integer
with $\log M=\poly(\text{input size})$ such that all the $\pi_v$s are integer multiples of
$1/M$; therefore, the optimal reward $\ropt$ (among rooted paths $P$ with 
$\{w\}\sse V(P)\sse N$) is also a multiple of $1/M$. Recall that $\binlb=1/\pi(N')$, where
$N'=N\sm\{r\}$. 

We can avoid the $(1-\e)$-factor loss in part (c) of Theorem~\ref{pcabsearch} by
continuing the binary search until $\ld_2-\ld_1\leq\frac{\binlb^2}{Mc_{\max}}$. Then, by
the analysis in Theorem~\ref{pcabsearch}, we have that
\begin{equation}
a\cdot\pi(V(T_{\ld_1}))+b\cdot\pi(V(T_{\ld_2}))-\ropt\geq-\frac{a(\ld_2-\ld_1)\pi(N')}{\ld_2}
\geq-\frac{\binlb^2\cdot\pi(N')}{\ld_2Mc_{\max}}>-\frac{1}{Mc_{\max}}
\label{avoidineq}
\end{equation}
where the final inequality is because $\ld_2>\binlb$ and $\binlb\cdot\pi(N')=1$. Now 
$a$ and $b$ are rational numbers whose denominators are bounded by $c_{\max}$. So the LHS
of \eqref{avoidineq} is an integer multiple of $\frac{1}{c'M}$ for some $c'\leq c_{\max}$. 
So if the LHS is strictly larger than $-\frac{1}{Mc_{max}}$ then it is in fact
nonnegative. 
\hfill \qed

\section{LP-relative bounds for orienteering} \label{append-lpbounds}
We show that the guarantees returned by our combinatorial algorithms for orienteering in
Section~\ref{orient-apps}---both the quality of the solution returned, and the upper
bounds---hold with respect to the optimal value of the LP-relaxation for orienteering
proposed by~\cite{FriggstadS17}. We focus on rooted orienteering, but similar arguments
apply to P2P- and cycle- orienteering as well.

Friggstad and Swamy~\cite{FriggstadS17} proposed the LP-relaxation \eqref{rorlp} for rooted
orienteering. Recall that the input is a metric space $(V,c)$, root $r\in V$, cost bound
$B$, and node rewards $\{\pi_v\}_{v\in V}$. Let $G=(V,E)$ be the complete graph on
$V$, and $D_G=(V,A_G)$ be the bidirected version of $G$, where both $(u,v)$ and $(v,u)$
get cost $c_{uv}$.
Let $V'=V\sm\{r\}$. To avoid unnecessary notational complication, we assume that
there is at least one node $w\in V'$ with $c_{rw}\leq B$.
\begin{alignat*}{3}  
\max & \quad & \pi_r+\sum_{u,w\in V'}&\pi_uz^w_u\tag{RO-P} \label{rorlp} \\
\text{s.t.} && 
x^w(\dt^\into(u)\bigr) & \geq x^w\bigl(\dt^\out(u)\bigr) \qquad && 
\forall u,w\in V' \\ 
&& x^w\bigl(\dt^\into(u)\bigr) & = 0 \qquad && 
\forall u,w\in V': c_{ru}>c_{rw} \\ 
&& x^w\bigl(\delta^{\into}(S)\bigr) & \geq z^w_u \qquad 
&& \forall w\in V', S\subseteq V', u\in S \\ 
&& \sum_{a\in A_G}c_ax^w_a & \leq B z^w_w  && \forall w \in V' \\ 
&& x^w\bigl(\dt^\out(r)\bigr) & =z^w_w \quad \forall w\in V',
&& \sum_{w} z^w_w = 1, \quad x,z \geq 0. \notag
\end{alignat*}
Let $\rorientlp[w]$ denote \eqref{rorlp}, when we fix $z^w_w=1$, and hence $z^v_v=0$
for all other $v\in V'$. As noted in~\cite{FriggstadS17}, $z^v_v=0$ implies that $z^v_u=0$
for all $u\in V'$, and hence, we may assume that $x^v_a=0$ for all $a\in A_G$.
So we obtain the following simpler LP.
\begin{alignat}{3}  
\max & \quad & \pi_r+\sum_{u\in V'}&\pi_uz_u \tag{RO-P${}_w$} \label{rorlpw} \\
\text{s.t.} && 
x(\dt^\into(u)\bigr) & \geq x\bigl(\dt^\out(u)\bigr) \qquad && 
\forall u\in V' \notag \\ 
&& x\bigl(\dt^\into(u)\bigr) & = 0 \qquad && 
\forall u\in V': c_{ru}>c_{rw} \label{preflpw2} \\
&& x\bigl(\delta^{\into}(S)\bigr) & \geq z_u \qquad 
&& \forall S\subseteq V', u\in S \label{conlpw} \\
&& \sum_{a\in A_G}c_ax_a & \leq B \notag \\ 
&& x\bigl(\dt^\out(r)\bigr) & =z_w=1, \quad x,z\geq 0. \notag
\end{alignat}
For notational convenience, define $\rorientlp[r]$ to be the vacuous LP with no variables
and constraints, and objective function fixed to $\pi_r$.

\begin{lemma} \label{rorlpopt}
We have $\rorientlpopt=\max_{w\in V:c_{rw}\leq B}\rorientlpopt[w]$.
\end{lemma}

\begin{proof}
First, note that we can exclude $r$ from the maximization on the RHS since
there is some $w\in V'$ with $c_{rw}\leq B$, and 
$\rorientlpopt[r]=\pi_r\leq\rorientlpopt[w]$.
It is clear that $\rorientlpopt\geq\max_{w:c_{rw}\leq B}\rorientlpopt[w]$ since 
\eqref{rorlpw} is obtained by fixing certain variables of \eqref{rorlp}.

For the other direction, let $(\bx,\bz)$ be an optimal solution to \eqref{rorlp}. We show
that the objective value of $(\bx,\bz)$ can be written as the $\bz^w_w$-weighted convex
combination of the objective values of feasible solutions to \eqref{rorlpw}, for nodes
$w\in V'$ such that $c_{rw}\leq B$. This proves the lemma. 

Consider any $w\in V'$ with $\bz^w_w>0$. It is
immediate that $\bigl(\{\bx^w_a/\bz^w_w\}_{a\in A_G},\{\bz^w_u/\bz^w_w\}_{u\in V'}\bigr)$
yields a feasible solution to \eqref{rorlpw}. Moreover, since \eqref{rorlpw} is feasible,
it must be that $c_{rw}\leq B$: constraints\eqref{conlpw} together with $z_w=1$ imply that
there is an $(r,w)$ flow $f\leq x$ of value $1$. But then the cost of $f$ is at least
$c_{rw}$ (due to a path decomposition of $f$) and at most $\sum_{a\in A_G}c_ax_a\leq B$.
\end{proof}

Consider any $w\in V$ with $c_{rw}\leq B$. 
Define $\UBrooted(w,B):=\min_{\ld\geq 0}\UBrooted(w,B;\ld)$. Recall that 
$\UBrooted(w,B;\ld)=\pi(\bV_w)+\frac{B-\dualub(\ld)}{\ld}$, where 
$\bV_w:=\{u\in V:c_{ru}\leq c_{rw}\}$, and $\dualub(\ld)$ denotes 
$\sum_{S\sse N'}y^{(\ld)}_S$ in the run \BinSearchPCA$\!(\bV_w,c,\pi,B,r,w;\e)$.

\begin{theorem} \label{lpboundsthm}
(a) $\UBrooted(w,B)\geq\rorientlpopt[w]$; and 
(b) the value of the solution returned by our algorithm for the guess $w$ is at least
$\frac{1-\e}{3}\cdot\UBrooted(w,B)$. 
\end{theorem}

Theorem~\ref{lpboundsthm} and Lemma~\ref{rorlpopt} together imply that the solution returned
by Theorem~\ref{orient-thm} has value at 
least $\frac{1-\e}{3}\cdot\UBrooted(B)$, and $\UBrooted(B)\geq\rorientlpopt$.

\begin{proof}
If $w=r$, then (a) and (b) trivially hold since
$\UBrooted(r,B;\ld)=\pi_r=\rorientlpopt[r]$ for all $\ld\geq 0$. So suppose $w\in V'$. 
Recall that \BinSearchPCA$\!(\bV_w,c,\pi,B,r,w;\e)$ calls \IterPCA on instances
$(D,c,\tpi(\ld),r)$ for varying $\ld$, where $D=(\bV_w,A)$ is the digraph obtained by
restricting $G$ to $\bV_w$ and bidirecting the edges so that $(u,v)$, $(v,u)$ both get
cost $c_{uv}$.

Let $(x^*,z^*)$ be an optimal solution to \eqref{rorlpw}. 
Note that $z^*_v\leq 1$ for all $v\in V'$, due to constraint \eqref{conlpw} for the set
$S=V'$. 
Observe that $\bigl(x^*|_D=\{x^*_a\}_{a\in A},\{p_v=1-z^*_v\}_{v\in\bV_w}\bigr)$ is a
feasible solution to $\lpname[(D,c,\tpi)]$. So by Theorem~\ref{ybnd}, 
we have
$\dualub(\ld)\leq\sum_{a\in A}c_ax^*_a+\sum_{v\in\bV_w}\ld\pi_V(1-z^*_v)
\leq B+\ld\cdot\pi(\bV_w)-\ld\cdot\rorientlpopt[w]$
where the last inequality is because constraints \eqref{preflpw2}, \eqref{conlpw} imply
that $z^*_v=0$ for all $v\notin\bV_w$. Rearranging the inequality yields that
$\rorientlpopt[w]\leq\UBrooted(w,B;\ld)$; since this holds for all $\ld\geq 0$, we also
have $\rorientlpopt[w]\leq\UBrooted(w,B)$. This proves (a).

For (b), suppose that \BinSearchPCA$\!(\bV_w,c,\pi,B,r,w;\e)$ returns the distribution
$(a,T_{\ld_1})$, $(b,T_{\ld_2})$ (where $a$ or $b$ could be $0$).   
Then, by Theorem~\ref{orient-rnd}, the reward of the solution returned is at least
$\bigl(a\cdot\pi(V(T_{\ld_1}))+b\cdot\pi(V(T_{\ld_2}))\bigr)/3$.
We now argue that $a\cdot\pi(V(T_{\ld_1}))+b\cdot\pi(V(T_{\ld_2}))\geq(1-\e)\UBrooted(w,B)$.
If \BinSearchPCA returns in step~\eqref{binupper}---so say $a=0$ and $\ld_2=\binub$---we 
have $\pi(V(T_{\ld_2}))=\pi(\bV_W)\geq\UBrooted(w,B)$. 
Otherwise, we have $a\cdot c(T_{\ld_1})+b\cdot c(T_{\ld_2})=B$ (this
holds even when \BinSearchPCA returns in step~\eqref{casei}), and
$\binlb\leq\ld_1\leq\ld_2\leq\ld_1+\e\cdot\binlb\cdot\min\{1,\binlb\}\cdot\lb$.
From \eqref{pcbnd}, we have that 
%
\begin{equation*}
c(T_{\ld_1})+\ld_1\cdot\pi(\bV_w\sm V(T_{\ld_1})) \leq
\dualub(\ld_1), \qquad  
c(T_{\ld_2})+\ld_2\cdot\pi(\bV_w\sm V(T_{\ld_2})) \leq 
\dualub(\ld_2). 
\end{equation*}
%
Multiplying the first inequality by $a$, and the second one by $b$, and adding and
simplifying, gives 
\begin{alignat}{1}
B+\ld_1\Bigl[a\cdot\pi(\bV_w\sm V(T_{\ld_1}))+b\cdot\pi(\bV_w\sm V(T_{\ld_2}))\Bigr]
& \leq a\dualub(\ld_1)+b\dualub(\ld_2) \label{ineq1} \\
\implies \quad B+\ld_1\cdot\pi(\bV_w)
-\ld_1\Bigl[a\cdot\pi(V(T_{\ld_1}))+b\cdot\pi(V(T_{\ld_2}))\Bigr] 
& \leq a\dualub(\ld_1)+b\dualub(\ld_2) \label{ineq2}
\end{alignat}
If $B\geq a\dualub(\ld_1)+b\dualub(\ld_2)$, then \eqref{ineq2} implies that 
$a\cdot\pi(V(T_{\ld_1}))+b\cdot\pi(V(T_{\ld_2}))\geq\pi(\bV_w)\geq\UBrooted(w,B)$.
Otherwise, from \eqref{ineq1}, we obtain that 
\begin{equation*}
\begin{split}
B+\ld_2\Bigl[a\cdot\pi(\bV_w\sm V(T_{\ld_1}))+b\cdot\pi(\bV_w\sm V(T_{\ld_2}))\Bigr]
& \leq a\dualub(\ld_1)+b\dualub(\ld_2)+(\ld_2-\ld_1)\pi(\bV_w\sm\{r\}) \\
& \leq a\dualub(\ld_1)+b\dualub(\ld_2)+\e\cdot\binlb^2\cdot\lb\cdot\pi(\bV_w\sm\{r\}) \\
& \leq a\dualub(\ld_1)+b\dualub(\ld_2)+\e\ld_2\cdot\UBrooted(w,B).
\end{split}
\end{equation*}
The last inequality above follows since $\binlb\leq\ld_2$,
$\binlb\cdot\pi(\bV_w\sm\{r\})=1$, and 
$\lb\leq\ropt_w\leq\UBrooted(w,B)$. 
Dividing the above by $\ld_2$ and rearranging yields,
\begin{equation*}
\begin{split}
a\cdot\pi(V(T_{\ld_1}))&+b\cdot\pi(V(T_{\ld_2}))
\geq \pi(\bV_w)+\frac{B-a\dualub(\ld_1)-b\dualub(\ld_2)}{\ld_2}-\e\cdot\UBrooted(w,B) \\
& \geq \pi(\bV_w)+\frac{B-a\dualub(\ld_1)-b\dualub(\ld_2)}{a\ld_1+b\ld_2}-\e\cdot\UBrooted(w,B) \\
& \geq \pi(\bV_w)
+\min\biggl\{\frac{B-\dualub(\ld_1)}{\ld_1},\frac{B-\dualub(\ld_2)}{\ld_2}\biggr\}
-\e\cdot\UBrooted(w,B) \geq (1-\e)\UBrooted(w,B).
\end{split}
\end{equation*}
The second inequality is because we are in the case where 
$B-a\dualub(\ld_1)-b\dualub(\ld_2)<0$, and the final inequality follows because
$\UBrooted(w,B)$ is the minimum of $\UBrooted(w,B;\ld)$ over all $\ld\geq 0$.
\end{proof}

\section{Improved analysis of \ptp-orienteering algorithm
  in~\cite{FriggstadS17}} \label{append-ptpimprov}. 
Friggstad and Swamy~\cite{FriggstadS17} modify their LP-relaxation for rooted orienteering
to obtain an LP-relaxation for P2P orienteering, using which they obtain an approximation
factor of $6$ for P2P orienteering. We show that their analysis can be improved to show an
approximation factor of $4$.

While the P2P-orienteering LP in~\cite{FriggstadS17} does not require any
``guesswork'' (i.e., enumeration steps), to get to the heart of their approach,
and our improved analysis, we assume that we know (or ``guess'') the node $w$ on the
optimal $r$-$t$ path with largest $c_{ru}+c_{ut}$ value. 
Recall that $\bV^\ptp_w=\{u\in V: c_{ru}+c_{ut}\leq c_{rw}+c_{wt}\}$, and
$\roptptp_w$ is the maximum reward of an $r$-$t$ path visits $w$, and only visits
nodes in $\bV^{\ptp}_w$.

In~\cite{FriggstadS17}, two weighted collections of trees, 
$\bigl(\gm_i^{(r)},T_i^{(r)}\bigr)_{i=1,\ldots,k}$ and 
$\bigl(\gm_i^{(t)},T_i^{(t)}\bigr)_{i=1,\ldots,\ell}$, are extracted from the LP solution,
such that: 
\begin{enumerate}[(a), topsep=1ex, itemsep=0.25ex, leftmargin=*]
\item all $\gm_i^{(r)}$ and $\gm_i^{(t)}$ values are nonnegative, and
$\sum_{i=1}^k\gm_i^{(r)}=1=\sum_{i=1}^\ell\gm_i^{(t)}$;
\item each $T_i^{(r)}$ and $T_i^{(t)}$ tree contains $w$, and contains only nodes from
  $\bV^\ptp_w$. 
\item each $T_i^{(r)}$ tree is rooted at $r$, and each $T_i^{(t)}$ tree is rooted at $t$;
\item $\sum_{i=1}^k\gm_i^{(r)}c(T_i^{(r)})+\sum_{i=1}^\ell\gm_i^{(t)}c(T_i^{(t)})\leq B$;
and
\item $\sum_{i=1}^k\gm_i^{(r)}\pi\bigl(V(T_i^{(r)})\bigr)
+\sum_{i=1}^\ell\gm_i^{(t)}\pi\bigl(V(T_i^{(t)})\bigr)\geq\roptptp_w$. 
\end{enumerate}

We convert each $T_i^{(r)}$ tree to an $r$-$w$ path $Q_i^{(r)}$, and
extract an $r$-rooted path $P_i^{(r)}$ from $Q_i^{(r)}$ ending at some node in $\bV^\ptp_w$ such
that $\creg(P_i^{(r)})\leq B-c_{rw}-c_{wt}$, and 
$\pi\bigl(V(P_i^{(r)})\bigr)\geq
\pi\bigl(V(Q_i^{(r)})\bigr)/\bigl(\frac{\creg(Q_i^{(r)})}{B-c_{rw}-c_{wt}}+1\bigr)$  
(see Lemma 5.1 in~\cite{BlumCKLMM07}, or Lemma 2.2 in~\cite{FriggstadS14});
note that $P_i^{(r)}$ can be completed to an $r$-$t$ path of cost at most $B$.
Similarly, we convert each $T_i^{(t)}$ tree to $t$-$w$ path $Q_i^{(t)}$, and
extract a $t$-rooted path $P_i^{(t)}$ from $Q_i^{(t)}$ ending at some node in $\bV^\ptp_w$
such that $\creg(P_i^{(t)})\leq B-c_{rw}-c_{wt}$, and 
$\pi\bigl(V(P_i^{(t)})\bigr)\geq
\pi\bigl(V(Q_i^{(t)})\bigr)/\bigl(\frac{\creg(Q_i^{(t)})}{B-c_{rw}-c_{wt}}+1\bigr)$;  
note that $P_i^{(t)}$ can be completed to an $r$-$t$ path of cost at most $B$.

We claim that the best of the $P_i^{(r)}$ and $P_i^{(t)}$ paths earns reward at least
$\roptptp_w/4$. The maximum reward earned by one of these paths is at least 
\begin{alignat*}{1}
\max\ \biggl\{&\max\,\biggl\{\frac{\pi\bigl(V(Q_i^{(r)})\bigr)}{\frac{\creg(Q_i^{(r)})}{B-c_{rw}-c_{wt}}+1}:
i=1,\ldots,k\biggr\},\ 
\max\,\biggl\{\frac{\pi\bigl(V(Q_i^{(t)})\bigr)}{\frac{\creg(Q_i^{(t)})}{B-c_{rw}-c_{wt}}+1}:
i=1,\ldots,\ell\biggr\}\biggr\} \\[0.5ex]
&
\geq\frac{\sum_{i=1}^k\gm_i^{(r)}\pi\bigl(V(Q_i^{(r)})\bigr)+\sum_{i=1}^\ell\gm_i^{(t)}\pi\bigl(V(Q_i^{(t)})\bigr)}
{\sum_{i=1}^k\gm_i^{(r)}\cdot\bigl(\frac{\creg(Q_i^{(r)})}{B-c_{rw}-c_{wt}}+1\bigr)+
\sum_{i=1}^\ell\gm_i^{(t)}\cdot\bigl(\frac{\creg(Q_i^{(t)})}{B-c_{rw}-c_{wt}}+1\bigr)} 
\tag{using (a)} \\[0.5ex]
& \geq
\frac{\roptptp_w}{\sum_{i=1}^k\gm_i^{(r)}\cdot\frac{2c(T_i^{(r)})-2c_{rw}}{B-c_{rw}-c_{wt}}
+\sum_{i=1}^\ell\gm_i^{(t)}\cdot\frac{2c(T_i^{(t)})-2c_{wt}}{B-c_{rw}-c_{wt}}+2}
\tag{using (a) and (e)} \\[0.5ex]
& \geq
\frac{\roptptp_w}{\frac{2B-2c_{rw}-2c_{wt}}{B-c_{rw}-c_{wt}}+2}=\frac{\roptptp_w}{4}.
\tag{using (a) and (d)}
\end{alignat*}

However, it is not immediate as to how one can leverage the above improved analysis via
our combinatorial PCA algorithm. The roadblock is property (d), which is a combined
guarantee on the $\gm$-weighted average cost of a tree in the two collections, and it is
not clear how one can obtain this by running \BinSearchPCA with roots $r$ and $t$
separately.  
One way of circumventing this obstacle is as follows. Let $P^*_w$ be the P2P-orienteering
solution yielding reward $\roptptp_w$. We guess (within some factor) the regrets of the
$r$-$w$ and $w$-$t$ portions of $P^*_w$, which add up to $B-c_{rw}-c_{wt}$, and run
\BinSearchPCA$\!(\bV^\ptp_w,c,\pi,\cdot,r,w;\e)$ and
\BinSearchPCA$\!(\bV^\ptp_w,c,\pi,\cdot,t,w;\e)$ with the corresponding length bounds. An
analysis similar to above can then be used to show that this yields a
$4/(1-\e)$-approximation.

\section{Computational results continued: omitted tables and further discussion} 
\label{sec:tables} \label{append-compres}
We complete our discussion of computational results by listing here the 
tables omitted from Section~\ref{compres}, 
which report the detailed results of our computational experiments for the following
datasets.
\begin{enumerate}[label=$\bullet$, topsep=1ex, itemsep=0.5ex, leftmargin=*]
\item For cycle orienteering, Table~\ref{cyc-gen2} contains the results for the Gen 2 data
  sets.
\item For P2P orienteering, Tables~\ref{ptp-tsil} and~\ref{ptp-chao} list the results
  obtained for the instances from Tsiligirides~\cite{Tsili84} (where we do have optimal
  values) and Chao~\cite{Chao96} (where we do not have optimal values) respectively.
\end{enumerate}
 
For cycle orienteering, we use optimal values as reported in~\cite{Kobeaga21}. Most of 
these were first computed in \cite{Fischetti98}. The distance bound used in
\cite{Fischetti98} for the dataset \texttt{gr229} 
seems to be incorrectly reported (it is far too small to support the optimum value they
claim), so we use the bound from \cite{Kobeaga21}. Further, \cite{Kobeaga21} include three
TSPLIB data sets with at most $400$ nodes that we not considered in \cite{Fischetti98}:
\texttt{berlin52}, \texttt{tsp225}, and \texttt{a280}. We include these datasets in our
experiments for cycle orienteering.

\begin{table}[t!]
\begin{center}
{\small
    \caption{Cycle orienteering results - Gen 2 \label{cyc-gen2}}

    \medskip\medskip
    \begin{tabular}{l|l|l|l|l|l|l|l}%
    \bfseries Dataset & {\boldmath $B$} & {\boldmath $\val$} & {\boldmath $\opt$} &
    {\boldmath $\ub$} & {\boldmath $\opt/\val$} & {\boldmath $\ub/\opt$} & {\boldmath $\ub/\val$} 
    \\\hline \hline \vspace{-4.325mm}
    \begin{filecontents*}{csv/cycles-gen2.csv}
Dataset, of nodes,Distance Limit,Orienteering,Distance,Upper bound,opt-val,Approximation factor,Upper bound / val,Upper Bound / OPT,Comment,
*att48,48,5314,1379,5100,1852.61,1717,1.245,1.343,1.079,Ceiling,docker run -i --init ghcr.io/ssinad/dcvr-thesis/cycle-orienteering-app -D 5314 < tsp/att48.tsp-gen-2.txt 2>/dev/null;
*gr48,48,2523,1380,2278,2046.2,1761,1.276,1.483,1.162,Not metric,docker run -i --init ghcr.io/ssinad/dcvr-thesis/cycle-orienteering-app -D 2523 < tsp/gr48.tsp-gen-2.txt 2>/dev/null;
*hk48,48,5731,1232,5625,1883.5,1614,1.31,1.529,1.167,Not metric,docker run -i --init ghcr.io/ssinad/dcvr-thesis/cycle-orienteering-app -D 5731 < tsp/hk48.tsp-gen-2.txt 2>/dev/null;
eil51,51,213,1330,192.933,1875.15,1674,1.259,1.41,1.12,,docker run -i --init ghcr.io/ssinad/dcvr-thesis/cycle-orienteering-app -D 213 < tsp/eil51.tsp-gen-2.txt 2>/dev/null;
berlin52,52,3771,1644,3609.33,2215.32,1897,1.154,1.348,1.168,,docker run -i --init ghcr.io/ssinad/dcvr-thesis/cycle-orienteering-app -D 3771 < tsp/berlin52.tsp-gen-2.txt 2>/dev/null;
*brazil58,58,12698,1864,12314,2622.08,2220,1.191,1.407,1.181,Not metric,docker run -i --init ghcr.io/ssinad/dcvr-thesis/cycle-orienteering-app -D 12698 < tsp/brazil58.tsp-gen-2.txt 2>/dev/null;
st70,70,338,1905,325.123,2730.13,2286,1.2,1.433,1.194,,docker run -i --init ghcr.io/ssinad/dcvr-thesis/cycle-orienteering-app -D 338 < tsp/st70.tsp-gen-2.txt 2>/dev/null;
eil76,76,269,1905,261.959,2844.06,2550,1.339,1.493,1.115,,docker run -i --init ghcr.io/ssinad/dcvr-thesis/cycle-orienteering-app -D 269 < tsp/eil76.tsp-gen-2.txt 2>/dev/null;
pr76,76,54080,2038,46833.1,3020.77,2708,1.329,1.482,1.115,,docker run -i --init ghcr.io/ssinad/dcvr-thesis/cycle-orienteering-app -D 54080 < tsp/pr76.tsp-gen-2.txt 2>/dev/null;
gr96,96,27605,2630,22953.8,3811.41,3396,1.291,1.449,1.122,,docker run -i --init ghcr.io/ssinad/dcvr-thesis/cycle-orienteering-app -D 27605 < tsp/gr96.tsp-gen-2.txt 2>/dev/null;
rat99,99,606,2179,593.805,3285.75,2944,1.351,1.508,1.116,,docker run -i --init ghcr.io/ssinad/dcvr-thesis/cycle-orienteering-app -D 606 < tsp/rat99.tsp-gen-2.txt 2>/dev/null;
kroA100,100,10641,2686,10610.4,3638.26,3212,1.196,1.355,1.133,,docker run -i --init ghcr.io/ssinad/dcvr-thesis/cycle-orienteering-app -D 10641 < tsp/kroA100.tsp-gen-2.txt 2>/dev/null;
kroB100,100,11071,2581,10205.2,3646.73,3241,1.256,1.413,1.125,,docker run -i --init ghcr.io/ssinad/dcvr-thesis/cycle-orienteering-app -D 11071 < tsp/kroB100.tsp-gen-2.txt 2>/dev/null;
kroC100,100,10375,2343,9634.32,3358.61,2947,1.258,1.433,1.14,,docker run -i --init ghcr.io/ssinad/dcvr-thesis/cycle-orienteering-app -D 10375 < tsp/kroC100.tsp-gen-2.txt 2>/dev/null;
kroD100,100,10647,2425,10621.1,3891.9,3307,1.364,1.605,1.177,,docker run -i --init ghcr.io/ssinad/dcvr-thesis/cycle-orienteering-app -D 10647 < tsp/kroD100.tsp-gen-2.txt 2>/dev/null;
kroE100,100,11034,2331,10338.9,3653.97,3090,1.326,1.568,1.183,,docker run -i --init ghcr.io/ssinad/dcvr-thesis/cycle-orienteering-app -D 11034 < tsp/kroE100.tsp-gen-2.txt 2>/dev/null;
rd100,100,3955,2735,3653.51,3749.78,3359,1.228,1.371,1.116,,docker run -i --init ghcr.io/ssinad/dcvr-thesis/cycle-orienteering-app -D 3955 < tsp/rd100.tsp-gen-2.txt 2>/dev/null;
eil101,101,315,2991,313.551,3975.64,3655,1.222,1.329,1.088,,docker run -i --init ghcr.io/ssinad/dcvr-thesis/cycle-orienteering-app -D 315 < tsp/eil101.tsp-gen-2.txt 2>/dev/null;
lin105,105,7190,2727,6843.25,4023,3544,1.3,1.475,1.135,,docker run -i --init ghcr.io/ssinad/dcvr-thesis/cycle-orienteering-app -D 7190 < tsp/lin105.tsp-gen-2.txt 2>/dev/null;
pr107,107,22152,2477,20809.1,2667,2667,1.077,1.077,1,,docker run -i --init ghcr.io/ssinad/dcvr-thesis/cycle-orienteering-app -D 22152 < tsp/pr107.tsp-gen-2.txt 2>/dev/null;
*gr120,120,3471,3285,3463,4819.97,4371,1.331,1.467,1.103,Not metric,docker run -i --init ghcr.io/ssinad/dcvr-thesis/cycle-orienteering-app -D 3471 < tsp/gr120.tsp-gen-2.txt 2>/dev/null;
pr124,124,29515,2914,28205.1,4609.9,3917,1.344,1.582,1.177,,docker run -i --init ghcr.io/ssinad/dcvr-thesis/cycle-orienteering-app -D 29515 < tsp/pr124.tsp-gen-2.txt 2>/dev/null;
bier127,127,59141,4622,58341,5874.66,5383,1.165,1.271,1.091,,docker run -i --init ghcr.io/ssinad/dcvr-thesis/cycle-orienteering-app -D 59141 < tsp/bier127.tsp-gen-2.txt 2>/dev/null;
pr136,136,48386,3380,47346.4,4854.86,4309,1.275,1.436,1.127,,docker run -i --init ghcr.io/ssinad/dcvr-thesis/cycle-orienteering-app -D 48386 < tsp/pr136.tsp-gen-2.txt 2>/dev/null;
gr137,137,34927,3147,33763.8,4925.73,4286,1.362,1.565,1.149,,docker run -i --init ghcr.io/ssinad/dcvr-thesis/cycle-orienteering-app -D 34927 < tsp/gr137.tsp-gen-2.txt 2>/dev/null;
pr144,144,29269,3389,26601,4816.59,4003,1.181,1.421,1.203,,docker run -i --init ghcr.io/ssinad/dcvr-thesis/cycle-orienteering-app -D 29269 < tsp/pr144.tsp-gen-2.txt 2>/dev/null;
kroA150,150,13262,4025,13096.4,5546.29,4918,1.222,1.378,1.128,,docker run -i --init ghcr.io/ssinad/dcvr-thesis/cycle-orienteering-app -D 13262 < tsp/kroA150.tsp-gen-2.txt 2>/dev/null;
kroB150,150,13065,3608,12587,5383.39,4869,1.35,1.492,1.106,,docker run -i --init ghcr.io/ssinad/dcvr-thesis/cycle-orienteering-app -D 13065 < tsp/kroB150.tsp-gen-2.txt 2>/dev/null;
pr152,152,36841,3403,36724,5150.82,4279,1.257,1.514,1.204,,docker run -i --init ghcr.io/ssinad/dcvr-thesis/cycle-orienteering-app -D 36841 < tsp/pr152.tsp-gen-2.txt 2>/dev/null;
u159,159,21040,4055,20954.9,5957.21,4960,1.223,1.469,1.201,,docker run -i --init ghcr.io/ssinad/dcvr-thesis/cycle-orienteering-app -D 21040 < tsp/u159.tsp-gen-2.txt 2>/dev/null;
rat195,195,1162,4108,1143.48,6098.92,5791,1.41,1.485,1.053,,docker run -i --init ghcr.io/ssinad/dcvr-thesis/cycle-orienteering-app -D 1162 < tsp/rat195.tsp-gen-2.txt 2>/dev/null;
d198,198,7890,5042,7870.68,8223.9,6670,1.323,1.631,1.233,,docker run -i --init ghcr.io/ssinad/dcvr-thesis/cycle-orienteering-app -D 7890 < tsp/d198.tsp-gen-2.txt 2>/dev/null;
kroA200,200,14684,4790,14558,7228.18,6547,1.367,1.509,1.104,,docker run -i --init ghcr.io/ssinad/dcvr-thesis/cycle-orienteering-app -D 14684 < tsp/kroA200.tsp-gen-2.txt 2>/dev/null;
kroB200,200,14719,4943,14675.8,7075.14,6419,1.299,1.431,1.102,,docker run -i --init ghcr.io/ssinad/dcvr-thesis/cycle-orienteering-app -D 14719 < tsp/kroB200.tsp-gen-2.txt 2>/dev/null;
gr202,202,20080,6338,19685.5,8836.24,7789,1.229,1.394,1.134,,docker run -i --init ghcr.io/ssinad/dcvr-thesis/cycle-orienteering-app -D 20080 < tsp/gr202.tsp-gen-2.txt 2>/dev/null;
ts225,225,63322,4995,63007.7,7332.41,6834,1.368,1.468,1.073,,docker run -i --init ghcr.io/ssinad/dcvr-thesis/cycle-orienteering-app -D 63322 < tsp/ts225.tsp-gen-2.txt 2>/dev/null;
tsp225,225,1958,5576,1829.87,7920.81,6987,1.253,1.421,1.134,,docker run -i --init ghcr.io/ssinad/dcvr-thesis/cycle-orienteering-app -D 1958 < tsp/tsp225.tsp-gen-2.txt 2>/dev/null;
pr226,226,40185,4991,39893.7,7583.05,6662,1.335,1.519,1.138,,docker run -i --init ghcr.io/ssinad/dcvr-thesis/cycle-orienteering-app -D 40185 < tsp/pr226.tsp-gen-2.txt 2>/dev/null;
gr229,229,67301,7512,64125.1,9858.24,9177,1.222,1.312,1.074,,docker run -i --init ghcr.io/ssinad/dcvr-thesis/cycle-orienteering-app -D 67301 < tsp/gr229.tsp-gen-2.txt 2>/dev/null;
gil262,262,1189,6265,1133.46,9255.62,8321,1.328,1.477,1.112,,docker run -i --init ghcr.io/ssinad/dcvr-thesis/cycle-orienteering-app -D 1189 < tsp/gil262.tsp-gen-2.txt 2>/dev/null;
pr264,264,24568,5892,21173.9,7071.99,6654,1.129,1.2,1.063,,docker run -i --init ghcr.io/ssinad/dcvr-thesis/cycle-orienteering-app -D 24568 < tsp/pr264.tsp-gen-2.txt 2>/dev/null;
a280,280,1290,6219,1253.86,9016.01,8428,1.355,1.45,1.07,,docker run -i --init ghcr.io/ssinad/dcvr-thesis/cycle-orienteering-app -D 1290 < tsp/a280.tsp-gen-2.txt 2>/dev/null;
pr299,299,24096,6898,21222.6,10712.4,9182,1.331,1.553,1.167,,docker run -i --init ghcr.io/ssinad/dcvr-thesis/cycle-orienteering-app -D 24096 < tsp/pr299.tsp-gen-2.txt 2>/dev/null;
lin318,318,21045,8121,20362.8,12296.2,10923,1.345,1.514,1.126,,docker run -i --init ghcr.io/ssinad/dcvr-thesis/cycle-orienteering-app -D 21045 < tsp/lin318.tsp-gen-2.txt 2>/dev/null;
rd400,400,7641,11099,7399.42,14671,13652,1.23,1.322,1.075,,docker run -i --init ghcr.io/ssinad/dcvr-thesis/cycle-orienteering-app -D 7641 < tsp/rd400.tsp-gen-2.txt 2>/dev/null;
\end{filecontents*}
\csvreader[head to column names]{csv/cycles-gen2.csv}{}
    {\\\hline\csvcoli & \csvcoliii & \csvcoliv & \csvcolvii & \csvcolvi & \csvcolviii  & \csvcolx &\csvcolix}
    \end{tabular}
}
\end{center}
\end{table}

\begin{table}[t!]
\begin{center}
{\small
    \caption{P2P orienteering results for the instances in~\cite{Tsili84} with varying
      distance bounds. \label{ptp-tsil}} 

    \medskip\medskip
    \begin{tabular}{l|l|l|l|l|l|l|l}%
    \bfseries Dataset & {\boldmath $B$} & {\boldmath $\val$} & {\boldmath $\opt$} &
    {\boldmath $\ub$} & {\boldmath $\opt/\val$} & {\boldmath $\ub/\opt$} & {\boldmath $\ub/\val$} 
    \\\hline \hline \vspace{-4.325mm}
    \begin{filecontents*}{csv/p2p_21.csv}
Distance Limit,Best so far,Reward,Distance,Upper bound,Upper bound / Reward,Reward,Distance,Upper bound,Approx,Upper bound / reward,UB / best,,
15,120,185,14.88,190.37,1.03,120,14.7573,176.176,1,1.468,1.468,docker run -i --init ssinad/p2p-orienteering-app -D 15 < 2/tsiligirides_problem_2.txt 2>/dev/null;,docker run -i --init ssinad/rooted-orienteering-app -D 15 < 2/2.txt 2>/dev/null;
20,200,230,19.26,254.23,1.11,150,18.1236,230,1.333,1.533,1.15,docker run -i --init ssinad/p2p-orienteering-app -D 20 < 2/tsiligirides_problem_2.txt 2>/dev/null;,docker run -i --init ssinad/rooted-orienteering-app -D 20 < 2/2.txt 2>/dev/null;
23,210,275,22.64,297.08,1.08,180,22.9026,261.012,1.167,1.45,1.243,docker run -i --init ssinad/p2p-orienteering-app -D 23 < 2/tsiligirides_problem_2.txt 2>/dev/null;,docker run -i --init ssinad/rooted-orienteering-app -D 23 < 2/2.txt 2>/dev/null;
25,230,275,22.64,326.68,1.19,180,22.9026,308.824,1.278,1.716,1.343,docker run -i --init ssinad/p2p-orienteering-app -D 25 < 2/tsiligirides_problem_2.txt 2>/dev/null;,docker run -i --init ssinad/rooted-orienteering-app -D 25 < 2/2.txt 2>/dev/null;
27,230,310,26.59,350.27,1.13,200,26.5336,347.219,1.15,1.736,1.51,docker run -i --init ssinad/p2p-orienteering-app -D 27 < 2/tsiligirides_problem_2.txt 2>/dev/null;,docker run -i --init ssinad/rooted-orienteering-app -D 27 < 2/2.txt 2>/dev/null;
30,265,315,28.88,384.66,1.22,220,29.1289,381.686,1.205,1.735,1.44,docker run -i --init ssinad/p2p-orienteering-app -D 30 < 2/tsiligirides_problem_2.txt 2>/dev/null;,docker run -i --init ssinad/rooted-orienteering-app -D 30 < 2/2.txt 2>/dev/null;
32,300,330,31.42,407.25,1.23,230,31.4641,404.665,1.304,1.759,1.349,docker run -i --init ssinad/p2p-orienteering-app -D 32 < 2/tsiligirides_problem_2.txt 2>/dev/null;,docker run -i --init ssinad/rooted-orienteering-app -D 32 < 2/2.txt 2>/dev/null;
35,320,355,34.51,439.55,1.24,255,34.3185,439.132,1.255,1.722,1.372,docker run -i --init ssinad/p2p-orienteering-app -D 35 < 2/tsiligirides_problem_2.txt 2>/dev/null;,docker run -i --init ssinad/rooted-orienteering-app -D 35 < 2/2.txt 2>/dev/null;
38,355,370,37.86,450,1.22,300,37.2812,450,1.2,1.5,1.268,docker run -i --init ssinad/p2p-orienteering-app -D 38 < 2/tsiligirides_problem_2.txt 2>/dev/null;,docker run -i --init ssinad/rooted-orienteering-app -D 38 < 2/2.txt 2>/dev/null;
40,395,390,39.81,450,1.15,300,37.2812,450,1.317,1.5,1.139,docker run -i --init ssinad/p2p-orienteering-app -D 40 < 2/tsiligirides_problem_2.txt 2>/dev/null;,docker run -i --init ssinad/rooted-orienteering-app -D 40 < 2/2.txt 2>/dev/null;
45,450,415,42.35,450,1.08,330,43.94,450,1.364,1.364,1,docker run -i --init ssinad/p2p-orienteering-app -D 45 < 2/tsiligirides_problem_2.txt 2>/dev/null;,docker run -i --init ssinad/rooted-orienteering-app -D 45 < 2/2.txt 2>/dev/null;
\end{filecontents*}
\csvreader[head to column names]{csv/p2p_21.csv}{}
    { 
      \csviffirstrow{\\ \hline \multirow{11}{*}{\begin{tabular}{c} 21 nodes\end{tabular}}
      & \csvcoli & \csvcolvii & \csvcolii & \csvcolix & \csvcolx & \csvcolxii &
      \csvcolxi}
      {\\ \cline{2-8} & \csvcoli & \csvcolvii & \csvcolii & \csvcolix & \csvcolx & \csvcolxii &
      \csvcolxi}
    } 
    \begin{filecontents*}{csv/p2p_32.csv}
Distance Limit,Prev Best,Reward,Distance,Upper bound,Upper bound / reward,Reward,Distance,Upper bound,Approx,Upper bound / reward,UB / best,,
5,10,15,4.76,16.94,1.13,10,4.14257,15,1,1.5,1.5,docker run -i --init ssinad/p2p-orienteering-app -D 5 < 1/tsiligirides_problem_1.txt 2>/dev/null;,docker run -i --init ssinad/rooted-orienteering-app -D 5 < 1/1.txt 2>/dev/null;
10,15,45,9.1,49.48,1.1,15,7.88509,32.4525,1,2.164,2.164,docker run -i --init ssinad/p2p-orienteering-app -D 10 < 1/tsiligirides_problem_1.txt 2>/dev/null;,docker run -i --init ssinad/rooted-orienteering-app -D 10 < 1/1.txt 2>/dev/null;
15,45,70,14.01,75.51,1.08,40,13.9376,69.3954,1.125,1.735,1.542,docker run -i --init ssinad/p2p-orienteering-app -D 15 < 1/tsiligirides_problem_1.txt 2>/dev/null;,docker run -i --init ssinad/rooted-orienteering-app -D 15 < 1/1.txt 2>/dev/null;
20,65,95,19.55,101.78,1.07,60,19.8255,105.706,1.083,1.762,1.626,docker run -i --init ssinad/p2p-orienteering-app -D 20 < 1/tsiligirides_problem_1.txt 2>/dev/null;,docker run -i --init ssinad/rooted-orienteering-app -D 20 < 1/1.txt 2>/dev/null;
25,90,105,24.99,127.93,1.22,85,24.783,135.92,1.059,1.599,1.51,docker run -i --init ssinad/p2p-orienteering-app -D 25 < 1/tsiligirides_problem_1.txt 2>/dev/null;,docker run -i --init ssinad/rooted-orienteering-app -D 25 < 1/1.txt 2>/dev/null;
30,110,125,29.23,152.53,1.22,95,28.6547,149.883,1.158,1.578,1.363,docker run -i --init ssinad/p2p-orienteering-app -D 30 < 1/tsiligirides_problem_1.txt 2>/dev/null;,docker run -i --init ssinad/rooted-orienteering-app -D 30 < 1/1.txt 2>/dev/null;
35,135,140,34.66,171.35,1.22,120,34.1898,171.063,1.125,1.426,1.267,docker run -i --init ssinad/p2p-orienteering-app -D 35 < 1/tsiligirides_problem_1.txt 2>/dev/null;,docker run -i --init ssinad/rooted-orienteering-app -D 35 < 1/1.txt 2>/dev/null;
40,150,155,39.98,190.17,1.23,135,39.5077,188.634,1.148,1.397,1.258,docker run -i --init ssinad/p2p-orienteering-app -D 40 < 1/tsiligirides_problem_1.txt 2>/dev/null;,docker run -i --init ssinad/rooted-orienteering-app -D 40 < 1/1.txt 2>/dev/null;
46,175,175,45.4,212.38,1.21,155,44.9353,209.31,1.129,1.35,1.196,docker run -i --init ssinad/p2p-orienteering-app -D 46 < 1/tsiligirides_problem_1.txt 2>/dev/null;,docker run -i --init ssinad/rooted-orienteering-app -D 46 < 1/1.txt 2>/dev/null;
50,190,190,48.93,226.5,1.19,165,49.3706,223.778,1.152,1.356,1.178,docker run -i --init ssinad/p2p-orienteering-app -D 50 < 1/tsiligirides_problem_1.txt 2>/dev/null;,docker run -i --init ssinad/rooted-orienteering-app -D 50 < 1/1.txt 2>/dev/null;
55,205,205,54.58,243.93,1.19,175,54.8114,241.349,1.171,1.379,1.177,docker run -i --init ssinad/p2p-orienteering-app -D 55 < 1/tsiligirides_problem_1.txt 2>/dev/null;,docker run -i --init ssinad/rooted-orienteering-app -D 55 < 1/1.txt 2>/dev/null;
60,220,205,54.58,261.28,1.27,190,58.4514,258.601,1.184,1.361,1.175,docker run -i --init ssinad/p2p-orienteering-app -D 60 < 1/tsiligirides_problem_1.txt 2>/dev/null;,docker run -i --init ssinad/rooted-orienteering-app -D 60 < 1/1.txt 2>/dev/null;
65,240,210,60.61,276.72,1.32,205,64.1352,274.642,1.171,1.34,1.144,docker run -i --init ssinad/p2p-orienteering-app -D 65 < 1/tsiligirides_problem_1.txt 2>/dev/null;,docker run -i --init ssinad/rooted-orienteering-app -D 65 < 1/1.txt 2>/dev/null;
70,260,215,65.5,285,1.33,215,69.614,285,1.209,1.326,1.096,docker run -i --init ssinad/p2p-orienteering-app -D 70 < 1/tsiligirides_problem_1.txt 2>/dev/null;,docker run -i --init ssinad/rooted-orienteering-app -D 70 < 1/1.txt 2>/dev/null;
73,265,225,72.03,285,1.27,215,72.3458,285,1.233,1.326,1.075,docker run -i --init ssinad/p2p-orienteering-app -D 73 < 1/tsiligirides_problem_1.txt 2>/dev/null;,docker run -i --init ssinad/rooted-orienteering-app -D 73 < 1/1.txt 2>/dev/null;
75,275,225,74.47,285,1.27,215,74.0175,285,1.256,1.326,1.036,docker run -i --init ssinad/p2p-orienteering-app -D 75 < 1/tsiligirides_problem_1.txt 2>/dev/null;,docker run -i --init ssinad/rooted-orienteering-app -D 75 < 1/1.txt 2>/dev/null;
80,280,230,78.42,285,1.24,225,78.8314,285,1.244,1.267,1.018,docker run -i --init ssinad/p2p-orienteering-app -D 80 < 1/tsiligirides_problem_1.txt 2>/dev/null;,docker run -i --init ssinad/rooted-orienteering-app -D 80 < 1/1.txt 2>/dev/null;
85,285,245,83.94,285,1.16,235,83.5021,285,1.213,1.213,1,docker run -i --init ssinad/p2p-orienteering-app -D 85 < 1/tsiligirides_problem_1.txt 2>/dev/null;,docker run -i --init ssinad/rooted-orienteering-app -D 85 < 1/1.txt 2>/dev/null;
\end{filecontents*}
\csvreader[head to column names]{csv/p2p_32.csv}{}
    { 
      \csviffirstrow{\\ \hline \multirow{18}{*}{\begin{tabular}{c} 32 nodes\end{tabular}}
      & \csvcoli & \csvcolvii & \csvcolii & \csvcolix & \csvcolx & \csvcolxii &
      \csvcolxi}
      {\\ \cline{2-8} & \csvcoli & \csvcolvii & \csvcolii & \csvcolix & \csvcolx & \csvcolxii &
      \csvcolxi}
    }
    \begin{filecontents*}{csv/p2p_33.csv}
Distance Limit,Best so far,Reward,Distance,Upper bound,Upper bound / Reward,Reward,Distance,Upper Bound,Approx,Upper bound / reward,UB / best,,
15,100,220,14.75,236.65,1.08,150,13.7596,196.966,1.133,1.313,1.97,docker run -i --init ssinad/p2p-orienteering-app -D 15 < 3/tsiligirides_problem_3.txt 2>/dev/null;,docker run -i --init ssinad/rooted-orienteering-app -D 15 < 3/3.txt 2>/dev/null;
20,140,270,18.83,297.62,1.1,180,19.1987,272.853,1.111,1.516,1.949,docker run -i --init ssinad/p2p-orienteering-app -D 20 < 3/tsiligirides_problem_3.txt 2>/dev/null;,docker run -i --init ssinad/rooted-orienteering-app -D 20 < 3/3.txt 2>/dev/null;
25,190,330,24.61,359.81,1.09,230,24.777,357.41,1.13,1.554,1.881,docker run -i --init ssinad/p2p-orienteering-app -D 25 < 3/tsiligirides_problem_3.txt 2>/dev/null;,docker run -i --init ssinad/rooted-orienteering-app -D 25 < 3/3.txt 2>/dev/null;
30,240,380,29.3,416.8,1.1,270,28.8977,411.068,1.185,1.522,1.713,docker run -i --init ssinad/p2p-orienteering-app -D 30 < 3/tsiligirides_problem_3.txt 2>/dev/null;,docker run -i --init ssinad/rooted-orienteering-app -D 30 < 3/3.txt 2>/dev/null;
35,290,400,33.61,469.77,1.17,320,33.108,464.726,1.219,1.452,1.603,docker run -i --init ssinad/p2p-orienteering-app -D 35 < 3/tsiligirides_problem_3.txt 2>/dev/null;,docker run -i --init ssinad/rooted-orienteering-app -D 35 < 3/3.txt 2>/dev/null;
40,340,450,38.3,517.21,1.15,360,39.2411,518.385,1.194,1.44,1.525,docker run -i --init ssinad/p2p-orienteering-app -D 40 < 3/tsiligirides_problem_3.txt 2>/dev/null;,docker run -i --init ssinad/rooted-orienteering-app -D 40 < 3/3.txt 2>/dev/null;
45,370,490,44.3,566.59,1.16,400,43.3618,571.741,1.175,1.429,1.545,docker run -i --init ssinad/p2p-orienteering-app -D 45 < 3/tsiligirides_problem_3.txt 2>/dev/null;,docker run -i --init ssinad/rooted-orienteering-app -D 45 < 3/3.txt 2>/dev/null;
50,420,520,49.18,603.89,1.16,460,49.985,608.073,1.13,1.322,1.448,docker run -i --init ssinad/p2p-orienteering-app -D 50 < 3/tsiligirides_problem_3.txt 2>/dev/null;,docker run -i --init ssinad/rooted-orienteering-app -D 50 < 3/3.txt 2>/dev/null;
55,460,550,54.88,641.18,1.17,490,52.7425,639.928,1.122,1.306,1.391,docker run -i --init ssinad/p2p-orienteering-app -D 55 < 3/tsiligirides_problem_3.txt 2>/dev/null;,docker run -i --init ssinad/rooted-orienteering-app -D 55 < 3/3.txt 2>/dev/null;
60,500,580,59.31,678.48,1.17,550,59.8608,675.7,1.055,1.229,1.351,docker run -i --init ssinad/p2p-orienteering-app -D 60 < 3/tsiligirides_problem_3.txt 2>/dev/null;,docker run -i --init ssinad/rooted-orienteering-app -D 60 < 3/3.txt 2>/dev/null;
65,530,610,64.31,715.77,1.17,550,60.0252,710.258,1.109,1.291,1.34,docker run -i --init ssinad/p2p-orienteering-app -D 65 < 3/tsiligirides_problem_3.txt 2>/dev/null;,docker run -i --init ssinad/rooted-orienteering-app -D 65 < 3/3.txt 2>/dev/null;
70,560,650,69.62,751.84,1.16,580,68.8047,747.245,1.103,1.288,1.334,docker run -i --init ssinad/p2p-orienteering-app -D 70 < 3/tsiligirides_problem_3.txt 2>/dev/null;,docker run -i --init ssinad/rooted-orienteering-app -D 70 < 3/3.txt 2>/dev/null;
75,600,680,74.76,781.4,1.15,610,73.4728,778.495,1.098,1.276,1.297,docker run -i --init ssinad/p2p-orienteering-app -D 75 < 3/tsiligirides_problem_3.txt 2>/dev/null;,docker run -i --init ssinad/rooted-orienteering-app -D 75 < 3/3.txt 2>/dev/null;
80,640,600,76.89,799.66,1.33,650,77.8302,798.285,1.092,1.228,1.247,docker run -i --init ssinad/p2p-orienteering-app -D 80 < 3/tsiligirides_problem_3.txt 2>/dev/null;,docker run -i --init ssinad/rooted-orienteering-app -D 80 < 3/3.txt 2>/dev/null;
85,670,650,81.58,800,1.23,650,77.8302,800,1.138,1.231,1.194,docker run -i --init ssinad/p2p-orienteering-app -D 85 < 3/tsiligirides_problem_3.txt 2>/dev/null;,docker run -i --init ssinad/rooted-orienteering-app -D 85 < 3/3.txt 2>/dev/null;
90,700,650,81.58,800,1.23,680,88.5741,800,1.132,1.176,1.143,docker run -i --init ssinad/p2p-orienteering-app -D 90 < 3/tsiligirides_problem_3.txt 2>/dev/null;,docker run -i --init ssinad/rooted-orienteering-app -D 90 < 3/3.txt 2>/dev/null;
95,740,650,81.58,800,1.23,620,94.5403,800,1.274,1.29,1.081,docker run -i --init ssinad/p2p-orienteering-app -D 95 < 3/tsiligirides_problem_3.txt 2>/dev/null;,docker run -i --init ssinad/rooted-orienteering-app -D 95 < 3/3.txt 2>/dev/null;
100,770,670,99.52,800,1.19,650,98.0959,800,1.231,1.231,1.039,docker run -i --init ssinad/p2p-orienteering-app -D 100 < 3/tsiligirides_problem_3.txt 2>/dev/null;,docker run -i --init ssinad/rooted-orienteering-app -D 100 < 3/3.txt 2>/dev/null;
105,790,710,104.22,800,1.13,660,102.552,800,1.212,1.212,1.013,docker run -i --init ssinad/p2p-orienteering-app -D 105 < 3/tsiligirides_problem_3.txt 2>/dev/null;,docker run -i --init ssinad/rooted-orienteering-app -D 105 < 3/3.txt 2>/dev/null;
110,800,750,109.5,800,1.07,670,108.703,800,1.194,1.194,1,docker run -i --init ssinad/p2p-orienteering-app -D 110 < 3/tsiligirides_problem_3.txt 2>/dev/null;,docker run -i --init ssinad/rooted-orienteering-app -D 110 < 3/3.txt 2>/dev/null;
\end{filecontents*}
\csvreader[head to column names]{csv/p2p_33.csv}{}
    { 
      \csviffirstrow{\\ \hline \multirow{20}{*}{\begin{tabular}{c} 33 nodes\end{tabular}}
      & \csvcoli & \csvcolvii & \csvcolii & \csvcolix & \csvcolx & \csvcolxii &
      \csvcolxi}
      {\\ \cline{2-8} & \csvcoli & \csvcolvii & \csvcolii & \csvcolix & \csvcolx & \csvcolxii &
      \csvcolxi}
    }
    \end{tabular}
}
\end{center}
\end{table}

\begin{table}[t!]
\begin{center}
{\small
    \caption{P2P orienteering results for the instances in~\cite{Chao96} with varying
      distance bounds. \label{ptp-chao}}  

    \medskip\medskip
    \begin{tabular}{l|l|l|l|l|l|l}%
    \bfseries Dataset & {\boldmath $B$} & {\boldmath $\val$} & {\boldmath $\mathsf{Best}\text{-}\val$} &
    {\boldmath $\ub$} & {\boldmath $\ub/\val$} & {\boldmath $\ub/\mathsf{Best}\text{-}\val$} 
    \\\hline \hline \vspace{-4.325mm}
    \begin{filecontents*}{csv/p2p_64.csv}
Distance Limit,Best overall from previous work,Chao Method,RootedBisectionOrienteering,Distance,Upper bound,Upper bound / reward,P2P Reward,P2P Distance,P2P Upper bound,Upper bound / reward,UB/best,,
15,96,96,270,14.1421,288.198,1.067,96,14.8284,139.089,1.449,1.449,docker run -i --init ssinad/p2p-orienteering-app -D 15 < set_64.txt 2>/dev/null;,docker run -i --init ssinad/rooted-orienteering-app -D 15 < 64.txt 2>/dev/null;
20,294,294,372,19.799,393.775,1.059,294,19.799,366.823,1.248,1.248,docker run -i --init ssinad/p2p-orienteering-app -D 20 < set_64.txt 2>/dev/null;,docker run -i --init ssinad/rooted-orienteering-app -D 20 < 64.txt 2>/dev/null;
25,390,390,444,24.6274,498.327,1.122,390,24.9756,499.509,1.281,1.281,docker run -i --init ssinad/p2p-orienteering-app -D 25 < set_64.txt 2>/dev/null;,docker run -i --init ssinad/rooted-orienteering-app -D 25 < 64.txt 2>/dev/null;
30,474,474,558,28.8701,602.878,1.08,474,29.9638,602.878,1.272,1.272,docker run -i --init ssinad/p2p-orienteering-app -D 30 < set_64.txt 2>/dev/null;,docker run -i --init ssinad/rooted-orienteering-app -D 30 < 64.txt 2>/dev/null;
35,576,570,654,33.1127,707.428,1.082,570,33.9411,707.428,1.241,1.228,docker run -i --init ssinad/p2p-orienteering-app -D 35 < set_64.txt 2>/dev/null;,docker run -i --init ssinad/rooted-orienteering-app -D 35 < 64.txt 2>/dev/null;
40,714,714,714,38.5269,804.823,1.127,714,38.7696,810.633,1.135,1.135,docker run -i --init ssinad/p2p-orienteering-app -D 40 < set_64.txt 2>/dev/null;,docker run -i --init ssinad/rooted-orienteering-app -D 40 < 64.txt 2>/dev/null;
45,816,816,798,44.7696,889.675,1.115,816,43.346,911.482,1.117,1.117,docker run -i --init ssinad/p2p-orienteering-app -D 45 < set_64.txt 2>/dev/null;,docker run -i --init ssinad/rooted-orienteering-app -D 45 < 64.txt 2>/dev/null;
50,900,900,840,48.9898,974.528,1.16,900,47.1522,974.528,1.083,1.083,docker run -i --init ssinad/p2p-orienteering-app -D 50 < set_64.txt 2>/dev/null;,docker run -i --init ssinad/rooted-orienteering-app -D 50 < 64.txt 2>/dev/null;
55,984,984,888,53.5887,1042.04,1.173,984,54.8985,1044.97,1.062,1.062,docker run -i --init ssinad/p2p-orienteering-app -D 55 < set_64.txt 2>/dev/null;,docker run -i --init ssinad/rooted-orienteering-app -D 55 < 64.txt 2>/dev/null;
60,1062,1044,990,59.8313,1105.68,1.117,1044,58.6466,1112.21,1.065,1.047,docker run -i --init ssinad/p2p-orienteering-app -D 60 < set_64.txt 2>/dev/null;,docker run -i --init ssinad/rooted-orienteering-app -D 60 < 64.txt 2>/dev/null;
65,1116,1116,1026,64.999,1169.32,1.14,1116,63.4882,1169.31,1.048,1.048,docker run -i --init ssinad/p2p-orienteering-app -D 65 < set_64.txt 2>/dev/null;,docker run -i --init ssinad/rooted-orienteering-app -D 65 < 64.txt 2>/dev/null;
70,1188,1176,1080,69.5887,1223.97,1.133,1176,69.8552,1227.36,1.044,1.033,docker run -i --init ssinad/p2p-orienteering-app -D 70 < set_64.txt 2>/dev/null;,docker run -i --init ssinad/rooted-orienteering-app -D 70 < 64.txt 2>/dev/null;
75,1236,1224,1128,74.7696,1266.4,1.123,1224,70.88,1277.78,1.044,1.034,docker run -i --init ssinad/p2p-orienteering-app -D 75 < set_64.txt 2>/dev/null;,docker run -i --init ssinad/rooted-orienteering-app -D 75 < 64.txt 2>/dev/null;
80,1284,1272,1176,79.598,1308.82,1.113,1272,79.3784,1308.82,1.029,1.019,docker run -i --init ssinad/p2p-orienteering-app -D 80 < set_64.txt 2>/dev/null;,docker run -i --init ssinad/rooted-orienteering-app -D 80 < 64.txt 2>/dev/null;
\end{filecontents*}
\csvreader[head to column names]{csv/p2p_64.csv}{}
    { 
      \csviffirstrow{\\ \hline \multirow{14}{*}{\begin{tabular}{c} 64 nodes\end{tabular}}
      & \csvcoli & \csvcolviii & \csvcolii & \csvcolx & \csvcolxi & \csvcolxii}
      {\\ \cline{2-7} & \csvcoli & \csvcolviii & \csvcolii & \csvcolx & \csvcolxi & \csvcolxii}
    } 
    \begin{filecontents*}{csv/p2p_66.csv}
Distance Limit,Best Solution from Previous Work,Chao method,P2P Reward,P2P Distance,P2P Upper bound,Upper bound / reward,UB / best,,
5,10,10,10,4.92081,17.6658,1.767,1.767,docker run -i --init ssinad/p2p-orienteering-app -D 5 < set_66.txt 2>/dev/null;,docker run -i --init ssinad/rooted-orienteering-app -D 5 < 66.txt 2>/dev/null;
10,40,40,40,8.92081,71.6148,1.79,1.79,docker run -i --init ssinad/p2p-orienteering-app -D 10 < set_66.txt 2>/dev/null;,docker run -i --init ssinad/rooted-orienteering-app -D 10 < 66.txt 2>/dev/null;
15,120,120,105,14.9465,181.896,1.732,1.516,docker run -i --init ssinad/p2p-orienteering-app -D 15 < set_66.txt 2>/dev/null;,docker run -i --init ssinad/rooted-orienteering-app -D 15 < 66.txt 2>/dev/null;
20,205,195,175,18.7082,301.747,1.724,1.472,docker run -i --init ssinad/p2p-orienteering-app -D 20 < set_66.txt 2>/dev/null;,docker run -i --init ssinad/rooted-orienteering-app -D 20 < 66.txt 2>/dev/null;
25,290,290,250,24.5188,395.022,1.58,1.362,docker run -i --init ssinad/p2p-orienteering-app -D 25 < set_66.txt 2>/dev/null;,docker run -i --init ssinad/rooted-orienteering-app -D 25 < 66.txt 2>/dev/null;
30,400,400,355,29.7104,482.882,1.36,1.207,docker run -i --init ssinad/p2p-orienteering-app -D 30 < set_66.txt 2>/dev/null;,docker run -i --init ssinad/rooted-orienteering-app -D 30 < 66.txt 2>/dev/null;
35,465,460,410,34.9314,570.381,1.391,1.227,docker run -i --init ssinad/p2p-orienteering-app -D 35 < set_66.txt 2>/dev/null;,docker run -i --init ssinad/rooted-orienteering-app -D 35 < 66.txt 2>/dev/null;
40,575,575,530,39.2691,657.881,1.241,1.144,docker run -i --init ssinad/p2p-orienteering-app -D 40 < set_66.txt 2>/dev/null;,docker run -i --init ssinad/rooted-orienteering-app -D 40 < 66.txt 2>/dev/null;
45,650,650,600,43.8515,745.381,1.242,1.147,docker run -i --init ssinad/p2p-orienteering-app -D 45 < set_66.txt 2>/dev/null;,docker run -i --init ssinad/rooted-orienteering-app -D 45 < 66.txt 2>/dev/null;
50,730,730,670,48.8932,832.88,1.243,1.141,docker run -i --init ssinad/p2p-orienteering-app -D 50 < set_66.txt 2>/dev/null;,docker run -i --init ssinad/rooted-orienteering-app -D 50 < 66.txt 2>/dev/null;
55,825,825,775,53.7104,920.38,1.188,1.116,docker run -i --init ssinad/p2p-orienteering-app -D 55 < set_66.txt 2>/dev/null;,docker run -i --init ssinad/rooted-orienteering-app -D 55 < 66.txt 2>/dev/null;
60,915,915,785,58.5188,1007.88,1.284,1.102,docker run -i --init ssinad/p2p-orienteering-app -D 60 < set_66.txt 2>/dev/null;,docker run -i --init ssinad/rooted-orienteering-app -D 60 < 66.txt 2>/dev/null;
65,980,980,815,64.5388,1087.54,1.334,1.11,docker run -i --init ssinad/p2p-orienteering-app -D 65 < set_66.txt 2>/dev/null;,docker run -i --init ssinad/rooted-orienteering-app -D 65 < 66.txt 2>/dev/null;
70,1070,1070,885,69.3472,1157.58,1.308,1.082,docker run -i --init ssinad/p2p-orienteering-app -D 70 < set_66.txt 2>/dev/null;,docker run -i --init ssinad/rooted-orienteering-app -D 70 < 66.txt 2>/dev/null;
75,1140,1140,985,74.6799,1198.84,1.217,1.052,docker run -i --init ssinad/p2p-orienteering-app -D 75 < set_66.txt 2>/dev/null;,docker run -i --init ssinad/rooted-orienteering-app -D 75 < 66.txt 2>/dev/null;
80,1215,1215,985,76.2678,1261.34,1.281,1.038,docker run -i --init ssinad/p2p-orienteering-app -D 80 < set_66.txt 2>/dev/null;,docker run -i --init ssinad/rooted-orienteering-app -D 80 < 66.txt 2>/dev/null;
85,1270,1270,1020,82.3872,1323.84,1.298,1.042,docker run -i --init ssinad/p2p-orienteering-app -D 85 < set_66.txt 2>/dev/null;,docker run -i --init ssinad/rooted-orienteering-app -D 85 < 66.txt 2>/dev/null;
90,1340,1340,1045,89.7465,1386.34,1.327,1.035,docker run -i --init ssinad/p2p-orienteering-app -D 90 < set_66.txt 2>/dev/null;,docker run -i --init ssinad/rooted-orienteering-app -D 90 < 66.txt 2>/dev/null;
95,1395,1380,1090,94.2021,1448.84,1.329,1.039,docker run -i --init ssinad/p2p-orienteering-app -D 95 < set_66.txt 2>/dev/null;,docker run -i --init ssinad/rooted-orienteering-app -D 95 < 66.txt 2>/dev/null;
100,1465,1435,1195,99.3942,1510.38,1.264,1.031,docker run -i --init ssinad/p2p-orienteering-app -D 100 < set_66.txt 2>/dev/null;,docker run -i --init ssinad/rooted-orienteering-app -D 100 < 66.txt 2>/dev/null;
105,1520,1510,1245,104.193,1569.17,1.26,1.032,docker run -i --init ssinad/p2p-orienteering-app -D 105 < set_66.txt 2>/dev/null;,docker run -i --init ssinad/rooted-orienteering-app -D 105 < 66.txt 2>/dev/null;
110,1560,1550,1275,109.476,1581.61,1.24,1.014,docker run -i --init ssinad/p2p-orienteering-app -D 110 < set_66.txt 2>/dev/null;,docker run -i --init ssinad/rooted-orienteering-app -D 110 < 66.txt 2>/dev/null;
115,1595,1595,1335,114.54,1619.11,1.213,1.015,docker run -i --init ssinad/p2p-orienteering-app -D 115 < set_66.txt 2>/dev/null;,docker run -i --init ssinad/rooted-orienteering-app -D 115 < 66.txt 2>/dev/null;
120,1635,1635,1360,115.741,1656.61,1.218,1.013,docker run -i --init ssinad/p2p-orienteering-app -D 120 < set_66.txt 2>/dev/null;,docker run -i --init ssinad/rooted-orienteering-app -D 120 < 66.txt 2>/dev/null;
125,1670,1655,1435,124.833,1678.13,1.169,1.005,docker run -i --init ssinad/p2p-orienteering-app -D 125 < set_66.txt 2>/dev/null;,docker run -i --init ssinad/rooted-orienteering-app -D 125 < 66.txt 2>/dev/null;
130,1680,1680,1465,128.86,1680,1.147,1,docker run -i --init ssinad/p2p-orienteering-app -D 130 < set_66.txt 2>/dev/null;,docker run -i --init ssinad/rooted-orienteering-app -D 130 < 66.txt 2>/dev/null;
\end{filecontents*}
\csvreader[head to column names]{csv/p2p_66.csv}{}
    { 
      \csviffirstrow{\\ \hline \multirow{26}{*}{\begin{tabular}{c} 66 nodes\end{tabular}}
      & \csvcoli & \csvcoliv & \csvcolii & \csvcolvi & \csvcolvii & \csvcolviii}
      {\\ \cline{2-7} & \csvcoli & \csvcoliv & \csvcolii & \csvcolvi & \csvcolvii & \csvcolviii}
    } 
    \end{tabular}
}
\end{center}
\end{table}

\end{document}